\documentclass[12pt]{article} 
\usepackage[letterpaper, margin=1in]{geometry}  % Set 1-inch margins
\usepackage{setspace}  % For double spacing
\usepackage{lmodern}  % Improved font rendering
\usepackage[hypertexnames=false]{hyperref}
\usepackage[numbers,authoryear]{natbib}  % For author-year citations
\usepackage{multibib}
\newcites{SM}{Appendix References}
\usepackage{amsmath, amsfonts, amssymb, mathrsfs}
\usepackage{array}
\usepackage{booktabs}
\usepackage{etoc}
\usepackage{graphicx}
\usepackage{multicol, multirow}
\usepackage{pdflscape}
\usepackage{amsthm}
\usepackage{hyperref}
\usepackage{graphicx}
\usepackage{algorithm}
\usepackage{upgreek}
\usepackage{authblk}
\usepackage{enumitem}
\usepackage{etoc}
\usepackage{algpseudocode}  
\hypersetup{colorlinks=true, linkcolor=blue, citecolor=blue}

\newtheorem{definition}{Definition}

\newtheorem{lemma}{Lemma}
\newtheorem{prop}{Proposition}
\newtheorem{assumption}{Assumption}

\newtheorem{theorem}{Theorem}
\newtheorem{corollary}{Corollary}
\newtheorem{remark}{Remark}

\DeclareMathOperator*{\argmin}{arg\,min}
\newcommand{\R}{\mathbb{R}}  % Real numbers
\newcommand{\E}{\mathbb{E}}  % Expectation
\newcommand{\Var}{\text{Var}} % Variance
\newcommand{\Prob}{\mathbb{P}} % Probability
\newcommand{\indep}{\perp\!\!\!\perp}  % Independence symbol
\newcommand{\CS}{\mathcal{S}_{Y \mid {X}}}  % Independence symbol
\newcommand{\CMS}{\mathcal{S}_{\E(Y \mid {X})}}  % Independence symbol
\newcommand{\CCMS}{\mathcal{S}_{\E(Y^{x})}}  % Independence symbol
\newcommand{\beq}{\begin{equation}}
\newcommand{\eeq}{\end{equation}}
\newcommand{\bal}{\begin{aligned}}
\newcommand{\eal}{\end{aligned}}

\newcommand{\po}[1]{Y^{#1}}               % \po{x} -> Y^x
\newcommand{\calA}{\mathfrak{A}}
\newcommand{\Apo}[1]{\calA_{\po{#1}}}            % \Apo{x} -> \mathfrak{A}_{Y^x}
\usepackage{mathtools}

\title{\textbf{Causal Sufficient Dimension Reduction for Multiple Continuous Exposures with an Application to Environmental Mixtures}}

\author[]{Thomas W. Hsiao\thanks{Corresponding author: thsiao3@emory.edu}}
\author[]{Howard H. Chang}
\author[]{Razieh Nabi}

\affil[]{Department of Biostatistics and Bioinformatics, Emory University}

\date{\today}

\begin{document}
\begin{singlespace}
\maketitle

\begin{abstract}
Estimating causal effects with multivariate continuous exposures is challenging because causal exposure-response surfaces can be high-dimensional, complicating estimation and interpretation of joint exposure effects. Such settings arise in environmental epidemiology, where interest centers on the health effects of chemical and pollutant mixtures. We develop \textit{causal sufficient dimension reduction} (CSDR), a semiparametric framework for representing causal exposure-response surfaces through low-dimensional exposure summaries. We formalize the reduction target as the \textit{causal central mean subspace} and propose a modular two-stage estimator that decouples nuisance-function estimation from subspace estimation, simplifying implementation relative to existing marginal structural model-based approaches. The reduced exposure preserves the information needed to characterize joint causal effects while enabling efficient downstream estimation. We establish a convergence rate for causal subspace recovery accounting for first-stage nuisance estimation error, show that the structural dimension can be estimated consistently, and introduce a \textit{subspace importance score} that quantifies the contribution of each exposure to the reduction. In simulations, CSDR yielded more accurate estimation and uncertainty quantification of the exposure-response surface than methods using noncausal dimension reduction or the original exposure. We apply CSDR to study the effect of maternal exposure to PFAS chemical mixtures on infant birth weight in the Atlanta African American Maternal-Child Cohort.
\end{abstract}

\vspace{1em}
\end{singlespace}
\newpage
% \setcounter{page}{1}
% \pagenumbering{arabic}
%%%%%%%%%%%%%%%%%%%%%%%%%%%%%%%%%%%%%%%%%%%%%%%%%%%%%%%%%%%%%%%%%%%%%%%%%%%%%%%%%%%%%%

\newpage

%%% Introduction
\section{Introduction}\label{sec-intro}

Causal inference with multivariate continuous treatments arises in many scientific applications, including environmental epidemiology where populations are simultaneously exposed to multiple, highly correlated exposures \citep{carlinUnravelingHealthEffects2013}. In these settings, investigators are often interested in estimating the causal exposure-response surface (ERS), $\mu(x) \coloneqq \E(Y^x)$, where $Y^x$ denotes the potential outcome under exposure level $x\in\R^p$. The ERS characterizes how the expected potential outcome varies under interventions on the exposure vector. Estimating the ERS is challenging when $p$ is moderate or large: flexible nonparametric regression and density estimation suffer from the curse of dimensionality, and the resulting high-dimensional surface can be difficult to interpret.

A promising way to address these challenges is to identify low-dimensional structure in the causal ERS. In environmental mixtures applications, such structure is often scientifically plausible because many exposures arise from shared environmental sources and exposure pathways. We refer to this problem as \textit{causal sufficient dimension reduction} (CSDR). The goal is to identify a low-dimensional linear summary of the exposure vector $X$ that preserves its causal relationship with $Y$. Specifically, suppose there exists a matrix $\beta\in\R^{p\times d}$ with $1\leq d<p$ and an unknown smooth function $g$ such that $\mu(x) = g(\beta^\top x)$. Because $\beta$ is identifiable only up to a change of basis, the statistical estimand is the subspace spanned by the columns of $\beta$. We refer to the minimal subspace satisfying this representation as the \textit{causal central mean subspace} (CCMS). Under this framework, the causal ERS depends on the original exposure $X$ only through the lower-dimensional summary $Z \coloneqq \beta^\top X\in \R^d$, allowing $X$ to be replaced by $Z$ for subsequent estimation and interpretation. The reduction can improve the efficiency of downstream ERS estimation while yielding interpretable low-dimensional summaries that preserve the causal ERS.

The proposed causal sufficient dimension reduction is distinct from both conventional exposure mixtures analysis \citep{joubertPoweringResearchInnovative2022} and classical sufficient dimension reduction (SDR) \citep{liSufficientDimensionReduction2018, cookPrincipalComponentsSufficient2018}. Existing mixtures methods span a spectrum of modeling assumptions. Flexible approaches such as Bayesian kernel machine regression (BKMR; \citealp{bobbBayesianKernelMachine2015}) and Bayesian additive regression trees (BART; \citealp{englertModelingJointHealth2025}) estimate complex multivariate exposure-response relationships. However, these methods may have limited power, be computationally expensive, and do not directly learn low-dimensional exposure summaries. At the other extreme, methods such as  weighted quantile sum regression \citep{carricoCharacterizationWeightedQuantile2015} and quantile $g$-computation \citep{keilQuantileBasedGComputationApproach2020} impose a one-dimensional exposure index that may be overly restrictive. Unsupervised approaches such as principal components pursuit \citep{gibsonPrincipalComponentPursuit2022} strike a compromise by reducing the dimension of the exposure prior to estimation, but ignore the outcome entirely and therefore need not recover causally relevant exposure summaries. More fundamentally, these approaches are generally formulated in terms of observed outcome regressions rather than the causal exposure-response surface itself.

Supervised dimension reduction methods are closer in spirit to the CSDR objective. In classical SDR, a common target is the central mean subspace, the minimal subspace that preserves the conditional mean of $Y$ given $X$. However, in observational studies with confounding, the conditional mean $\E(Y \, | \, X = x)$ can differ substantially from the causal mean $\E(Y^x)$. Consequently, classical SDR procedures na{\"i}vely applied to $(Y, X)$ may recover low-dimensional summaries that fail to preserve the causal exposure-response relationship. 

Existing work has begun to connect SDR with causal inference. Much of this literature focuses on reducing the dimension of covariates for estimating average treatment effects \citep{maRobustEfficientApproach2019, luoMatchingUsingSufficient2020,chengSufficientDimensionReduction2022}, or individualized treatment regimes \citep{sonDimensionreducedOutcomeweightedLearning2026}. In contrast, our goal is to reduce the dimension of a continuous exposure vector while preserving the causal ERS. Our work builds on the foundation established by \cite{nabiSemiparametricCausalSufficient2022}, who developed the first causal SDR estimator for multivariate continuous exposures by embedding semiparametric estimating equations for the central mean subspace \citep{maSemiparametricApproachDimension2012} within a marginal structural model framework \citep{robinsMarginalStructuralModels2000a}. While theoretically appealing, implementation requires estimating multiple nuisance functions that depend on the unknown reduction together with computationally intensive score-based optimization procedures.

In this paper, we propose a modular framework for estimating the CCMS that substantially reduces the nuisance estimation burden of the original marginal structural model approach. Rather than deriving estimating equations indexed by the unknown reduction $\beta$, we decouple nuisance-function estimation from subspace estimation through a two-stage procedure. In the first stage, we construct transformed variables $(\tilde{Y}, \tilde{X})$ whose central mean subspace coincides with the target causal central mean subspace. We show that these transformations can be constructed using nuisance-function estimators developed for causal inference with continuous treatments. In the second stage, we apply existing SDR methods to estimate the central mean subspace of these transformed variables. The two-stage formulation allows existing tools for causal inference and sufficient dimension reduction to be combined in a transparent and computationally efficient manner. 

Our work also relates to recent developments in mixtures analysis and causal inference. In the mixtures literature, \cite{mcgeeBayesianMultipleIndex2023} proposed a Bayesian multiple index model as a supervised dimension reduction alternative to BKMR, although the resulting indices do not generally admit a causal interpretation. More recently, \cite{shinTreatmentEffectHeterogeneity2025} introduced causal estimands for environmental mixtures using nonparametric Bayesian models for heterogeneous dose-response estimation, while \cite{kramerCausalInferenceHighDimensional2026} studied doubly robust estimation of the causal mean for high-dimensional continuous treatments under sparsity assumptions. Our framework is also closely connected to the growing literature on causal inference with continuous treatments. Kernel-based identification strategies have been used to extend inverse probability weighting and doubly robust estimation from discrete to continuous exposures \citep{kallusPolicyEvaluationOptimization2018, klosinAutomaticDoubleMachine2021, colangeloDoubleDebiasedMachine2026, zhangDoublyRobustInference2025}. Our approach differs from these methods by using a causal inference machinery to learn low-dimensional exposure summaries that preserve the causal exposure-response surface itself.

The main contributions of this paper are as follows: 
(i) We clarify the causal central mean subspace as the target estimand and develop a modular two-stage estimation framework that substantially simplifies implementation relative to existing marginal structural model approaches.
(ii) We establish theoretical guarantees for the proposed estimator, including consistency of both the estimated subspace and its structural dimension.
(iii) We introduce subspace importance scores for interpreting the contribution of the original exposures to the estimated causal subspace.
(iv) We evaluate the proposed approach in simulation studies and apply it to study the effect of maternal PFAS exposure mixtures on infant birth weight in the Atlanta African-American Maternal Child cohort. 

The remainder of the paper is organized as follows. Section~\ref{sec:prelim} introduces the causal sufficient dimension reduction framework and defines the causal central mean subspace. Section~\ref{sec:estimation} presents the proposed estimation strategy and theoretical results. Section~\ref{sec:importance_score} introduces subspace importance scores for interpretation of the estimated subspace. Section~\ref{sec:sims} contains simulation studies, and Section~\ref{sec:real_data} presents the PFAS application. We conclude with a discussion in Section~\ref{sec:discussion}. All proofs are deferred to the supplementary materials. 

\section{Causal sufficient dimension reduction framework}
\label{sec:prelim}

\subsection{Setup and identification}
\label{subsec:setup}

Let the observed data consist of $n$ independent and identically distributed observations $({O}_1,...,{O}_n)$, where $O=(Y, X, C)$ with support $\mathcal{O}=\mathcal{Y}\times \mathcal{X} \times \mathcal{C}$. Here, $Y\in\mathbb{R}$ denotes the outcome of interest, $X\in\mathbb{R}^p$ the continuous multidimensional exposure, and $C\in\mathbb{R}^q$ a vector of covariates. We define causal effects using the potential outcomes framework, letting $Y^{x}$ denote the potential outcome when $X=x$. Our primary estimand is the causal mean, or exposure-response surface (ERS), defined as the map given by $\mu: \mathbb{R}^p \to \mathbb{R}$  

\vspace{-1.2cm}
\begin{align*}
\mu(x) \coloneqq \E(Y^x), \qquad x\in\mathcal{X},
\end{align*}

\vspace{-0.5cm} \noindent
which characterizes the expected potential outcome under exposure level $x$. 

To reduce the dimension of the exposure while preserving the ERS, we consider a low-dimensional representation $Z \coloneqq \beta^\top X$, where $\beta\in\mathbb{R}^{p\times d}$ is full rank with $1\leq d < p$. We denote by $Y^z$ the potential outcome under the reduced exposure level $Z=z$, and define the corresponding reduced ERS as
\[
\mu_\beta(z) \coloneqq \E(Y^z),
\]
where the subscript $\beta$ emphasizes that the reduced ERS depends on the choice of dimension reduction subspace. The support of $Z$ is denoted by $\mathcal{Z} \coloneqq \beta^\top\mathcal{X}$. 

Let $\Prob$ denote the distribution of $O$. For the original exposure $X$, define the \textit{outcome regression} as $m(x,c)\coloneqq\E(Y\, | \, X=x, C=c)$ and the \textit{generalized propensity score} (GPS) as $\pi(x\, | \, c)=\partial \Prob(X\leq x\, | \, C=c) / \partial x$. For the reduced exposure $Z$, define analogously $m_\beta(z,c)\coloneqq\E(Y\, | \, Z=z, C=c)$ and $\pi_\beta(z\, | \, c)\coloneqq \partial \Prob(Z \leq z \, | \, C = c) / \partial z$. 

We assume the following standard identification conditions for the  effect of $X$ on $Y$. 

\begin{assumption}[Identification conditions for the effect of $X$]\label{ass:identificationX}
Assume: 
\begin{enumerate}[label=(\alph*)]
\item {Consistency:}  
If $X=x$ then $Y=Y^{x}$, for all $x\in\mathcal{X}$.

\item {Strong positivity:}  
For $\varepsilon>0$, $\pi(x \mid c) \geq \varepsilon$ for all $(x,c)\in \mathcal{X}\times\mathcal{C}$.

\item {Conditional ignorability:}  
$Y^{x} \indep X \mid C$ for all $x\in\mathcal{X}$.
\end{enumerate}
\end{assumption}

Although weaker positivity conditions suffice for identification (i.e., $\pi(x \, | \, c) > 0$ for all $(x,c)\in\mathcal{X}\times\mathcal{C}$), strong positivity is commonly imposed for estimation and inference with continuous exposures. In practice, estimation of the ERS should therefore be restricted to regions of the exposure space that are adequately represented in the observed data \citep{antonelliCausalAnalysisAir2024}.

Under Assumption~\ref{ass:identificationX}, the ERS is identified by the $g$-formula \citep{robinsNewApproachCausal1986}:

\vspace{-1.2cm}
\begin{align}
\mu(x)=\E[m(x,C)].
\label{eq:identification_gcomp}
\end{align}

\vspace{-0.5cm}
For continuous exposures, however, the ERS is a nonregular parameter and therefore does not admit a pathwise differentiable influence function or estimation at the usual $\sqrt{n}$-rate \citep{bickelEfficientAdaptiveEstimation1998}. As a result, modern ERS estimation typically relies on kernel localization, which requires additional smoothness assumptions \citep{colangeloDoubleDebiasedMachine2026}.

\begin{assumption}[Smoothness conditions]
For any $(y,x,c)\in\mathcal{Y} \times \mathcal{X} \times \mathcal{C}$, assume:  
\begin{enumerate}[label=(\alph*)]
\item The outcome regression $m(x,c)$ and generalized propensity score $\pi(x \mid c)$ are three times continuously differentiable with respect to $x$, with bounded derivatives over $\mathcal{X}\times\mathcal{C}$. 

\item $\Var(Y\mid X=x, C=c)$ and its derivatives w.r.t. $x$ are uniformly bounded over $\mathcal{X}\times \mathcal{C}$;  

\item Define the product kernel $K_h(X_i-x) = \prod_{j=1}^p k\!\left(\frac{X_{ij}-x_j}{h}\right)$, where $k$ is a second-order symmetric kernel satisfying 
$\int k(u)\,du = 1$, 
$\int uk(u)\,du = 0$, and 
$\int u^2k(u)\,du < \infty$. 
Assume further that $k$ is bounded and continuously differentiable, and that for some positive constants $c$, $\bar u$, and $\nu>1$, $|dk(u)/du| \leq c|u|^{-v}$ for $|u| > \bar{u}$. 
\end{enumerate}
\label{ass:smooth}
\end{assumption}

Under Assumptions~\ref{ass:identificationX} and \ref{ass:smooth}, the ERS also admits the following inverse probability of treatment weighted (IPTW) representation: 
\begin{align}
\mu(x) 
&= \lim_{h\to 0}\E\!\Bigg[\frac{1}{h^p}\,
\frac{K_h\!\left({X-x}{}\right)}{\pi(X \mid C)}\, Y \Bigg].
\label{eq:identification_ipw}
\end{align}

These identification results motivate several standard estimators of the ERS. Regression adjustment (RA), or the $g$-formula estimator, estimates $\mu(x)$ by averaging an estimated outcome regression:

\vspace{-1.6cm}
\begin{align}\label{eq:RA}
    \hat{\mu}_{RA}(x) = \frac{1}{n} \sum_{i=1}^n \widehat{m}(x, C_i),
\end{align}

\vspace{-0.15cm} \noindent
where $\widehat{m}$ denotes an estimator of the outcome regression $m$. 

The IPTW estimator instead weights observations using the inverse GPS: 

\vspace{-1.1cm}
\begin{align}
    \hat{\mu}_{IPTW}(x) = \frac{1}{n \, h^p} \sum_{i=1}^n \frac{K_h\left({X_i - x}\right)}{\widehat{\pi}(X_i \mid C_i)}\, Y_i, 
    \label{eq:IPW}
\end{align} 

\vspace{-0.25cm} \noindent 
where $\hat{\pi}$ denotes an estimator of the GPS $\pi$. 

A doubly robust (DR) estimator combines both the RA and IPTW estimators as:

\vspace{-1.cm}
\begin{align}
    \hat{\mu}_{DR}(x) = \frac{1}{n \, h^p} \sum_{i=1}^n \bigg\{\frac{K_h\left({X_i - x}\right)}{\widehat{\pi}(X_i \mid C_i)}\, [Y_i - \widehat{m}(x, C_i)] + h^p \,  \widehat{m}(x, C_i)\bigg\}.
    \label{eq:DR}
\end{align}

\vspace{-0.15cm}
These estimators provide principled approaches for estimating the ERS under multivariate continuous exposures. However, estimation becomes increasingly difficult as the exposure dimension $p$ grows due to the curse of dimensionality inherent in nonparametric regression and density estimation. This motivates replacing the original exposure vector with a lower-dimensional representation that still preserves the causal ERS. We next review classical sufficient dimension reduction to set up the concepts underlying the causal framework developed below. 

\subsection{From classical SDR to causal SDR}

Sufficient dimension reduction (SDR) seeks a low-dimensional linear representation of a high-dimensional predictor vector ${X} \in \mathbb{R}^p$ that preserves the information in $X$ relevant to an outcome $Y$. A standard target in classical SDR is a reduction that preserves the conditional distribution of $Y$ given $X$. For a matrix $\beta \in \mathbb{R}^{p \times d}$ with $d < p$, the subspace $\text{span}(\beta)$ is a dimension reduction subspace if $Y \perp\!\!\!\perp X \mid \beta^\top X$, or equivalently,  

\vspace{-1.2cm}
\begin{align*}%\label{eq:CS}
    \Prob(Y\leq y\mid {X} = {x}) = \Prob(Y\leq y\mid \beta^\top {X} = \beta ^\top {x}).
\end{align*}

\vspace{-0.5cm} \noindent
The central subspace (CS), denoted $\CS$, is the intersection of all such dimension reduction subspaces. Because $\beta$ is identifiable only up to its  column space, the estimand is $\CS$, rather than $\beta$ itself. 

In many applications, interest centers on the conditional mean rather than the full conditional distribution. This motivates the central mean subspace (CMS), denoted $\CMS$, defined as the intersection of all subspaces $\text{span}(\beta)$ satisfying 

\vspace{-1.2cm}
\begin{align}\label{eq:CMS}
    \E(Y \mid X = x)
    =
    \E(Y \mid \beta^\top X = \beta^\top x).
\end{align} 

\vspace{-0.5cm} \noindent 
Under mild smoothness conditions, \eqref{eq:CMS} is equivalent to the multiple-index representation 

\vspace{-1.2cm}
\begin{align*}
    Y = g(\beta^\top X) + \varepsilon,
    \qquad
    \E(\varepsilon \mid X)=0,
\end{align*}

\vspace{-0.5cm} \noindent 
for some unknown smooth function $g$. 

A large literature studies estimation of $\CMS$. Semiparametric approaches use tools from semiparametric efficiency theory \citep{bickelEfficientAdaptiveEstimation1998, tsiatisSemiparametricTheoryMissing2006} to characterize regular asymptotically linear (RAL) estimators of $\beta$. Building on this framework, \cite{maSemiparametricApproachDimension2012} derived the class of influence functions for $\CMS$ estimators, while \cite{maEstimationEfficiencyCentral2014} and \cite{luoEfficientDimensionReduction2014} obtained the semiparametric efficient score
\begin{align}
    {U}(\beta)=\frac{1}{\sigma^2(X)}\left[ X - \frac{\E\{X / \sigma^2(X) \mid \beta^\top X\}}{\E\{1 / \sigma^2(X) \mid \beta^\top X\}} \right] \frac{\partial \E(Y\mid \beta^\top X)}{\partial (\beta^\top X)} \big[Y - \E(Y \mid \beta^\top {X})\big],
    \label{eq:S_eff}
\end{align}
where $\sigma^2(X) \coloneqq \Var(Y\mid X)$. These methods provide efficiency guarantees and support influence function based inference, but practical implementation can be challenging because evaluating \eqref{eq:S_eff} requires estimating several high-dimensional nuisance functions while repeatedly updating $\beta$ through numerical optimization. 

Gradient-based methods provide a computationally simpler alternative. A prominent example is the \textit{minimum average variance estimator} (MAVE) of \cite{xiaAdaptiveEstimationDimension2002}, which exploits the identity

\vspace{-1.2cm}
\begin{align*}
\CMS
= \text{span}\left\{\E\!\left[
\frac{\partial \E(Y \mid X = x)}{\partial x}
\frac{\partial \E(Y \mid X = x)}{\partial x^\top}
\right]\right\},
\end{align*} 

\vspace{-0.35cm} \noindent
provided the support of $Z=\beta^\top X$ is convex. Therefore, $\beta$ can be recovered through estimation of the gradient of the regression function. Under \eqref{eq:CMS}, MAVE jointly estimates both $\beta$ and the gradient of $\E(Y\mid X)$ through the local linear regression objective
\begin{equation}\label{eq:MAVE}
    \argmin_{\beta^\top\beta=I_d, {a}, {b}}  \sum_{j=1}^n \sum_{i=1}^n \big[Y_i - \{a_j + {b}_j^\top \beta^\top (X_i - X_j)\}\big]^2 w_{ij},
\end{equation}
where $a=(a_1,a_2,...,a_n)$, $b=({b}_1^\top,...,{b}^\top_n)^\top$, $w_{ij}$ are kernel weights based on the distance between $X_i$ and $X_j$, and $h$ is a bandwidth distinct from the bandwidth in \eqref{eq:DR}. MAVE alternates between least squares updates for $(a,b)$ and $\beta$, yielding a numerically stable procedure that avoids directly solving estimating equations derived from the efficient score. Although gradient-based methods may sacrifice some efficiency relative to semiparametric estimators, they are often substantially easier to implement in practice. 

Importantly, however, both semiparametric and gradient-based SDR methods target associational structure in the observed data distribution. In particular, they reduce the conditional mean regression surface $\E(Y\mid X=x)$ rather than the causal exposure-response surface $\E(Y^x)$. In the presence of confounding, these quantities can differ substantially, implying that classical SDR subspaces need not preserve causal effects. 

Our goal is instead to identify a low-dimensional representation of the exposure vector that preserves the causal exposure-response surface. This motivates a causal analogue of the CMS, which we formalize in the next subsection through the causal central mean subspace. 

\subsection{Causal central mean subspace}

We now formalize a causal analogue of the central mean subspace. Suppose there exists a matrix $\beta \in \mathbb{R}^{p\times d}$ such that

\vspace{-1.2cm}
\begin{align}\label{eq:cs_constraint}
    \mu(x) = g(\beta^\top x), \qquad \forall x \in \mathcal{X}.
\end{align} 

\vspace{-0.35cm} 
We refer to \eqref{eq:cs_constraint} as the causal sufficiency condition. The above states that the causal exposure-response surface depends on the exposure vector only through the lower-dimensional projection $\beta^\top x$. As in classical SDR, the basis matrix $\beta$ is not uniquely identified because $\beta$ and $\beta A$ generate the same subspace for any full-rank $d\times d$ matrix $A$. Therefore, we focus on the subspace $\text{span}(\beta)$ rather than any particular basis representation. 

Because each $\beta^\top x$ corresponds uniquely to the orthogonal projection of $x$ onto $\text{span}(\beta)$, condition \eqref{eq:cs_constraint} can equivalently be written as 

\vspace{-1.2cm}
\begin{align}
    \mu(x)=\mu(P_\beta x), \quad \text{where} \quad P_\beta = \beta(\beta^\top\beta)^{-1}\beta^\top. 
\end{align} 

\vspace{-0.5cm} \noindent 
Here $P_\beta$ denotes the orthogonal projection matrix onto $\text{span}(\beta)$. We now define the causal central mean subspace.

\begin{definition}[Causal central mean subspace]\label{def:cCMS}
    If there exists a $p\times d$ matrix $\beta$ satisfying \eqref{eq:cs_constraint}, then $\text{span}(\beta)$ is called a causal sufficient dimension reduction subspace for the causal mean. The intersection of all such subspaces, provided the intersection itself also satisfies \eqref{eq:cs_constraint}, is called the causal central mean subspace (CCMS), denoted $\mathcal{S}_{\E(Y^{x})}$. 
\end{definition}

The CCMS therefore defines the minimal linear subspace preserving the causal exposure-response surface. Unlike the CMS $\CMS$, which summarizes the observational regression function $\E(Y\mid X=x)$, the CCMS summarizes the counterfactual mean function $\E(Y^x)$.

As in classical SDR theory, existence of the CCMS is not automatic and requires excluding pathological cases where intersections of sufficient subspaces fail to remain sufficient. 

\begin{assumption}[Existence of CCMS] 
\label{ass:existence_CCMS}
Let $\mathfrak{A}$ denote the collection of all causal sufficient dimension reduction subspaces satisfying \eqref{eq:cs_constraint}. Then
    $\mathcal{S}_{\E(Y^{x})} \coloneqq \cap \{ \mathcal{S}: \mathcal{S} \in \mathfrak{A} \} \in\mathfrak{A}$.
\end{assumption}

To ensure that the reduced exposure $Z \coloneqq \beta^\top X$ admits a meaningful causal interpretation, we additionally require that interventions on $Z$ correspond to well-defined interventions on the original exposure vector $X$.

\begin{assumption}[Well-defined reduced exposure]
    For any $\mathcal{S} \in \mathfrak{A}$ and any matrix $\beta$ satisfying $\text{span}(\beta)=\mathcal{S}$, if $\beta^\top x=\beta^\top x'$ then $Y^x=Y^{x'}$ almost surely.
    \label{ass:well_defined_Z}
\end{assumption}

Note that Assumption~\ref{ass:well_defined_Z} is not necessary to define the CCMS, but ensures that the potential outcome indexed by the reduced exposure is uniquely defined. Without this condition, multiple exposure vectors could map to the same value of $Z$ while producing different counterfactual outcomes, leading to a version of the ``multiple versions of treatment'' problem discussed by \cite{vanderweeleCausalInferenceMultiple2013}. Under Assumption~\ref{ass:well_defined_Z}, we may define the reduced potential outcome $Y^z$ unambiguously and can use $Z$ in causal ERS estimation.

We also require positivity for the reduced exposure mechanism.

\begin{assumption}[Strong positivity w.r.t. $Z$]
    \label{ass:strong_positivity_Z}
    For any $\mathcal{S} \in \mathfrak{A}$ and any matrix $\beta$ satisfying $\text{span}(\beta)=\mathcal{S}$, there exists  $\varepsilon>0$, such that $\pi_\beta(z \mid {c}) \geq \varepsilon > 0$ for all $({z}, {c})\in \mathcal{Z} \times \mathcal{C}$.
\end{assumption}

Although weak positivity for $X$ transfers naturally to the reduced exposure $Z$, strong positivity need not. Assumption~\ref{ass:strong_positivity_Z} therefore guarantees that causal effects indexed by $Z$ remain identifiable over the reduced exposure space.

Under these assumptions, the reduced exposure inherits the same identification structure established for the original exposure vector. 

\begin{corollary}[Identification of the causal mean with respect to $Z$]
Under Assumptions~\ref{ass:identificationX}-\ref{ass:strong_positivity_Z}, there exists a unique potential outcome $Y^z$ for each $z \coloneqq \beta^\top x$, for all $x \in \mathcal{X}$. Moreover, the reduced exposure $Z$ satisfies consistency, strong positivity and conditional ignorability with respect to covariates $C$. Consequently,  $\mu_\beta(z)\coloneqq\E(Y^z)$ is identified by the g-formula and inverse probability weighting functionals in \eqref{eq:identification_gcomp} and \eqref{eq:identification_ipw}, replacing $m(x,c)$ and $\pi(x\mid c)$ with $m_\beta(z, c)$ and $\pi_\beta(z\mid c)$. 
\label{cor:identification_Z}
\end{corollary}

See a proof in Appendix~\ref{sec:proof_cor_identification_Z}. 

A potential concern is that the reduced exposure $Z$ may not itself admit a natural intervention interpretation. Our objective, however, is not to reinterpret $Z$ as a new treatment, but rather to identify a sufficient causal summary of the original exposure vector $X$. If the exposure-response surface depends on $X$ only through $\beta^\top X$, then the CCMS characterizes the causal directions along which variation in the exposure alters the outcome while discarding directions that are causally irrelevant.

Finally, $\CCMS$ inherits the affine equivariance property of both $\CS$ and $\CMS$. This property justifies standardizing exposures prior to estimation and will be useful in subsequent sections when selecting smoothing parameters and interpreting variable importance.

\begin{prop}[Affine equivariance of $\mathcal{S}_{\E(Y^{{x}})}$]
Suppose Assumptions~\ref{ass:identificationX}-\ref{ass:existence_CCMS} hold, let $A\in \mathbb{R}^{p\times p}$ be full rank, and let $b\in \mathbb{R}^p$. Define $Y^{Ax+b}$ as the potential outcome under intervention $AX+b=Ax+b$. Then $\mathcal{S}_{\E(Y^{x})}= A^\top \mathcal{S}_{\E(Y^{A{x}+b})}$.
\label{prop:affine}
\end{prop}

See a proof in Appendix~\ref{sec:proofs_prop:affine}. 

The results above establish that the CCMS is identified through the causal exposure-response surface under the aforementioned assumptions. We now turn to estimation. Existing semiparametric approaches to causal SDR extend efficient-score methods from association-based SDR, but can be computationally burdensome in practice. We therefore develop a modular two-stage framework that separates causal effect estimation from dimension reduction. 

\section{Modular estimation of the CCMS}
\label{sec:estimation}

For association-based SDR, \cite{luoEfficientDimensionReduction2014} proposed a general semiparametric framework for estimating $\beta$ when the target is a statistical functional of the conditional distribution $Y \mid X$. The CCMS falls outside this setting because identification of the causal exposure-response surface depends additionally on the distribution of covariates $C$. To address this challenge, 
\cite{nabiSemiparametricCausalSufficient2022} reformulated the causal sufficiency condition in \eqref{eq:cs_constraint} as a marginal structural model (MSM), allowing association-based estimating equations to be transformed into valid causal estimating equations using the framework of \cite{robinsMarginalStructuralModels2000a}. We briefly review this approach before introducing our modular two-stage estimation framework for causal SDR.

\subsection{Existing semiparametric approaches}

An MSM specifies a parametric model for the exposure-response surface, $\E(Y^x) = g(x;\beta)$, where $\beta$ is finite-dimensional and $g$ is known up to $\beta$. In contrast, CSDR assumes $\E(Y^x) = g(\beta^\top x)$, where $g$ is an \textit{unknown} smooth function of a low-dimensional index $\beta^\top x$. The causal sufficiency condition in \eqref{eq:cs_constraint} can therefore be viewed as a nonparametric multiple-index MSM. 

Extending the general MSM framework of \cite{robinsMarginalStructuralModels2000a} to CSDR (details in Lemma~\ref{lem:robins_no_phi}), \cite{nabiSemiparametricCausalSufficient2022} derived a class of locally efficient regular and asymptotically linear (RAL) estimators for $\CCMS$. Let $\mathbb{E}^*(\cdot)$ denote expectation under the interventional distribution so that $\E^*(Y\mid X=x) = \E(Y^x)$. Their estimators are defined as solutions to the empirical estimating equations corresponding to $\mathbb{E}[\tilde{U}(O, \beta)] = 0$, where 
\begin{equation}\label{eq:RAL_MSM_beta}
\begin{aligned}
\tilde{U}(O; \beta)
&=
\frac{\pi^*(X)}{\pi(X \mid C)}
\left\{
U^*(O; \beta)
- \mathbb{E}\!\left[ U^*(O; \beta)\mid X, C \right]
\right\}
+
\mathbb{E}^*\big[
\mathbb{E}\!\left[U^*(O; \beta)\mid X, C\right]
\ \big| \ C
\big],
\\
U^*(O; \beta)
&=
\{Y - \mu_\beta(Z)\}
\{\alpha(X) - \nu_\beta(Z)\}.
\end{aligned}
\end{equation}
Here, $\pi^*(X)$ is any specified density with the same support as $X$, $\nu_\beta\coloneqq\E^*[\alpha(X) \mid Z]$, and $\alpha(\cdot)$ is an arbitrary function of $X$. The estimating function $U^*(\beta)$ mirrors the structure of a valid estimating function for $\CMS$ \citep{maSemiparametricApproachDimension2012}, but replaces all regression quantities with their causal analogues, namely $\mu_\beta(\cdot)$ and $\nu_\beta(\cdot)$ in place of $\E(Y\mid Z)$ and $\E[\alpha(X) \mid Z]$. Efficiency of $\tilde{U}(O;\beta)$ can be further improved by choosing $U^*(O; \beta)$ to be the causal analogue of the association-based efficient score in \eqref{eq:S_eff}, yielding a practical approximation to the efficient score for the MSM \citep{robinsMarginalStructuralModels2000a}.

Estimating $\beta$ through the MSM approach presents several challenges. First, suitable choices of $\pi^*(X)$ and $\alpha(\cdot)$ are required, yet performance can be sensitive to these choices and optimal specifications are generally not straightforward. Second, the estimating equation for $\beta$ depends on multiple nuisance functions, including (i) $\mu_\beta(\cdot)$, (ii) $\nu_\beta(\cdot)$, (iii) $\pi(\cdot\mid \cdot)$, and (iv) $m(\cdot,\cdot)$. Under the practical efficient score approximation for the MSM described above, additional nuisances must be estimated, including (v) $\mu'_\beta(z)\coloneqq\partial \mu_\beta(z) / \partial z$, (vi) $\Var(Y^x)$, (vii) $\mathbb{E}^*\left[X / \Var(Y^X) \mid \beta^\top X\right]$, and (viii) $\mathbb{E}^*\left[1 /\Var(Y^X) \mid \beta^\top X\right]$. Although the double robustness of \eqref{eq:RAL_MSM_beta} relaxes the need for all nuisances to be correctly specified, valid inference and stable implementation still require accurate estimation of a substantial subset of these quantities. Another complication is that the nuisance models for $\mu_\beta(\cdot)$ and $m(\cdot, \cdot)$ are not \textit{variationally independent}, meaning that specification of one model implicitly restricts the other. Finally, score-based estimation of $\beta$ relies on iterative numerical optimization in which all $\beta$-dependent nuisance functions must be re-estimated at each new update. Even though nuisances such as $\mu_\beta(\cdot)$ and $\nu_\beta(\cdot)$ are defined on the lower dimensional $Z$, nonparametric causal mean estimation remains computationally intensive relative to standard regression. Numerical optimization additionally requires finite-difference approximations of the Hessian, introducing further approximation error. These difficulties motivate an alternative strategy that separates causal effect estimation from dimension reduction. 

\subsection{A modular framework for causal SDR}

Rather than deriving and solving a causal SDR estimating equation directly, we reduce estimation of $\CCMS$ to a standard SDR problem on transformed data. Suppose the causal sufficiency condition in \eqref{eq:cs_constraint} holds. Further suppose we can construct transformed variables $(\tilde Y,\tilde X)$ satisfying 

\vspace{-1.6cm}
\begin{align}
\label{eq:modular_condition}
    \E(\tilde Y\mid \tilde X)
    =
    \tilde g(\beta^\top \tilde X)
\end{align}

\vspace{-0.35cm} \noindent
for some smooth function $\tilde g$. Then the CMS of $\tilde Y$ given $\tilde X$ coincides with the target CCMS. This observation motivates a modular estimation strategy: (i) Construct transformed variables $(\tilde Y,\tilde X)$ satisfying \eqref{eq:modular_condition}. (ii) Apply a standard SDR estimator of $\CMS$ to  $(\tilde Y,\tilde X)$ so that $\text{span}(\hat\beta)$ estimates $\CCMS$. (iii) Form the reduced exposure $\hat Z=\hat\beta^\top X$ and estimate the causal ERS through the reduced surface $\mu_\beta(z)$. Separating causal nuisance estimation from the dimension reduction avoids much of the technical and computational burden associated with MSM-based estimation of the CCMS, and is similar in spirit to pseudo-outcome methods for continuous treatments and heterogeneous treatment effects \citep{kennedyNonparametricMethodsDoubly2017, bonviniFastConvergenceRates2026, kramerCausalInferenceHighDimensional2026}.

\subsubsection{Constructing causal SDR targets}

We consider three constructions of $(\tilde Y,\tilde X)$ satisfying \eqref{eq:modular_condition}. The first uses direct estimates of the exposure-response surface, the second uses doubly robust pseudo-outcomes whose conditional mean given $X$ equals the ERS, and the third uses residualized variables under an additive confounding model. These constructions entail different tradeoffs between robustness, statistical efficiency, and computational complexity.

\paragraph{Exposure-response surface targets.}

The most direct way to satisfy \eqref{eq:modular_condition} is to use the exposure-response surface itself by setting

\vspace{-1.2cm}
\begin{align}
    \tilde X = X,
    \qquad
    \tilde Y = \mu(X).
\end{align}

\vspace{-0.5cm} \noindent
Under the causal sufficiency condition \eqref{eq:cs_constraint}, $\E[\mu(X)\mid X=x]=\mu(x)=\mu_\beta(\beta^\top x)$. Since $\mu(X)$ is not directly observed, we estimate it at the observed exposures $\{X_j\}_{j=1}^n$ using the RA or DR estimators from Section~\ref{subsec:setup}. The DR estimator evaluated at $X_j$ is 
\begin{equation*}
\hat{\mu}_{DR}(X_{j})=\frac{1}{nh^p}\sum_{i=1}^n \left\{\frac{K_{h}(X_{i}-X_{j})}{\hat{\pi}(X_{i}\mid C_{i})}\big[Y_{i}-\widehat{m}(X_{j},C_{i})\big] \right\}+\frac{1}{n}\sum_{i=1}^n\widehat{m}(X_{j}, C_{i}),
%\label{eq:DR_outcome}
\end{equation*}
while the RA estimator is
\begin{equation*}
\hat{\mu}_{RA}(X_{j})=\frac{1}{n}\sum_{i=1}^n\widehat{m}(X_{j}, C_{i}).
\end{equation*}

The transformed sample $\{\big(\hat{\mu}(X_i), X_i\big)\}_{i=1}^n$ can then be used within any valid CMS estimator. We refer to these variants as RA and DR depending on whether $\hat\mu_{RA}$ or $\hat\mu_{DR}$ is used in the first stage.

\paragraph{Pseudo-outcome targets.}

A second construction uses a pseudo-outcome (PO) whose conditional mean given $X$ equals the ERS. Extending the pseudo-outcome of \cite{kennedyNonparametricMethodsDoubly2017} for continuous treatments to the multivariate setting, define
\begin{equation}
\xi_{j} \coloneqq \frac{\frac{1}{n}\sum_{i=1}^n \pi(X_{j}\mid C_{i}) }{\pi(X_{j}\mid C_{j})} \left\{ Y_{j}-m(X_{j},C_{j}) \right\} + \frac{1}{n}\sum_{i=1}^n m(X_{j}, C_{i}).
\label{eq:pseudo_outcome}
\end{equation}

Under the identification assumptions of Section~\ref{subsec:setup}, $\E(\xi\mid X = x)=\mu(x)=\mu_\beta(\beta^\top x)$ even when $p>1$, satisfying \eqref{eq:modular_condition}; see Lemma~\ref{lemma:pseudo_outcome_unbiased}. Since $\xi$ is also not observed, we construct $\hat\xi$ by replacing $m(\cdot,\cdot)$ and $\pi(\cdot\mid\cdot)$ in \eqref{eq:pseudo_outcome} with estimates. 

Both the DR and pseudo-outcome constructions are doubly robust in the sense that $\E(\tilde Y\mid X=x)=\mu(x)$ whenever either the outcome regression or generalized propensity score is correctly specified. A key distinction is that the DR construction targets the deterministic function $\mu(X)$, whereas the pseudo-outcome construction uses a stochastic response whose conditional mean equals $\mu(X)$. This distinction may matter in finite samples, as the pseudo-outcome construction can present a harder second-stage SDR problem due to the additional layer of noise.

The preceding constructions still require estimation of $m(\cdot,\cdot)$ and $\pi(\cdot\mid\cdot)$, which may appear at odds with the goal of avoiding direct high-dimensional modeling in $X$. However, the MSM approach requires these same nuisance functions together with several additional nuisance quantities inside iterative score-based optimization. In contrast, the modular framework uses nuisance estimation only to construct first-stage targets before applying standard SDR methods. Moreover, the first-stage estimator need only be sufficiently accurate to recover a useful estimate of the CCMS, after which the final ERS can be estimated on the reduced exposure space. We formalize this point in Theorem~\ref{thm:csMAVE_convergence}.

\paragraph{Residualized pairs.} 

Under additional structural assumptions, however, it is possible to avoid direct regression on $X$ altogether. Specifically, under an additive confounding model, we propose a third construction, the \emph{residualized pair} (RP). 

\begin{assumption}[Additive confounding model]
\label{ass:additive_confounding}
There exist measurable functions $g$, $h$, and $f$ such that
\[
\begin{aligned}
Y &= g(\beta^\top X) + h(C) + \varepsilon_Y, \qquad \text{where} \quad \mathbb{E}(\varepsilon_Y \mid X, C) = 0, \\
X &= f(C) + \varepsilon_X, \qquad \text{where} \quad \mathbb{E}(\varepsilon_X \mid C) = 0 \ \text{and} \ \varepsilon_X \indep C.
\end{aligned}
\]
\end{assumption}

Under Assumption~\ref{ass:additive_confounding}, it is possible to avoid estimation of both $m(\cdot,\cdot)$ and $\pi(\cdot\mid\cdot)$ by appealing to a variation of the classical Frisch-Waugh-Lovell (FWL) theorem \citep{frischPartialTimeRegressions1933, lovellSeasonalAdjustmentEconomic1963}.

\begin{prop}
\label{prop:FWL}
    Let $\widetilde{Y}=Y-\E(Y\mid C)$ and $\tilde{X}=X-\E(X\mid C)$. Under Assumption~\ref{ass:additive_confounding}, there exists a smooth function $\tilde{g}$ such that $\E(\tilde{Y} \mid \tilde{X}) = \tilde{g}(\beta^\top \tilde{X})$. 
\end{prop}

See a proof in Appendix~\ref{sec:proofs_prop:FWL}. 

Proposition~\ref{prop:FWL} shows that nuisance estimation of $m(\cdot,\cdot)$, a regression with $p+q$ regressors, and $\pi(\cdot\mid\cdot)$, a $p$-variate conditional density with $q$ regressors, can instead be replaced by estimation of $p+1$ independent regressions each involving only the confounders $C$. These regressions are lower-dimensional, can be fit in parallel, and may converge faster than direct nonparametric estimators of $m(\cdot,\cdot)$ and $\pi(\cdot\mid\cdot)$, particularly when $p$ is large. However, unlike the DR and PO constructions, the residualized-pair approach is not doubly robust. Validity of Proposition~\ref{prop:FWL} requires all $p+1$ residualization regressions to be correctly specified. 

\begin{remark}
At first glance, the quantity $\E(X\mid C)$ may resemble the inverse regression $\E(X\mid Y)$ used in sliced inverse regression (SIR) \citep{liSlicedInverseRegression1991}. However, the objectives are fundamentally different. In 
SIR, inverse regression is used to estimate covariance structures characterizing the central subspace. Here, $\E(X\mid C)$ is used only to remove confounding variation from the exposure vector. Consequently, only accurate pointwise estimation is required, and higher-order features of the conditional distribution of $X$ given $C$ are not of interest. 
\end{remark}

\subsubsection{Estimating the CCMS via csMAVE}

Once a transformed pair $(\tilde{Y}, \tilde{X})$ satisfying \eqref{eq:modular_condition} has been constructed, the CCMS can be estimated using any valid CMS estimator applied to the transformed data. For the second stage, we propose using MAVE in \eqref{eq:MAVE} due to its favorable balance between statistical and computational efficiency. We refer to the resulting estimator as \textit{causally sufficient MAVE} (\textit{csMAVE}). 

Our theoretical analysis focuses on csMAVE applied to the ERS construction ${\{(\hat{\mu}}(X_i), X_i)\}_{i=1}^n$, which captures the essential structure of the modular framework while avoiding additional stochastic variation contained in the PO and RP variants. The following assumptions adapt standard MAVE regularity conditions to this setting.  

\begin{assumption}[csMAVE]\label{ass:csMAVE}
The following conditions hold:
\begin{enumerate}[label=(\alph*)]
\item The marginal density $f_X(x)$ of $X$ has bounded fourth derivative and is strictly bounded away from zero on a compact support $\mathcal{X}$.
\item The ERS satisfies $\mu(x)=g(\beta^\top x)$ where $g(\cdot)$ has bounded continuous third derivatives.
\item The kernel $K(\cdot)$ used to construct the MAVE weights is a spherically symmetric density function with a bounded derivative and support. All moments of $K(\cdot)$ exist and $\int UU^\top K(U) dU=I$. 
\item Let $V(z)= \partial g(z_1,...,z_{d_0}) / \partial z$ and $V_k(z)=\partial g(z_1,...,z_{d_0}) / \partial z_k$. The matrix 

$\E\{ V(\beta_0^\top X)V^\top(\beta_0^\top X) \}$ is nonsingular.
\item The first-stage estimator satisfies the empirical uniform convergence condition 

$\max\limits_{1 \le i \le n} \| \hat{\mu}(X_i) - \mu(X_i) \| = O_p(R_n)$ where $R_n=o(h^2)$. 
\end{enumerate}
\end{assumption}

Further discussion of these assumptions is provided in Appendix~\ref{sec:ass}. The next result characterizes how first-stage estimation error propagates through the second-stage MAVE procedure. 

\begin{theorem}[Convergence of csMAVE]\label{thm:csMAVE_convergence}
    Let $\beta_0$ denote a basis matrix of the true CCMS and $\hat{\beta}$ denote a minimizer of the csMAVE objective obtained using an ERS construction. Define $\delta_n=\{\log n / (nh^p)\}^{1/2}$. Under Assumptions~\ref{ass:identificationX}-\ref{ass:existence_CCMS}, \ref{ass:csMAVE}, and \ref{ass:mave-local-moments} (see Appendix~\ref{subsec:intermediate_thm1}) if $h\to 0$, $d \geq d_0$, and $nh^p / \log(n)\to \infty$, then 
    \begin{equation*}
        \big\lVert (I_p-\hat{\beta}\hat{\beta}^\top) \beta_0\big\rVert = O_p(h^3 + h \delta_n + h^{-1}R_n)
    \end{equation*}
\end{theorem}

See a proof in Appendix~\ref{sec:proofs_thm:csMAVE_convergence}. 

Theorem~\ref{thm:csMAVE_convergence} shows that first-stage estimation using the ERS construction contributes the additional term $O_p(h^{-1}R_n)$ on top of the standard MAVE convergence rate $O_p(h^3 + h\delta_n + h^{-1}\delta_n^2)$. The $h^{-1}\delta_n^2$ term from the original MAVE is absent in the result because the new response $\mu(X)$ is a deterministic function of $X$. Consequently, consistency requires the first-stage ERS estimator to converge sufficiently quickly relative to the MAVE bandwidth. Assumption~\ref{ass:csMAVE}(e) imposes the stronger condition $R_n=o(h^2)$, which is used to establish consistency of the structural dimension estimator.

\begin{remark}
The convergence rate in Theorem~\ref{thm:csMAVE_convergence} depends on the exposure dimension $p$ because standard MAVE constructs kernel weights in the original exposure space. Refined MAVE (RMAVE) instead defines the kernel weights on the lower-dimensional index $\beta^\top X$, potentially yielding faster convergence rates \citep{xiaAdaptiveEstimationDimension2002}. We do not pursue a formal RMAVE analysis here because it introduces substantial additional technical complexity, but analogous arguments can be used to establish a corresponding result. 
\end{remark}

We estimate the structural dimension of $\CCMS$ using the  cross-validation (CV) procedure proposed by \cite{xiaAdaptiveEstimationDimension2002}, replacing the original data $(Y,X)$ with the transformed pair $(\tilde{Y}, \tilde{X})$. In the setting of our theoretical analysis, $\tilde{Y}=\hat{\mu}(X)$ and $\tilde{X}=X$. For each candidate dimension \( d \in \{0,1,2,\dots,p\} \) with corresponding csMAVE estimate \( \hat{\beta}_d \), we estimate \( \E\!\left[\hat{\mu}(X_i) \mid \hat{\beta}_d^\top X_i \right] \) using a Nadaraya--Watson kernel estimator. The cross-validation criterion \( \mathrm{CV}(d) \) is defined as the mean squared error between \( \hat{\mu}(X_i) \) and the kernel estimate, and the selected dimension is given by
\begin{equation*}
\hat{d} = \argmin_{d \in \{0,1,\dots,p\}} \mathrm{CV}(d).
\end{equation*}

\begin{theorem}[Consistency of $\hat{d}$]
   Under Assumptions~\ref{ass:identificationX}-\ref{ass:existence_CCMS} and \ref{ass:csMAVE},  the estimator $\hat{d}$ converges to the true dimension $d_0$ in probability as $n\to\infty$. 
    \label{thm:structural_dimension_consistency}
\end{theorem}

See a proof in Appendix~\ref{sec:proofs_thm:structural_dimension_consistency}. 

\subsection{Efficient-score refinement in modular causal SDR}

In association-based SDR for $\CMS$, a standard route to efficient estimation is to obtain an initial estimate of $\beta$ using MAVE and then refine it through Newton--Raphson updates based on the semiparametric efficient score \citep{luoEfficientDimensionReduction2014, luoNewEstimatorEfficient2016}. The efficient score in \eqref{eq:S_eff} underlying this refinement contains inverse variance weighting through $\sigma^2(X)=\Var(Y\mid X)$, and is therefore well-defined only when $\sigma^2(x)$ is strictly positive over $\mathcal X$. Under the ERS construction, however, the new response variable is $\tilde Y=\mu(X)$, so that $\Var(\tilde Y\mid X)=0$. It follows that in our CSDR setting, the inverse variance weighting in the efficient score becomes degenerate and the usual semiparametric efficiency theory underlying efficient-score refinement no longer applies. MAVE, by contrast, does not require nondegenerate conditional variance. It remains well-defined whenever $\E(\tilde Y\mid X) = g(\beta^\top X)$ with sufficiently smooth $g$. Thus, the ERS construction is compatible with MAVE even though the efficient score is not theoretically justified.

The fact that efficient-score refinement is not available for modular CSDR estimation using the ERS construction may suggest that MAVE is merely a fallback estimator and that efficiency is lost. We argue that this is not the case. Proposition~\ref{prop:MAVE_asymptotic_efficient} shows that when $\sigma^2(x)=\sigma^2>0$ for all $x\in\mathcal X$, the estimating equation induced by MAVE converges to the semiparametric efficient score where the inverse-variance weights drop out as a constant. Hence, under homoscedasticity, MAVE is already asymptotically equivalent to the efficient score.

The ERS construction corresponds to the degenerate limiting case of this same idea. Once $\tilde{Y}=\mu(X)$ is used as the new response, the conditional variance is constant at zero and the efficient score is undefined. However, just as in the constant variance case, there is no conditional variance for the inverse-variance weighting to exploit. Thus, MAVE can be viewed as a natural second-stage estimator when the conditional variance is either constant or degenerate.  

\begin{prop}[Asymptotic efficiency of RMAVE]
\label{prop:MAVE_asymptotic_efficient}
    Let $h\to 0$ and $nh^{d+2} / \log(n) \to \infty$. If $\sigma^2(x)=\sigma^2>0$ for $x\in\mathcal{X}$, then the estimating equation implied by RMAVE is asymptotically equivalent to the efficient score for $\CMS$ in \eqref{eq:S_eff}.
    % then as $n\to\infty$, the RMAVE objective function is asymptotically equivalent to the efficient score in \eqref{eq:S_eff} 
\end{prop}
See a proof in Appendix~\ref{sec:proof_MAVE_eff}.

To further investigate this issue empirically, Appendix~\ref{sec:reproducing_sdr} documents our implementation of the efficient-score refinement procedure of \cite{luoNewEstimatorEfficient2016}, while Appendix~\ref{sec:sim_mave_vs_ee} compares MAVE and efficient-score updates within the modular CSDR framework. In our numerical experiments, efficient-score refinement degraded performance for the ERS construction and produced only marginal improvements for the pseudo-outcome construction. Consequently, all numerical results in the paper use csMAVE without additional efficient-score refinement.

\section{Interpreting the CCMS: subspace importance score}
\label{sec:importance_score} 

Interpretability is an important component of mixture analysis and dimension reduction. In principal component analysis (PCA), the entries of the loading vectors quantify how strongly each original variable contributes to the resulting principal components. Because PCA is obtained through an eigendecomposition, the loading vectors are unique up to sign, making interpretation of their entries relatively straightforward.
In causal SDR, however, direct interpretation of the entries of $\beta$ can be misleading because $\beta$ itself is not uniquely identified. Any nonsingular linear transformation of $\beta$ also spans the CCMS and therefore corresponds to an identical causal reduction. Consequently, different bases for the CCMS may yield substantially different coefficient values despite representing the same underlying subspace. To obtain an interpretable quantity that depends only on the CCMS itself rather than on a particular basis, we propose the \emph{subspace importance score}.

\begin{definition}[Subspace importance score]
Let $\beta$ be a basis matrix for a $d$-dimensional CCMS with orthogonal projection matrix $P = \beta(\beta^\top \beta)^{-1} \beta^\top$. The subspace importance score for exposure $j \in \{1,\dots,p\}$ is defined as $\lambda_j \coloneqq P_{jj}$, or $e_j^\top P e_j = \|P e_j\|^2$.
\end{definition}

The quantity $\lambda_j$ measures how strongly the coordinate direction corresponding to exposure $X_j$ aligns with the CCMS. Because $P$ depends only on the subspace $\text{span}(\beta)$, the importance scores are invariant to the particular basis representation used for estimation.

The subspace importance scores satisfy several immediate properties.

\vspace{-0.25cm}
\begin{enumerate}
    \item \textbf{Basis invariance:} $\lambda_j$ depends only on the CCMS $\text{span}(\beta)$ and is invariant to full-rank linear transformations of $\beta$.
    \item \textbf{Boundedness:} For each $j=1,\dots,p$, $\lambda_j \in [0,1]$.
    \item \textbf{Additive decomposition:} $\sum_{j=1}^p \lambda_j = d$.
\end{enumerate}

These properties imply that $\{\lambda_j\}_{j=1}^p$ decomposes the structural dimension across the original exposure coordinates. Larger values of $\lambda_j$ indicate that the $j$th exposure direction contributes more strongly to the CCMS. When $\lambda_j=1$, the coordinate direction $e_j$ lies entirely within the CCMS. In contrast, $\lambda_j=0$ implies that the coordinate direction is orthogonal to the CCMS and therefore does not contribute to the reduced causal representation. 

Although asymptotic theory for the estimates of the projection matrix and the subspace important scores, $\hat P$ and $\hat\lambda_j$ respectively, can in principle be derived using the Delta method, the resulting expressions are cumbersome because of the nonlinear dependence on the estimated projection matrix. In practice, we recommend using the nonparametric bootstrap to construct confidence intervals for the subspace importance scores.

\section{Simulation studies}
\label{sec:sims}

Our simulation study evaluates two aspects of causal SDR: (i) estimation of the CCMS, and (ii) estimation of the reduced causal exposure-response surface after replacing $X$ with the estimated reduced exposure $\hat Z=\hat\beta^\top X$.  

We generated $B=1000$ Monte Carlo replications with $p=10$ exposures and $q=5$ confounders. The confounders were distributed as $C\sim N_q(0, I_q)$. Conditioned on $C$, the exposures followed $[X \mid C]\sim N_p(\mu_{X\mid C}, \Sigma_{X\mid C})$ where $\mu_{X\mid C}\coloneqq\E(X\mid C)=(0.5C_1, 0.5C_2, 3C_3, 3C_4,\mathbf{0}_{p-4}^\top)^\top$ and $\Sigma_{X\mid C}[i,j]=0.8^{|i-j|}$, to resemble the strong correlations commonly observed in environmental mixtures. We additionally set $\Sigma_{X\mid C}[k,k]=10.0$ for $k\in\{5,6\}$. We define the latent causal exposure as $Z = \beta^\top X$, where $\beta = [e_1, e_2] \in \mathbb{R}^{p \times 2}$. Thus, the true structural dimension is $d_0 = 2$ and $Z = (X_1, X_2)$. The outcome was generated according to 
\begin{equation} \label{eq:main_sim_dgp}
\begin{aligned}
    Y &= h(Z)+g(C)+f(Z,C)+\varepsilon, \qquad \varepsilon \sim N(0, 0.5^2),\\
    h(Z) &= 4\sin(Z_1)+2Z_2^2+Z_1Z_2, \\
    g(C) &= 5 \{\tanh(C_1) + [C_2^2-1] + [C_3^2-1] + \sin(1.5C_4) \}, \\
    f(Z,C) &= 5 C_5 (Z_1+Z_2).
\end{aligned}
\end{equation}

Under this data generating process (DGP), the true causal ERS is $\mu_\beta({z})= h({z})$. We considered sample sizes $n=100, 500,1000,2500,$ and $5000$. For each simulated dataset, we estimated $\beta$ using baseline methods: (1) PCA, (2) partial canonical correlation analysis (pCCA), (3) association-based MAVE, along with causal SDR variants: (4) Oracle, (5) RA, (6) DR, (7) PO, and (8) RP. 

The Oracle estimator is not practically implementable and assumes access to the true values $\{\mu(X_i)\}_{i=1}^n$. It is included as a benchmark representing the ideal second-stage SDR procedure absent first-stage estimation error. For all causal SDR variants, the structural dimension was estimated rather than fixed at the true value $d_0$.

All nuisance functions were estimated using 5-fold cross-fitting. We estimated $\pi(x\, | \, c)$ under a working Gaussian model assumption by first estimating $\E(X\, | \, C)$ using linear regression and then estimating the covariance with the sample covariance of the residuals $X-\hat{\E}(X\, | \, C)$. For estimation of $m(x,c)$, $\E(X\, | \, C)$, and $\E(Y\, | \, C)$ in the first stage, we used the Super Learner \citep{vanderlaanSuperLearner2007} with 10-fold cross-validation and three learners: linear regression (\texttt{SL.glm}), elastic net (\texttt{SL.glmnet}), and multivariate adaptive regression splines (\texttt{SL.earth}). For estimation of $m_\beta(\cdot, \cdot)$, we fit a neural network in \texttt{torch} with 2 hidden layers ($100 \times 50$ nodes) fit for 150 epochs, an initial learning rate of 0.05, a step-decay scheduler that halved the learning rate every 50 epochs, Sigmoid Linear Unit (SiLU) activation functions, MSE loss, no weight decay, and input standardization. 

We measured estimation error of the $\CCMS$ using $\lVert {P}_{\text{span}(\hat{\beta})}-{P}_{\text{span}(\beta_0)}\rVert_F$, where ${P}_{\text{span}(\hat{\beta})}=\hat{\beta}(\hat{\beta}^\top \hat{\beta})^{-1} \hat{\beta}^\top$ is the unique orthogonal projection into the $d$-dimensional subspace spanned by $\hat{\beta}$ in $\mathbb{R}^p$ and $\Vert \cdot \rVert_F$ is the Frobenius norm. Estimation error of $\hat{d}$ was evaluated by the proportion of simulations the correct dimension $d_0$ was chosen. To evaluate the estimation of the causal ERS, we computed the root mean squared error (RMSE) and pointwise 95\% coverage of $\hat{\mu}_\beta(z)$ across 100 evaluation points sampled from the true density of $(Z_1, Z_2)$. 

Additional simulations assessing CSDR performance under additive confounding in the DGP, alternative nuisance estimators for ERS estimation, and a less sparse $\beta$ are reported in Appendix~\ref{sec:additional_sim}. 

\subsection{Estimation of CCMS}

% Frobenius norm error
\begin{table}[t]
\centering
\caption{Frobenius norm error between estimated and true subspace using $\hat{d}$.}
\label{tab:frob_error_dhat}
% --- Start Resizebox ---
\resizebox{\textwidth}{!}{
\begin{tabular}{llccccc} % Added an 'l' here for the new category column
\toprule
Category & Method & $n = 100$ & $n = 500$ & $n = 1000$ & $n = 2500$ & $n = 5000$ \\ 
\midrule
\multirow{3}{*}{Baseline} 
& PCA         & 1.997 (0.002) & 1.997 (0.001) & 1.997 (0.001) & 1.997 (0.000) & 1.997 (0.000) \\
& pCCA        & 1.316 (0.119) & 1.202 (0.177) & 1.100 (0.202) & 0.925 (0.219) & 0.743 (0.203) \\
& MAVE        & 1.558 (0.221) & 1.292 (0.165) & 1.174 (0.103) & 1.097 (0.041) & 1.069 (0.021) \\
\midrule % Adding a midrule here makes the separation much cleaner
\multirow{5}{*}{CSDR} 
& Oracle & 0.216 (0.136) & 0.030 (0.010) & 0.015 (0.005) & 0.006 (0.002) & 0.003 (0.001) \\
& RA      & 1.343 (0.341) & 0.447 (0.392) & 0.223 (0.288) & 0.073 (0.157) & 0.026 (0.035) \\
& DR      & 1.537 (0.267) & 0.568 (0.369) & 0.303 (0.257) & 0.140 (0.151) & 0.078 (0.058) \\
& PO      & 1.487 (0.318) & 0.663 (0.450) & 0.432 (0.401) & 0.287 (0.345) & 0.221 (0.293) \\
& RP      & 1.768 (0.218) & 1.413 (0.259) & 1.192 (0.240) & 0.913 (0.255) & 0.676 (0.285) \\
\bottomrule
\end{tabular}
} % --- End Resizebox ---
\end{table}

Table~\ref{tab:frob_error_dhat} reports the average distance between the estimated and true $\CCMS$, with errors ranging from zero to two. The causal SDR methods outperformed the baseline approaches in recovering $\CCMS$, and this gap increased with $n$. The Oracle variant shows that knowledge of the true $\mu(X)$ was the main limiting factor for estimating $\beta$ well, where at $n=100$ the Oracle error was already down to $0.203$ whereas the next closest error was the RA variant at $1.343$. This indicates a strong dimension reduction signal disrupted by noise from confounding, interactions, and measurement error. Although both pCCA and MAVE improved with sample size, their error decreased noticeably slower than those of the causal SDR methods. As expected unsupervised PCA was poor across all settings, as it prioritized directions of high variance in $X$ which were unrelated to $\CCMS$. 

% Dhat table
\begin{table}[!b]
\centering
\caption{Dimension selection proportions. True dimension ($d_0 = 2$) is highlighted in bold.}
\label{tab:dim_selection}
\begin{tabular}{ll c c >{\bfseries}c c c c}
\toprule
% & & \multicolumn{6}{c}{Selected Dimension Proportion} \\
\cmidrule(lr){3-8}
Sample Size & Method & $d=0$ & $d=1$ & $d=2$ & $d=3$ & $d=4$ & $d \ge 5$ \\ 
\midrule
\multirow{6}{*}{$n = 500$} 
& MAVE        & 0.00 & 0.00 & 0.11 & 0.75 & 0.14 & 0.00 \\
& Oracle      & 0.00 & 0.00 & 1.00 & 0.00 & 0.00 & 0.00 \\
& RA     & 0.01 & 0.03 & 0.88 & 0.07 & 0.01 & 0.00 \\
& DR     & 0.01 & 0.05 & 0.85 & 0.07 & 0.01 & 0.00 \\
& PO     & 0.03 & 0.06 & 0.73 & 0.14 & 0.03 & 0.00 \\
& RP     & 0.00 & 0.28 & 0.29 & 0.23 & 0.15 & 0.04 \\
\midrule
\multirow{6}{*}{$n = 5000$} 
& MAVE        & 0.00 & 0.00 & 0.00 & 1.00 & 0.00 & 0.00 \\
& Oracle      & 0.00 & 0.00 & 1.00 & 0.00 & 0.00 & 0.00 \\
& RA     & 0.00 & 0.00 & 1.00 & 0.00 & 0.00 & 0.00 \\
& DR     & 0.00 & 0.00 & 1.00 & 0.00 & 0.00 & 0.00 \\
& PO     & 0.01 & 0.00 & 0.92 & 0.07 & 0.00 & 0.00 \\
& RP     & 0.00 & 0.21 & 0.74 & 0.05 & 0.00 & 0.00 \\
\bottomrule
\end{tabular}
\end{table}

Among the causal SDR approaches, the RA and DR variants, which rely on first-stage estimation of $\mu(X)$, consistently outperformed PO and RP. One possible explanation is that there exists an advantage of estimating a deterministic quantity like $\mu(x)$, in contrast to the stochastic PO and RP target variables. Furthermore, the presence of an interaction term violates the additive confounding assumption underlying the RP method, which explains its comparable performance to pCCA. In Appendix~\ref{subsec:add_sim_add_confound}, we examined the RP variant under an additive confounding DGP by setting $f(Z,C)=0$ in \eqref{eq:main_sim_dgp}. Under this setting, RP estimation of the CCMS dramatically improved, in line with theoretical expectations.  

Table~\ref{tab:dim_selection} reports the proportion of simulations in which the structural dimension is correctly estimated. All causal SDR methods recovered the true dimension $d_0 = 2$, although PO and RP did so less frequently. The association-based MAVE ultimately selected $d = 3$ due to confounding. The SIS in Figure~\ref{fig:sim_SIS} reveals additional insight that MAVE placed $X_3$ and $X_7$--$X_{10}$ in $\CCMS$, whereas CSDR correctly identified $X_1$ and $X_2$.

\begin{figure}[t]
    \centering
    \includegraphics[width=1\textwidth]{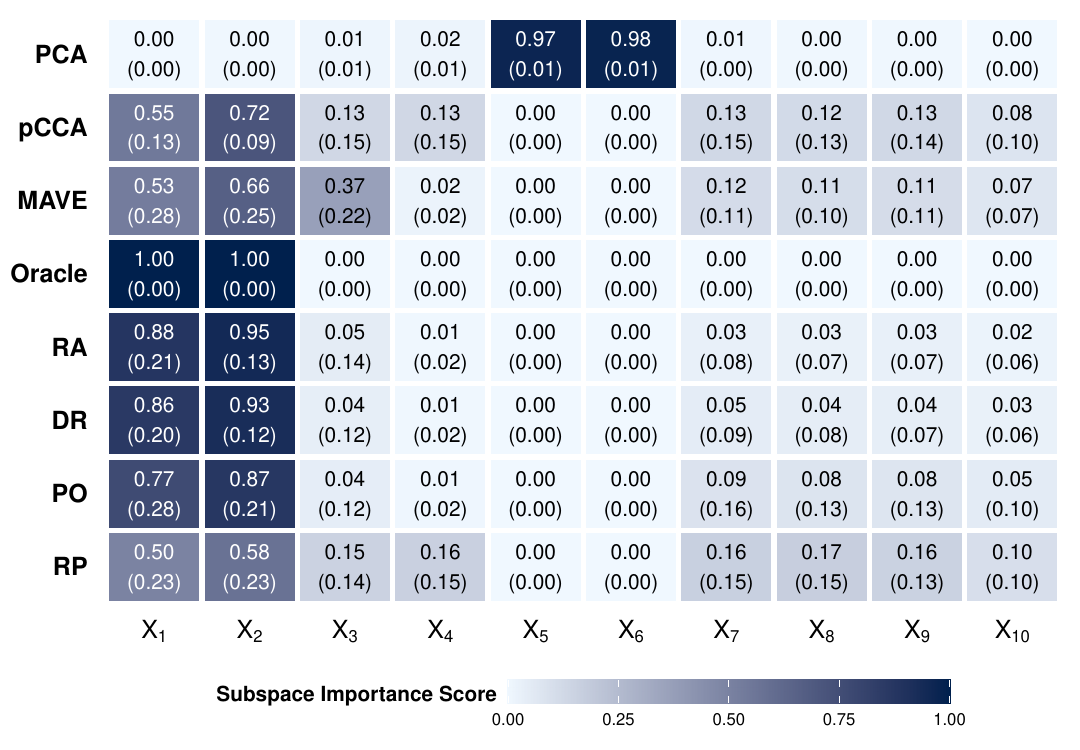}

    \vspace{-0.25cm}
    \caption{Subspace importance scores (SIS) for each method under a working dimension of $d=2$ at $n=500$. The top row in each cell reports the mean SIS, and the bottom row reports the standard deviation across simulations.}
    \label{fig:sim_SIS}
\end{figure}

\subsection{Estimation of ERS}

\begin{table}[htbp]
\centering
\caption{Mean (SD) RMSE for ERS estimation across the 100 evaluation points.}
\label{tab:rmse_ers}
\resizebox{\textwidth}{!}{
\begin{tabular}{llccccc}
\toprule
Category & Method & $n = 100$ & $n = 500$ & $n = 1000$ & $n = 2500$ & $n = 5000$ \\
\midrule
\multirow{3}{*}{Baseline}
& pCCA   & 14.07 (10.15) & 11.49 (6.41) & 10.95 (6.64) & 7.89 (5.15) & 5.75 (3.86) \\
& MAVE   & 10.09 (6.61)  & 4.17 (1.76)  & 3.13 (1.00)  & 2.52 (0.45) & 2.35 (0.34) \\
& Full $X$ & 9.36 (5.86)  & 3.98 (1.15)  & 3.01 (0.70)  & 2.35 (0.38) & 2.12 (0.26) \\
\midrule
\multirow{5}{*}{CSDR}
& Oracle & 5.03 (4.65) & 1.47 (1.06) & 0.97 (0.63) & 0.62 (0.43) & 0.48 (0.32) \\
& RA     & 7.48 (5.26) & 1.97 (1.36) & 1.22 (0.78) & 0.69 (0.53) & 0.49 (0.34) \\
& DR     & 8.65 (5.70) & 2.09 (1.29) & 1.28 (0.79) & 0.73 (0.50) & 0.53 (0.36) \\
& PO     & 8.26 (5.86) & 2.19 (1.47) & 1.36 (0.76) & 0.82 (0.57) & 0.62 (0.40) \\
& RP     & 14.36 (9.35) & 6.40 (3.59) & 4.41 (2.03) & 3.05 (1.16) & 2.48 (0.86) \\
\bottomrule
\end{tabular}
}
\end{table}

Table~\ref{tab:rmse_ers} reports the error in estimating $\mu_\beta(\cdot)$ over the 100 evaluation points, using $\hat{Z} = \hat{\beta}^\top X$ in place of $X$. PCA was excluded due to its poor performance in estimating $\CCMS$, and replaced with ``Full $X$,'' which applies the DR estimator directly to $X$ without dimension reduction. All CSDR methods (except the misspecified RP variant) outperformed Full $X$, pCCA, and MAVE in point estimation across all $n$. Although pCCA achieved the lowest error among baseline methods for estimating $\CCMS$ in Table~\ref{tab:frob_error_dhat}, it had much higher error than even MAVE for estimating $\mu_\beta(\cdot)$.

\begin{figure}[ht]
    \centering
    \includegraphics[width=1\textwidth]{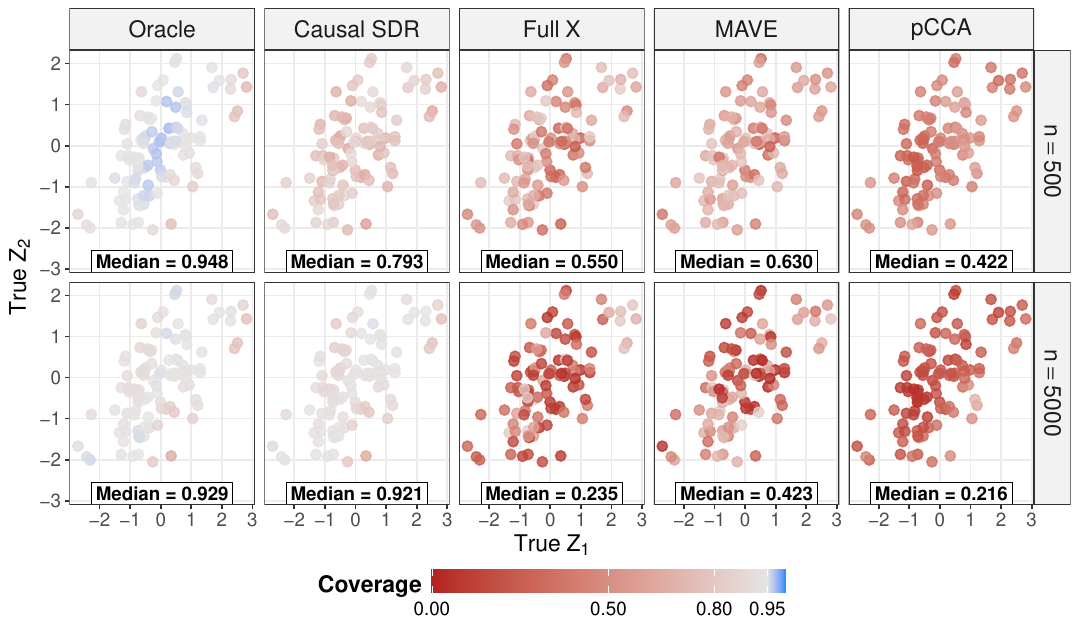}

    \vspace{-0.25cm}
    \caption{Pointwise 95\% coverage of ERS estimates across 100 evaluation points. Red is undercoverage, light gray is close to nominal 95\% level, and blue is overcoverage. The text box in each panel reports the median pointwise coverage for that specific setting. Causal SDR here refers to the RA variant (chosen due to lowest RMSE in Table~\ref{tab:rmse_ers}).}
    \label{fig:sim_pointwise_coverage}
\end{figure}

As shown in Figure~\ref{fig:sim_pointwise_coverage}, the causal SDR estimators achieved substantially improved pointwise 95\% coverage relative to the baseline methods. At $n=5000$, the RA variant attained coverage comparable to the Oracle. In contrast, Full $X$ exhibited notably low coverage, likely due to the increased difficulty of estimating the nuisance function $m(x,c)$ in the full covariate space, as opposed to the reduced representation $Z$. This may also explain its inferior coverage relative to MAVE and pCCA, which reduced the dimension even when they did not estimate the 
$\CCMS$ well. Alternatively, the undercoverage of Full $X$ may reflect poor fit in the neural network learner. In Appendix~\ref{subsec:add_sim_superlearner}, we assessed sensitivity to choice of nuisance estimator by replacing the neural network with a SuperLearner ensemble. This improved coverage of Full $X$ relative to the main simulation, but it still remained below that of the CSDR methods. In Appendix~\ref{subsec:add_sim_nonsparse}, we also considered the same DGP as in the main simulation study but with a less sparse $\beta$ where the two structural directions were linear combinations of $X_1$--$X_4$ and $X_7$--$X_{10}$. CSDR continued to outperform the baseline methods and Full $X$, but estimation of the CCMS, dimension selection, and pointwise accuracy and 95\% coverage of the ERS were all worse than in the main simulation. These results suggest that first-stage nuisance estimation error and final ERS estimation are sensitive to the dimension reduction structure of $\CCMS$. Nevertheless, the simulations highlight the potential efficiency gains granted by dimension reduction in causal effect estimation.  

\section{Real data analysis}
\label{sec:real_data}

We applied our method to estimate the joint effect of four per- and polyfluoroalkyl substances (PFOS, PFOA, PFNA, and PFHxS) on infant birthweight in grams. Data were drawn from a cohort of African American women with term births collected between 2014 and 2020 \citep{brennanProtocolEmoryUniversity2019}. Following the analysis on the same cohort by \cite{changPolyfluoroalkylSubstancePFAS2022}, all PFAS exposure variables were $\log_2$-transformed and effects were adjusted for maternal age, body mass index (BMI), education level, tobacco use, marijuana use, and infant sex.

Given the moderate sample size of this dataset ($n=305$) compared to our simulations, standard neural networks risk overfitting. To ensure flexible estimation without specifying strict parametric forms, we utilized the SuperLearner ensemble \citep{vanderlaanSuperLearner2007}. Our library included linear regression (\texttt{glm}), elastic net (\texttt{glmnet}), generalized additive models (\texttt{gam}), multivariate adaptive regression splines (\texttt{earth}), gradient boosted trees (\texttt{xgboost}), and random forests (\texttt{ranger}). For the nuisance estimation of conditional exposure densities ($\pi$), we continued using the multivariate normal density estimator.

\begin{figure}
    \centering
    \includegraphics[width=1\textwidth]{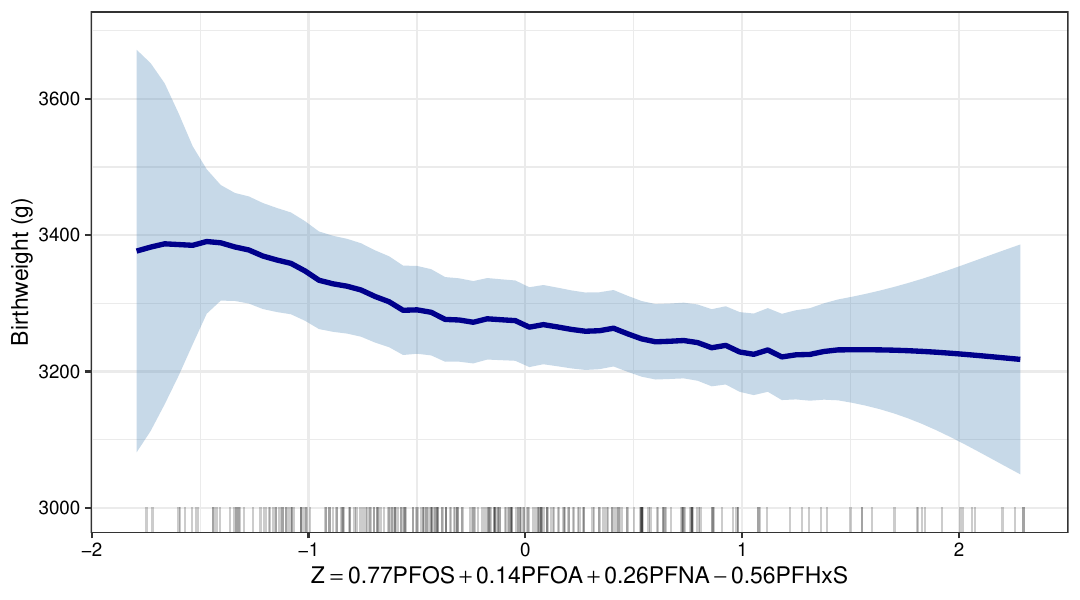}

    \vspace{-0.15cm}
    \caption{Estimated one-dimensional ERS for the effect of PFAS on birthweight ($n=305$). The density of observed $\hat{Z}\coloneqq\hat{\beta}^\top X$  is given by the rug plot.}
    \label{fig:PFAS_ERS}
\end{figure}

Our procedure estimated a structural dimension of $\hat{d}=1$, with a dimension reduction of $Z = 0.77 \times \text{PFOS} + 0.14 \times \text{PFOA} + 0.26 \times \text{PFNA} - 0.56 \times \text{PFHxS}$. In Figure~\ref{fig:PFAS_ERS}, the estimated ERS reveals a non-linear, monotonically decreasing relationship between $Z$ and birthweight that is mostly flat in the region of dense support between $Z=-0.5$ to $Z=1$. The rug plot on the bottom indicates data sparsity at the extreme ranges of $Z$, which is reflected by the wider confidence bands in those regions. 

The SIS for $\hat{\beta}$ were $\lambda_{ERS} = (0.60, 0.02, 0.07, 0.31)$ for PFOS, PFOA, PFNA, and PFHxS, respectively. This indicates that PFOS was the primary contributor to the CCMS, with PFHxS playing a secondary role acting in the opposite direction. In contrast, when applying standard MAVE directly to the exposures and outcome without our proposed causal adjustment, the resulting importance scores were $\lambda_{MAVE} = (0.17, 0.68, 0.07, 0.07)$, placing nearly all emphasis on PFOA. We also applied quantile $g$-computation \citep{keilQuantileBasedGComputationApproach2020} and BKMR \citep{bobbBayesianKernelMachine2015} to the same dataset for comparison with existing environmental mixtures methods, yielding similar overall conclusions regarding the harmful effect of maternal PFAS exposure (Appendix~\ref{sec:additional_RDA}).

\section{Discussion}
\label{sec:discussion}

We propose a new estimator for CSDR to study the causal effects of multiple continuous exposures, motivated by estimating the joint effect of environmental mixtures on health outcomes. By formalizing $\CCMS$ as the target estimand, our CSDR method provides a principled approach to dimension reduction that preserves the ERS. The method leverages existing tools for causal inference for continuous exposures and SDR for $\CMS$, substantially reducing computational and nuisance estimation burden relative to the original MSM approach. Simulation studies showed strong finite sample performance, yielding better recovery of the CCMS and more accurate estimation of the ERS using the reduced exposure compared to that of the original. We applied CSDR on an analysis of the effect of PFAS on newborn birthweight, yielding a one-dimensional exposure summary that provides an interpretable new variable for further investigation.

Our method simplifies CSDR into a modular series of steps: (1) first-stage estimation of new exposure and outcome variables, (2) second-stage SDR for $\CMS$ on the new variables to target $\CCMS$, and (3) causal ERS estimation using the newly estimated $\hat{Z}$. We provide specific theoretical guarantees for (1) ERS variant in the first stage and (2) csMAVE estimator in the second stage. However, the modular structure of our CSDR framework allows use of other estimators for these components, such as gradient-based kernel dimension reduction \citep{fukumizuGradientBasedKernelDimension2014} or stochastic MAVE \citep{pautrelRiemannianStochasticOptimization2026} for SDR, and higher-order influence function based causal dose-response estimators \citep{bonviniFastConvergenceRates2026} for ERS estimation, which may offer better statistical properties in some settings. 

Much work remains to improve causal inference for continuous exposures. \cite{zhangNonparametricInferenceDoseResponse2025} and \cite{zhangDoublyRobustInference2025} recently proposed estimators of both the ERS and its derivative for single continuous exposures that relax the positivity assumption under an additive confounding assumption. Extending these ideas to the multivariate case and more general confounding structures would be a valuable direction for future work. In addition, kernel-based identification for dose-response functions, although widely used in recent work \citep{kallusPolicyEvaluationOptimization2018, klosinAutomaticDoubleMachine2021, zhangDoublyRobustInference2025,bonviniFastConvergenceRates2026, colangeloDoubleDebiasedMachine2026}, requires smoothness conditions up to the third derivative for nuisance functions and bandwidth selection. Weakening the assumptions needed for causal inference for continuous exposures would have clear and direct benefits for CSDR.

\section*{Acknowledgments}

This work is supported by NIH grants R21ES036795 and R01ES037289. The authors thank Dr. Donghai Liang and Dr. Anne Dunlop for their valuable consultation on the applied analysis and for providing access to the Atlanta African American Maternal-Child Cohort, which was generated with support from R01ES035738, U24ES029490, and UH3OD023318. 

\section*{Data availability statement}
The code to reproduce the simulations and applied data analysis of the ATL-AA cohort are available at \url{https://github.com/tXiao95/causalpca}. The individual-level ATL-AA cohort data used in the application are not publicly available because they contain sensitive participant information and are subject to Emory University institutional, ethical, and data-use restrictions. Access to the cohort data may be requested from the ATL-AA study investigators, subject to appropriate approvals and data-use agreements. 

% ##############################################

% \newpage

\vspace{0.7cm}
\begingroup
\renewcommand{\baselinestretch}{1}
\selectfont  % tighter line spacing
\setlength{\bibsep}{11pt}    % vertical space between entries
\bibliographystyle{apalike}
\bibliography{references}
\endgroup

% ##############################################

\clearpage

\appendix

\renewcommand{\thetable}{S\arabic{table}}
\renewcommand{\thefigure}{S\arabic{figure}}
\setcounter{table}{0}  % Reset table counter to start at S1
\setcounter{figure}{0}  % Reset figure counter to start at S1
\setcounter{page}{1}
\pagenumbering{arabic}  % or roman, alph, etc.

\part*{Appendix}
\setcounter{lemma}{0}
\setcounter{theorem}{0}
\renewcommand{\thelemma}{S\arabic{lemma}}
\renewcommand{\theexample}{S\arabic{example}}
\renewcommand{\thetheorem}{S\arabic{theorem}}

\addcontentsline{toc}{part}{Appendix}
\etocsettocstyle{}{}
{\small
\localtableofcontents
}

\section{Assumptions}\label{sec:ass}

We discuss in detail the assumptions needed for our causal SDR framework by comparing them to the assumptions required for the original 1) MAVE and 2) causal ERS estimation.

\subsection{Assumptions for MAVE}\label{sec:supp_MAVE_assumption}

Listed below are the assumptions from \citeSM{wangAdaptiveEstimationMAVE2012SM}, which are similar to the assumptions stated in the original MAVE paper \citeSM{xiaAdaptiveEstimationDimension2002SM}. 

\begin{enumerate}[label=C\arabic*.]
  \item $\{(X_i, Y_i), i = 1, \ldots, n\}$ are i.i.d. samples from the density $f_{X,Y}(x,y)$.
  
  \item $\{\varepsilon_i\}$ are i.i.d. with $\mathbb{E}(\varepsilon_i) = 0$, $\mathbb{E}(|\varepsilon_i|^3) < \infty$. $\{X_i\}$ and $\{\varepsilon_i\}$ are mutually independent. Additionally, the predictor $X$ has a bounded support.
  
  \item The density $f_{\varepsilon}(\cdot)$ of $\varepsilon$ has bounded continuous derivatives up to order 4. Let $\ell(\varepsilon) = \log f_{\varepsilon}(\varepsilon)$.  Assume $\ell^{'''}(\cdot)$ is bounded and $\mathbb{E}[\ell'(\varepsilon)^2 + |\ell''(\varepsilon)| + |\ell^{'''}(\varepsilon)|] < \infty$.
  
  \item $\mathbb{E}|y|^k < \infty$ for all $k > 0$, $\mathbb{E}\|X\|^k < \infty$ for all $k > 0$.
  
  \item The density function $f_X(x)$ of $X$ has bounded derivatives up to order 4 
  and is bounded away from 0 in a neighborhood around 0.
  
  \item The density function $f_Y(\cdot)$ of $Y$ has a bounded derivative 
  and is bounded away from 0 on a compact support.
  
  \item The conditional densities $f_{X|Y}(\cdot)$ of $X$ given $Y$ and 
  $f_{(X_0, X_1)|(Y_0, Y_1)}(\cdot)$ of $(X_0, X_1)$ given $(Y_0, Y_1)$ 
  are bounded for all $l \ge 1$.
  
  \item $g(\cdot)$ has bounded, continuous third derivatives. 
  
  \item $\mathbb{E}(X \mid Y)$ and $\mathbb{E}(XX^\top \mid Y)$ have bounded, continuous third derivatives.
  
  \item $K(\cdot)$ is a spherical symmetric density function with a bounded derivative and support. All the moments of $K(\cdot)$ exist and  $\int U U^\top K(U) \, dU = I$. 
\end{enumerate}

In the proofs of Theorems~\ref{thm:csMAVE_convergence} and \ref{thm:structural_dimension_consistency}, we keep conditions C5, C8, and C10, partially keep conditions C2 and C9, and discard the rest. These are the basis for Assumption~\ref{ass:csMAVE}. Condition C1 is redundant. For Condition C2 we keep the bounded support on $X$ but remove the remaining part as it no longer applies to $\varepsilon=0$. In our case, we introduce a new additive error term $R_n$ which has no guarantee that $\E(R_n \mid X)=0$. We also do not assume $X$ and $R_n$ are independent, since $R_n$ represents the estimation error and is likely to be larger at certain values of $X$. This independence assumption was previously relaxed in \citeSM{xiaAdaptiveEstimationDimension2002SM}.  Condition C3 is not required, since the error distribution $f_\varepsilon$ is assumed to be Gaussian, as in the original MAVE procedure, and thus the boundedness and smoothness conditions on the density of $\varepsilon$ are automatically satisfied for a least-squares objective in the local linear regression. Condition C4 is redundant because Condition C2 states that $X$ has bounded support implying all moments of $X$ already exist. Likewise any response of the form $\mu(X) = g(\beta^\top X)$ is also bounded when $g$ is continuous, since $\beta^\top X$ has bounded support. This is clear by Condition C8, where $g$ is assumed to have bounded continuous third derivatives, which implies that $g$, $g'$, and $g''$ are bounded and continuous. C5 is a standard requirement for the uniform rate of consistency of the kernel regression methods used in MAVE. Conditions C6, C7, and C9 pertain to the marginal density of $Y$, the conditional density of $X$ given $Y$, and the first two moments of $X$ given $Y$, respectively. They are specific to inverse regression-based MAVE approaches focused on $\E(X\mid Y)$ and kernel smoothing of dependent data that are introduced in \citeSM{xiaAdaptiveEstimationDimension2002SM} and carried over into \citeSM{wangAdaptiveEstimationMAVE2012SM}, but are not needed in our setting. Condition C8 is critical for the third-order Taylor expansion of $g(\cdot)$ to attain the important $O(h^3)$ convergence rate of MAVE. Condition C10  is a list of standard kernel properties. The MAVE papers only allow for kernels with bounded support, such as the triangular, Epanechnikov, or truncated Gaussian kernels. We keep this condition for simplicity of the proofs, but note that a tail condition on the derivative can also be used to allow for kernels with unbounded support such as the Gaussian. 

\subsection{Assumptions for DML with continuous treatments}

We focus here on the smoothness assumptions needed to properly identify and estimate causal effects with continuous treatments. We adopt all of these assumptions in our work. Assumptions 2.1(a) and (b) in their paper are just the standard causal identification assumptions. We begin from Assumption 2.1(c). 
\begin{itemize}
  \item Assumption 2.1(c): $\pi_O(y,x,c)$ and $m(x,c)$ are three-times differentiable with respect to $x$ with all three derivatives being bounded uniformly over $(y,x^\top,c^\top)^\top \in \mathcal{O}$.
  \item Assumption 2.1(d): $\text{Var}(Y\mid X=x, C=c)$ and its derivatives with respect to $x$ are bounded uniformly over $\mathcal{X}\times \mathcal{C}$.
  \item Assumption 2.2: The second-order symmetric kernel function $k()$ is bounded differentiable, that is $\int^\infty_{-\infty} uk(u)du=0$ and $0<\kappa:=\int u^2 k(u)du<\infty$. For some finite positive constants $C, \bar{U}$, and for some $\nu>1$, $|dk(u) / du| \leq C|u|^{-\nu}$ for $|u| > \bar{U}$. 
  \item Assumption 2.3: There exist functions $\bar{m}(x,c)$ and $\bar{\pi}(x\mid c)$ that are three-times differentiable with respect to $x$ with all three derivatives being bounded uniformly over $\mathcal{X} \times \mathcal{C}$ and $\bar{\pi}(\cdot \mid \cdot)$ also satisfies strong positivity.
  \item Assumption 2.4: For $\ell=1,...,L$ and $x\in \mathcal{X}$, (a) $\lVert \hat{m}_\ell - \bar{m}\rVert_{F_{xC}}=o_p(1)$ and  $\lVert \hat{\pi}_\ell - \bar{\pi}\rVert_{F_{xC}}=o_p(1)$, (b) $\sqrt{nh^p} \lVert \hat{m}_\ell - \bar{m}\rVert_{F_{xC}}\lVert \hat{\pi}_\ell - \bar{\pi}\rVert_{F_{xC}}=o_p(1)$, (c) Either $\bar{m}=m$ or $\bar{\pi}=\pi$. 
  \item Theorem 3.1 (asymptotic normality): Let $h\to 0$, $nh^p \to \infty$, and $nh^{p+4}\to C\in [0,\infty)$. 
  \item Theorem 3.1 (bias correction): Let $\E\big[ |Y - \bar{m}(X,C) | ^3 \mid X=x, C \big]$ and its derivatives with respect to $x$ be bounded uniformly over $\mathcal{X}\times \mathcal{C}$ and $\int_{-\infty}^\infty k(u)^3 du<\infty$. 
\end{itemize}

Assumption 2.1(c) is used to justify the Taylor expansions underlying the kernel calculations. In particular, it gives the second-order bias expansion and a uniform remainder for localized quantities involving the treatment density and the outcome regression. In the paper this is used in the derivation of the asymptotic bias term \(B_x\) in Theorem 3.1 and in the kernel expansions in the supplement. Assumption 2.1(d) controls the local second moment of the kernel-weighted residual term, which determines the asymptotic variance. Although the kernel localizes in \(X\), the relevant stochastic term is conditional on both \(X\) and \(C\), because both the regression function and the generalized propensity score depend on \(C\).

Assumption 2.2 is a list of kernel properties. The derivative condition is used to control remainder terms in the tails uniformly, including kernels with unbounded support such as the Gaussian kernel. Assumptions 2.3 and 2.4 are conditions on the nuisance estimators. Assumption 2.3 guarantees that even the probability limits of the nuisance estimators are smooth enough for the same kernel expansions to go through. Assumption 2.4(a) gives consistency toward those limits in the partial norm relevant for a fixed target exposure value \(x\). Assumption 2.4(b) is the product-rate condition that makes the second-order nuisance remainder asymptotically negligible. Assumption 2.4(c) is the doubly robust requirement that at least one nuisance limit is correctly specified. All estimators are assumed to be cross-fit, avoiding any need of Donsker conditions. 

Finally, Theorem 3.1 has two layers of assumptions baked in. The first part gives a standard asymptotically linear representation of the kernel estimator that contains a $h^2$-order smoothing bias. The conditions \(h\to 0\) and \(nh^{p}\to\infty\) are standard. The additional condition \(nh^{p+4}\to C\) ensures that the smoothing bias is of the same asymptotic order as the estimator’s random sampling variability under \(\sqrt{nh^p}\)-scaling. Although the bias is not written explicitly in the first asymptotic representation, it is contained in the mean of the kernel-weighted summand on the right-hand side. Thus \(nh^{p+4}\to C\) is what makes the deterministic bias term \(h^2B_x\) asymptotically visible rather than negligible or dominant, which is necessary for an eventual bias correction. A special case is undersmoothing, where \(nh^{p+4}\to 0\), so the bias vanishes asymptotically \citepSM{wassermanAllNonparametricStatistics2006SM}. This formulation leads to the asymptotic representation
\[
\sqrt{nh^p} \big(\hat{\mu}(x) - \mu(x) \big)
=
\sqrt{\frac{h^p}{n}} \sum_{i=1}^n
\left\{
\frac{K_h(X_i-x)}{\bar{\pi}(x \mid C_i)} (Y_i - \bar{m}(x, C_i))
+ \bar{m}(x, C_i) - \mu(x)
\right\}
+ o_p(1).
\]

The second part of Theorem 3.1 adds a third-moment condition on the residual part of the estimator and the integrability condition \(\int k^3(u)\,du<\infty\). These are used to satisfy the Lyapunov condition allowing the authors to conclude asymptotic normality after subtracting the smoothing bias. This leads to
\[
\sqrt{nh^p}\big(\hat{\mu}(x) - \mu(x)-h^2B_x\big) \overset{d}{\to} N(0, V_x),
\]
for some variance \(V_x\) and bias term $B_x$. 
\clearpage

\section{Technical Lemmas}

The following lemmas are used as intermediary results to support the discussion and proofs for the main text and results in Appendix~\ref{sec:proofs}.

\begin{lemma}
    Let $A$ be an invertible matrix on a vector space $V$ and $\{\mathcal{S}_\alpha\}_{\alpha\in I}$ be a family of subspaces on $V$ indexed by $I$. Then $A^\top \cap_\alpha \{\mathcal{S}_\alpha\}=\cap_\alpha \{A^\top \mathcal{S}_\alpha\}$. 
    \label{lemma:subspace_tranpose}
\end{lemma}
\begin{proof}
    Suppose $x\in\cap_\alpha\{\mathcal{S}_\alpha\}$. Then $x\in\mathcal{S}_\alpha$ and $A^\top x\in A^\top \mathcal{S}_\alpha$ for every $\alpha$. It follows that $A^\top \cap_\alpha\{\mathcal{S}_\alpha\}\subseteq \cap_\alpha \{A^\top \mathcal{S}_\alpha\}$.   

    Now suppose $y\in\cap_\alpha \{A^\top \mathcal{S}_\alpha\}$. Then there exists a $x_\alpha$ in every $\mathcal{S}_\alpha$ such that $y=A^\top x_\alpha$. Because $A$ is invertible, $A^\top$ is injective, so that $x_\alpha=x$ for some $x\in\cap_\alpha\{\mathcal{S}_\alpha\}$. Then $y=A^\top x\in A^\top \cap_\alpha \{\mathcal{S}_\alpha\}$, completing the proof.
\end{proof}

\begin{lemma}[Pseudo-outcome]
\label{lemma:pseudo_outcome_unbiased}
    Under Assumptions~\ref{ass:identificationX} and \ref{ass:smooth}, let
\(\xi = \xi(O;\bar\pi,\bar{m})\) be defined as in \eqref{eq:pseudo_outcome}, where $\bar{\pi}(\cdot \mid \cdot)$ and $\bar{m}(\cdot, \cdot)$ may not be equal to the true $\pi(\cdot\mid \cdot)$ and $m(\cdot, \cdot)$. Then
\[
\E\{\xi(O;\bar{\pi}, \bar{m})\mid X = x)\} = \mu(x)
\]
for almost every \(x\), provided that either \(\bar\pi = \pi\) or \(\bar{m} = m\).
\end{lemma}
\begin{proof}
    This is a direct result of Section 3 in the Web Appendix of \citeSM{kennedyNonparametricMethodsDoubly2017SM}, extended to a multidimensional treatment variable $X\in\mathbb{R}^p$. Their paper is focused on a univariate treatment but the argument readily applies to $p>1$.
\end{proof}

\begin{lemma}[Robins, Thms.\ 3.2--3.3]\label{lem:robins_no_phi}
Assume the conditional density $\pi(X\mid C)$ is unknown, a MSM holds, Assumption~\ref{ass:identificationX}--~\ref{ass:well_defined_Z}, and the data $O$. Let $\mathcal H$ be a collection of measurable vector-valued functions $h$ indexing a class of estimating equations.
For each $h\in\mathcal H$, let $U_{\mathrm{sm}}(\beta,h)$ denote a regular association-model score for $\beta$. Let $W=\pi(X\mid C)/\pi^*(X)$ denote the inverse probability weight. For each $h\in\mathcal H$, define the locally efficient projected score
\[
U(\beta,h)
=
\frac{U_{\mathrm{sm}}(\beta,h)}{W}
-
\E\!\left[\frac{U_{\mathrm{sm}}(\beta,h)}{W}\Big|X,C\right]
+
\E\!\left\{\E\!\left[\frac{U_{\mathrm{sm}}(\beta,h)}{W}\Big|X,C\right]\Big|C\right\}.
\]

Let $\beta_0$ denote the true parameter and define
\[
\kappa(h)
=
-\left.\frac{\partial}{\partial\beta}\E\{U(\beta,h)\}\right|_{\beta=\beta_0}.
\]

Under regularity conditions, any estimator $\hat\beta(h)$ solving
\[
\sum_{i=1}^n U_i(\beta,h)=0
\]
is regular and asymptotically linear with influence function $\mathrm{IF}(h)=\kappa(h)^{-1}U(\beta_0,h)$. Moreover, $\{\mathrm{IF}(h):h\in\mathcal H\}$ is the set of all influence functions for $\beta_0$.
\end{lemma}

\begin{proof}
    This result is a consolidation and direct consequence of Theorems 3.1, 3.2, and 3.3 from \citeSM{robinsMarginalStructuralModels2000aSM} when the conditional density $\pi(x\mid c)$ is unknown and there is only one time point recorded as opposed to multiple timepoints or censoring times. A similar result was given in Theorem 1 of \citeSM{nabiSemiparametricCausalSufficient2022SM}. 
\end{proof}

\begin{remark}[Known exposure density]\label{rem:robins_known_density}
If instead the exposure-process density is assumed known, one may enlarge the class to $\{\hat\beta(h,\phi):h\in\mathcal H,\phi\in\Phi\}$, and then
$\{\mathrm{IF}(h,\phi):h\in\mathcal H,\phi\in\Phi\}$ is the full set of influence functions (Robins, Thm.\ 3.1).
\end{remark}

% Theorem 1 materials
\subsection{Intermediate results for Theorem~\ref{thm:csMAVE_convergence}}\label{subsec:intermediate_thm1}
The following set of lemmas all relate to the proof of Theorem~\ref{thm:csMAVE_convergence}. We first introduce some notation that will be used in both the intermediate results and main proof.
\paragraph{Notation.}
Let the exact-response be $Y_i$ and the generated-response be $\widehat Y_i=Y_i+\rho_{ni}$. Define $r_{n,\infty}=\max_{1\le i\le n}|\rho_{ni}|$. Assume $r_{n,\infty}=O_p(R_n)$ and $R_n=o(h)$. For any candidate
$
\beta\in\mathbb R^{p\times d}
$
with $\beta^\top\beta=I_d$, define
\begin{equation*}
    P_\beta=\beta\beta^\top,
    \qquad
    \Delta_\beta=(I_p-P_\beta)\beta_0=(I_p-\beta\beta^\top)\beta_0.
\end{equation*}
With slight abuse of notation, denote the $k$th column of $\Delta_\beta$ as
\begin{equation*}
    \delta_{\beta,k}=(I_p-\beta\beta^\top)\beta_{0k},
    \qquad k=1,\ldots,d_0.
\end{equation*}
For fixed $(\beta,x)$, define
\begin{equation*}
    X_{h,i}(\beta,x)
    =
    \begin{pmatrix}
        1\\
        h^{-1}\beta^\top(X_i-x)
    \end{pmatrix},
    \qquad
    K_{h,i}(x)=K_h(X_i-x)=h^{-p}K\left(\frac{X_i-x}{h}\right),
\end{equation*}
and
\begin{equation*}
    S_n(\beta,x)
    =
    \frac1n\sum_{i=1}^n
    K_{h,i}(x)X_{h,i}(\beta,x)X_{h,i}(\beta,x)^\top, \qquad
    \varsigma_n(x)=\frac1n\sum_{i=1}^nK_{h,i}(x).
\end{equation*}
We also define
\begin{equation*}
\begin{aligned}
    N_n(\beta,x)
    &=
    S_n(\beta,x)^{-1}
    \frac1n\sum_{i=1}^n
    K_{h,i}(x)X_{h,i}(\beta,x)(X_i-x)^\top, \\
    L_{n,i}(\beta,x)
    &=(X_i-x)^\top-X_{h,i}(\beta,x)^\top N_n(\beta,x).
\end{aligned}
\end{equation*}
For a generic response vector $Y=(Y_1,\ldots,Y_n)$, define the normalized profiled least-squares MAVE criterion
\begin{equation*}
    Q_n^Y(\beta)
    =
    \frac1{n^2}
    \sum_{j=1}^n\sum_{i=1}^n
    \frac{K_{h,i}(X_j)}{\varsigma_n(X_j)}
    \{Y_i-a_j-b_j^\top\beta^\top(X_i-X_j)\}^2.
\end{equation*}
\begin{assumption}[Uniform nonsingularity]
\label{ass:mave-local-moments}
Assume that $S_n(\beta,x)$ is uniformly nonsingular in the sense that
\begin{equation*}
    \sup_{\beta,x}\|S_n(\beta,x)^{-1}\|=O_p(1).
\end{equation*}
\end{assumption}

\begin{lemma}[Exact-response expansion]
\label{lem:exact-response-WY-expansion}
Let $\mathcal S_{Y,k,n}(\beta)$ denote the $k$th direction of the profiled least-squares MAVE score based on the exact response $Y_i=g(\beta_0^\top X_i)$. Then for $k=1,\ldots,d_0$,
\begin{equation}
\begin{aligned}
    \mathcal S_{Y,k,n}(\beta)
    &=
    -h^2\sum_{\ell=1}^{d_0}
    M_{k\ell,n}\delta_{\beta,\ell}
    +h^4\Delta_\beta a_{k,n}  \\
    &\quad
    +O_p\!\left((h^3+h\delta_n)\|\Delta_\beta\|\right)
    +O_p(h^5+h^3\delta_n),
\end{aligned}
\label{eq:exact-score-expansion}
\end{equation}
where
\begin{equation*}
    M_{k\ell,n}
    =
    \frac1n\sum_{j=1}^n
    V_k(\beta_0^\top X_j)V_\ell(\beta_0^\top X_j),
\end{equation*}
and $a_{k,n}\in\mathbb R^{d_0}$ satisfies $\|a_{k,n}\|=O_p(1)$.
\end{lemma}

\begin{proof}[Proof Sketch]
This is the least-squares MAVE expansion obtained by \citeSM{wangAdaptiveEstimationMAVE2012SM}, which is a generalization of \citeSM{xiaAdaptiveEstimationDimension2002SM}. For $X_i$ in a kernel neighborhood of $x$, Taylor expansion and Assumption~\ref{ass:csMAVE}(b) give
\begin{equation}
    g(\beta_0^\top X_i)
    =g(\beta_0^\top x)
    +(X_i-x)^\top\beta_0V(\beta_0^\top x)
    +P_{2,i}(x)+P_{3,i}(x)+R_i(x),
\label{eq:taylor-g}
\end{equation}
where $P_{2,i}(x)=O(h^2)$, $P_{3,i}(x)=O(h^3)$, and the Peano remainder satisfies $R_i(x)=o(h^3)$ uniformly on the kernel support because $g$ has bounded continuous third derivatives. Using
\begin{equation*}
    \beta_0=\beta\beta^\top\beta_0+(I_p-\beta\beta^\top)\beta_0,
\end{equation*}
the part of the first-order term in $\operatorname{span}(\beta)$ is absorbed by the local fitted slope, while the residual part contributes
\begin{equation*}
    (X_i-x)^\top(I_p-\beta\beta^\top)\beta_0V(\beta_0^\top x).
\end{equation*}
After profiling out the local intercept and slope, the term $X_i-x$ is residualized into $L_{n,i}(\beta,x)$. This is the role of $L_{n,i}$, $Q_{n,i}$, $\Xi_{n,i}$, and $R_{n,i}$ in Appendix B of \citeSM{wangAdaptiveEstimationMAVE2012SM}. 

Substituting the profiled residual into the $k$th coordinate of the profiled score yields the leading curvature term $-\sum_{\ell=1}^{d_0}\delta_{\beta,\ell}h^2M_{k\ell,n}$. The second- and third-order Taylor terms yield the bias contribution $h^4\Delta_\beta a_{k,n}$. The empirical and local smoothing remainder terms are bounded by $O_p(h^5+h^3\delta_n)$,
following the same procedure in equations (B.10)--(B.14) of \citeSM{wangAdaptiveEstimationMAVE2012SM}. However, the main difference between standard MAVE and csMAVE is that our choice of exact-response $Y=\mu(X)$ is a deterministic function of $X$. Therefore, all error terms involving $\varepsilon_i$ and $\Xi_{n,i}$ which were the source of the $h^{-1}\delta_n^2$ term in the subspace rate, are absent here. This gives \eqref{eq:exact-score-expansion}.
\end{proof}
\begin{lemma}[Profiling identities]
\label{lem:profiling-residualization}
(i) For any response vector $Y=(Y_1,\ldots,Y_n)$, define the local linear regression intercept and slope estimates for fixed $(\beta,x)$ as
\begin{equation*}
    \widehat\gamma_Y(\beta,x)
    =
    \begin{pmatrix}
        \widehat a_Y(\beta,x)\\
        h\widehat b_Y(\beta,x)
    \end{pmatrix}.
\end{equation*}
Then
\begin{equation}
    \widehat\gamma_Y(\beta,x)
    =
    S_n(\beta,x)^{-1}
    \frac1n\sum_{i=1}^n
    K_{h,i}(x)X_{h,i}(\beta,x)Y_i.
\label{eq:gamma-Z}
\end{equation}
(ii) Recall
\begin{equation*}
\begin{aligned}
    N_n(\beta,x)
    &=
    S_n(\beta,x)^{-1}
    \frac1n\sum_{i=1}^n
    K_{h,i}(x)X_{h,i}(\beta,x)(X_i-x)^\top, \\
    L_{n,i}(\beta,x)
    &=(X_i-x)^\top-X_{h,i}(\beta,x)^\top N_n(\beta,x).
\end{aligned}
\end{equation*}
Then
\begin{equation}
    \frac1n\sum_{i=1}^n
    K_{h,i}(x)L_{n,i}(\beta,x)^\top X_{h,i}(\beta,x)^\top=0.
\label{eq:L-orthogonality-final}
\end{equation}
(iii) Furthermore, if $\widehat Y_i=Y_i+\rho_{ni}$, then
\begin{equation}
    e_{n,i}^{\widehat Y}(\beta,x)
    =e_{n,i}^{Y}(\beta,x)+\rho_{ni}-\Phi_{n,i}(\beta,x),
\label{eq:residual-decomposition-final}
\end{equation}
where
\begin{equation*}
    e_{n,i}^{Y}(\beta,x)
    =Y_i-X_{h,i}(\beta,x)^\top\widehat\gamma_Y(\beta,x),
\end{equation*}
\begin{equation*}
    \Phi_{n,i}(\beta,x)
    =X_{h,i}(\beta,x)^\top\widehat\gamma_\rho(\beta,x),
\end{equation*}
and
\begin{equation*}
    \widehat\gamma_\rho(\beta,x)
    =S_n(\beta,x)^{-1}
    \frac1n\sum_{i=1}^n
    K_{h,i}(x)X_{h,i}(\beta,x)\rho_{ni}.
\end{equation*}
\end{lemma}

\begin{proof}
The local fitted value is
\begin{equation*}
    \widehat a_Y(\beta,x)
    +\widehat b_Y(\beta,x)^\top\beta^\top(X_i-x)
    =X_{h,i}(\beta,x)^\top\widehat\gamma_Y(\beta,x).
\end{equation*}
Then $\widehat\gamma_Y(\beta,x)$ minimizes
\begin{equation*}
    \frac1n\sum_{i=1}^n
    K_{h,i}(x)
    \{Y_i-X_{h,i}(\beta,x)^\top\gamma\}^2.
\end{equation*}
The weighted least-squares normal equations give \eqref{eq:gamma-Z}.

By the definition of $L_{n,i}(\beta,x)$,
\begin{equation*}
\begin{aligned}
    &\frac1n\sum_iK_{h,i}(x)X_{h,i}(\beta,x)L_{n,i}(\beta,x) \\
    &=\frac1n\sum_iK_{h,i}(x)X_{h,i}(\beta,x)(X_i-x)^\top
    -S_n(\beta,x)N_n(\beta,x)
    =0.
\end{aligned}
\end{equation*}
Taking transposes leads to \eqref{eq:L-orthogonality-final}.

Finally, \eqref{eq:gamma-Z} shows that local least squares is linear in the response. Therefore
\begin{equation*}
    \widehat\gamma_{\widehat Y}(\beta,x)
    =\widehat\gamma_Y(\beta,x)+\widehat\gamma_\rho(\beta,x).
\end{equation*}
Substituting this into the definition of the profiled residual gives \eqref{eq:residual-decomposition-final}.
\end{proof}

\begin{lemma}[Generated-response local-fit bounds]
\label{lem:local-fit-bounds}
Under Assumption~\ref{ass:mave-local-moments} and $r_{n,\infty}=O_p(R_n)$,
\begin{equation}
    \sup_{\beta,x}\|\widehat\gamma_\rho(\beta,x)\|=O_p(R_n),
\label{eq:gamma-rho-rate}
\end{equation}
\begin{equation}
    \sup_{\beta,x}\|\widehat b_\rho(\beta,x)\|=O_p(R_n/h),
\label{eq:b-rho-rate}
\end{equation}
and, on the kernel support,
\begin{equation}
    \sup_{\substack{i,\beta,x:\ K_{h,i}(x)\ne0}}
    |\rho_{ni}-\Phi_{n,i}(\beta,x)|=O_p(R_n).
\label{eq:rho-minus-Phi-rate}
\end{equation}
\end{lemma}

\begin{proof}
If $K_{h,i}(x)\ne0$, bounded support of $K$ by Assumption~\ref{ass:csMAVE}(c) implies $\|X_i-x\|\le Ch$. Hence
\begin{equation*}
    \|X_{h,i}(\beta,x)\|\le C
\end{equation*}
on the kernel support. By Lemma~\ref{lem:profiling-residualization},
\begin{equation*}
    \widehat\gamma_\rho(\beta,x)
    =S_n(\beta,x)^{-1}
    \frac1n\sum_iK_{h,i}(x)X_{h,i}(\beta,x)\rho_{ni}.
\end{equation*}
Using Assumption~\ref{ass:mave-local-moments} and $r_{n,\infty}=O_p(R_n)$,
\begin{equation*}
\begin{aligned}
    \|\widehat\gamma_\rho(\beta,x)\|
    &\le
    \|S_n(\beta,x)^{-1}\|
    \left\|
    \frac1n\sum_iK_{h,i}(x)X_{h,i}(\beta,x)\rho_{ni}
    \right\| \\
    &\le
    C\|S_n(\beta,x)^{-1}\|r_{n,\infty}
    \frac1n\sum_i|K_{h,i}(x)|
    =O_p(R_n).
\end{aligned}
\end{equation*}
This proves \eqref{eq:gamma-rho-rate}. Since
$
\widehat\gamma_\rho=(\widehat a_\rho,h\widehat b_\rho^\top)^\top,
$
\eqref{eq:b-rho-rate} follows. Finally, on the kernel support,
\begin{equation*}
    |\Phi_{n,i}(\beta,x)|
    \le
    \|X_{h,i}(\beta,x)\|\,\|\widehat\gamma_\rho(\beta,x)\|
    =O_p(R_n),
\end{equation*}
which combined with the triangle inequality gives \eqref{eq:rho-minus-Phi-rate}.
\end{proof}
\begin{lemma}[Generated-response score perturbation]
\label{lem:score-perturbation}
The difference between the $k$th direction of the generated-response score and that of the exact-response score follows
\begin{equation}
    \mathcal S_{\widehat Y,k,n}(\beta)-\mathcal S_{Y,k,n}(\beta)
    =
    O_p(hR_n)+o_p(h^2)\|\Delta_\beta\|,
    \qquad k=1,\ldots,d_0.
\label{eq:score-perturbation-final}
\end{equation}
\end{lemma}

\begin{proof}
For the exact response, the profiled least-squares objective is
\begin{equation*}
    Q_n^Y(\beta)
    =
    \frac1{n^2}
    \sum_{j=1}^n\sum_{i=1}^n
    \frac{K_{h,i}(X_j)}{\varsigma_n(X_j)}
    \{e_{n,i}^{Y}(\beta,X_j)\}^2,
\end{equation*}
For the generated response $\widehat Y_i=Y_i+\rho_{ni}$,
\begin{equation*}
    Q_n^{\widehat Y}(\beta)
    =
    \frac1{n^2}
    \sum_{j=1}^n\sum_{i=1}^n
    \frac{K_{h,i}(X_j)}{\varsigma_n(X_j)}
    \{e_{n,i}^{\widehat Y}(\beta,X_j)\}^2.
\end{equation*}

Following Lemma~\ref{lem:profiling-residualization}, we can write the profiled residual as
\begin{equation}
    e_{n,i}^{Y}(\beta,x)
    =
    L_{n,i}(\beta,x)\Delta_\beta V(\beta_0^\top x)
    +
    A_{n,i}(\beta,x),
\label{eq:S7-exact-residual-A}
\end{equation}
where
\begin{equation*}
\begin{aligned}
    A_{n,i}(\beta,x)
    &=
    h^2\bigl\{P_{2,h,i}(x)+hP_{3,h,i}(x)\bigr\}
    +
    R_i(x)  \\
    &\quad-
    X_{h,i}(\beta,x)^\top
    S_n(\beta,x)^{-1}
    \frac1n\sum_{\ell=1}^n
    K_{h,\ell}(x)X_{h,\ell}(\beta,x)  \\
    &\qquad\qquad\times
    \left[
        h^2\bigl\{P_{2,h,\ell}(x)+hP_{3,h,\ell}(x)\bigr\}
        +
        R_\ell(x)
    \right].
\end{aligned}
\end{equation*}
The Taylor-remainder expression inside the brackets is $O(h^2)$ uniformly on the kernel support. Since $X_{h,i}(\beta,x)=O(1)$ and $\sup_{\beta,x}\|S_n(\beta,x)^{-1}\|=O_p(1)$ by Assumption~\ref{ass:mave-local-moments},
\begin{equation*}
    A_{n,i}(\beta,x)=O_p(h^2)
\end{equation*}
Combining this with
$\|L_{n,i}(\beta,X_j)\|=O_p(h)$ we have
\begin{equation}
\begin{aligned}
    e_{n,i}^{Y}(\beta,X_j)
    &=
    L_{n,i}(\beta,X_j)\Delta_\beta V(\beta_0^\top X_j)
    +
    O_p(h^2)  \\
    &=
    O_p\bigl(h\|\Delta_\beta\|\bigr)+O_p(h^2).
\end{aligned}
\label{eq:S7-exact-residual-rate}
\end{equation}
By \eqref{eq:residual-decomposition-final}, the generated-response residual is
\begin{equation}
\begin{aligned}
    e_{n,i}^{\widehat Y}(\beta,x)
    &=
    L_{n,i}(\beta,x)\Delta_\beta V(\beta_0^\top x)
    +
    A_{n,i}(\beta,x)  \\
    &\quad+
    \rho_{ni}
    -
    \Phi_{n,i}(\beta,x).
\end{aligned}
\label{eq:S7-generated-residual-A}
\end{equation}
Substituting \eqref{eq:S7-generated-residual-A} into $Q_n^{\widehat Y}(\beta)$ gives
\begin{equation*}
\begin{aligned}
    Q_n^{\widehat Y}(\beta)
    &=
    \frac1{n^2}
    \sum_{j=1}^n\sum_{i=1}^n
    \frac{K_{h,i}(X_j)}{\varsigma_n(X_j)}
    \Big[
        L_{n,i}(\beta,X_j)\Delta_\beta V(\beta_0^\top X_j) \\
    &\qquad\qquad\qquad\qquad
        +
        A_{n,i}(\beta,X_j)
        +
        \rho_{ni}
        -
        \Phi_{n,i}(\beta,X_j)
    \Big]^2 .
\end{aligned}
\end{equation*}

We now take the score following \citeSM{wangAdaptiveEstimationMAVE2012SM} equations (B.6) and (B.7). Let
\begin{equation*}
    \mathcal D_{Y,k,n,i}(\beta,X_j)
    \coloneqq
    -\frac{\partial e_{n,i}^{Y}(\beta,X_j)}{\partial\delta_{k}},
    \qquad
    \mathcal D_{\widehat Y,k,n,i}(\beta,X_j)
    \coloneqq
    -\frac{\partial e_{n,i}^{\widehat Y}(\beta,X_j)}
    {\partial\delta_{k}},
\end{equation*}
Then
\begin{equation*}
\begin{aligned}
    \mathcal S_{Y,k,n}(\beta)
    &=
    \frac1{n^2}
    \sum_{j=1}^n\sum_{i=1}^n
    \frac{K_{h,i}(X_j)}{\varsigma_n(X_j)}
    e_{n,i}^{Y}(\beta,X_j)
    \mathcal D_{Y,k,n,i}(\beta,X_j),
\end{aligned}
\end{equation*}
and
\begin{equation*}
\begin{aligned}
    \mathcal S_{\widehat Y,k,n}(\beta)
    &=
    \frac1{n^2}
    \sum_{j=1}^n\sum_{i=1}^n
    \frac{K_{h,i}(X_j)}{\varsigma_n(X_j)}
    e_{n,i}^{\widehat Y}(\beta,X_j)
    \mathcal D_{\widehat Y,k,n,i}(\beta,X_j).
\end{aligned}
\end{equation*}
We know from least squares that
\begin{equation*}
    \mathcal D_{Y,k,n,i}(\beta,X_j)
    =
    \widehat b_{Y,k}(\beta,X_j)
    L_{n,i}(\beta,X_j)^\top.
\end{equation*}
We decompose this derivative as
\begin{equation}
\begin{aligned}
    \mathcal D_{Y,k,n,i}(\beta,X_j)
    &=
    V_k(\beta_0^\top X_j)
    L_{n,i}(\beta,X_j)^\top
    +
    \mathcal B^Y_{k,n,i}(\beta,X_j),
\end{aligned}
\label{eq:S7-DY-decomposition}
\end{equation}
where
\begin{equation}
    \mathcal B^Y_{k,n,i}(\beta,X_j)
    \coloneqq
    \bigl\{
        \widehat b_{Y,k}(\beta,X_j)-V_k(\beta_0^\top X_j)
    \bigr\}
    L_{n,i}(\beta,X_j)^\top .
\label{eq:S7-BY-def}
\end{equation}
Similarly, since
$\widehat\gamma_{\widehat Y}=\widehat\gamma_Y+\widehat\gamma_\rho$,
\begin{equation*}
    \widehat b_{\widehat Y,k}(\beta,X_j)
    =
    \widehat b_{Y,k}(\beta,X_j)
    +
    \widehat b_{\rho,k}(\beta,X_j).
\end{equation*}
Therefore
\begin{equation}
    \mathcal D_{\widehat Y,k,n,i}(\beta,X_j)
    =
    \mathcal D_{Y,k,n,i}(\beta,X_j)
    +
    \mathcal B^\rho_{k,n,i}(\beta,X_j),
\label{eq:S7-DYhat-decomposition}
\end{equation}
where
\begin{equation}
    \mathcal B^\rho_{k,n,i}(\beta,X_j)
    \coloneqq
    \widehat b_{\rho,k}(\beta,X_j)
    L_{n,i}(\beta,X_j)^\top .
\label{eq:S7-Brho-def}
\end{equation}

Subtracting the two scores and using \eqref{eq:residual-decomposition-final}, \eqref{eq:S7-DY-decomposition}, and \eqref{eq:S7-DYhat-decomposition} decomposes the difference as
\begin{equation}
\begin{aligned}
    &\mathcal S_{\widehat Y,k,n}(\beta)-\mathcal S_{Y,k,n}(\beta)\\
    &=
    H_{k,n}(\beta)
    +
    R^Y_{k,n}(\beta)
    +
    R^{\rho,1}_{k,n}(\beta)
    +
    R^{\rho,2}_{k,n}(\beta),
\end{aligned}
\label{eq:S7-score-four-term-decomp}
\end{equation}
where
\begin{equation}
\begin{aligned}
    H_{k,n}(\beta)
    &=
    \frac1{n^2}
    \sum_{j=1}^n\sum_{i=1}^n
    \frac{K_{h,i}(X_j)}{\varsigma_n(X_j)}
    V_k(\beta_0^\top X_j)
    L_{n,i}(\beta,X_j)^\top  \\
    &\qquad\qquad\times
    \bigl\{\rho_{ni}-\Phi_{n,i}(\beta,X_j)\bigr\},
\end{aligned}
\label{eq:S7-Hkn-def}
\end{equation}
\begin{equation}
\begin{aligned}
    R^Y_{k,n}(\beta)
    &=
    \frac1{n^2}
    \sum_{j=1}^n\sum_{i=1}^n
    \frac{K_{h,i}(X_j)}{\varsigma_n(X_j)}
    \bigl\{\rho_{ni}-\Phi_{n,i}(\beta,X_j)\bigr\}
    \mathcal B^Y_{k,n,i}(\beta,X_j),
\end{aligned}
\label{eq:S7-RY-def}
\end{equation}
\begin{equation}
\begin{aligned}
    R^{\rho,1}_{k,n}(\beta)
    &=
    \frac1{n^2}
    \sum_{j=1}^n\sum_{i=1}^n
    \frac{K_{h,i}(X_j)}{\varsigma_n(X_j)}
    e_{n,i}^{Y}(\beta,X_j)
    \mathcal B^\rho_{k,n,i}(\beta,X_j),
\end{aligned}
\label{eq:S7-Rrho1-def}
\end{equation}
and
\begin{equation}
\begin{aligned}
    R^{\rho,2}_{k,n}(\beta)
    &=
    \frac1{n^2}
    \sum_{j=1}^n\sum_{i=1}^n
    \frac{K_{h,i}(X_j)}{\varsigma_n(X_j)}
    \bigl\{\rho_{ni}-\Phi_{n,i}(\beta,X_j)\bigr\}
    \mathcal B^\rho_{k,n,i}(\beta,X_j).
\end{aligned}
\label{eq:S7-Rrho2-def}
\end{equation}

We first bound \eqref{eq:S7-Hkn-def}. By \eqref{eq:L-orthogonality-final} from Lemma~\ref{lem:profiling-residualization}, for fixed $x$,
\begin{equation*}
\begin{aligned}
    &\frac1n\sum_{i=1}^n
    K_{h,i}(x)
    L_{n,i}(\beta,x)^\top
    \Phi_{n,i}(\beta,x)  \\
    &=
    \left\{
        \frac1n\sum_{i=1}^n
        K_{h,i}(x)
        L_{n,i}(\beta,x)^\top
        X_{h,i}(\beta,x)^\top
    \right\}
    \widehat\gamma_\rho(\beta,x)
    =
    0.
\end{aligned}
\end{equation*}
Therefore
\begin{equation*}
\begin{aligned}
    H_{k,n}(\beta)
    &=
    \frac1{n^2}
    \sum_{j=1}^n\sum_{i=1}^n
    \frac{K_{h,i}(X_j)}{\varsigma_n(X_j)}
    V_k(\beta_0^\top X_j)
    L_{n,i}(\beta,X_j)^\top
    \rho_{ni}.
\end{aligned}
\end{equation*}
By boundedness of $V_k$, bounded kernel support, Assumption~\ref{ass:mave-local-moments},
and $r_{n,\infty}=O_p(R_n)$,
\begin{equation}
    H_{k,n}(\beta)=O_p(hR_n).
\label{eq:S7-Hkn-rate}
\end{equation}

Next consider \eqref{eq:S7-RY-def}. For fixed $X_j$, $\mathcal B^Y_{k,n,i}(\beta,X_j)$ is $\{\widehat b_{Y,k}(\beta,X_j)-V_k(\beta_0^\top X_j)\}L_{n,i}(\beta,X_j)^\top$, where the scalar factor is constant in the inner sum over $i$. Hence its $\Phi$ part cancels by \eqref{eq:L-orthogonality-final} as well. Since $\widehat b_{Y,k}(\beta,X_j)$ and $V_k(\beta_0^\top X_j)$ are uniformly bounded on the local neighborhood, $\|\mathcal B^Y_{k,n,i}(\beta,X_j)\|=O_p(h)$. Therefore by \eqref{eq:rho-minus-Phi-rate} in Lemma~\ref{lem:profiling-residualization}
\begin{equation}
    R^Y_{k,n}(\beta)=O_p(hR_n).
\label{eq:S7-RY-rate}
\end{equation}

It remains to bound the final two terms \eqref{eq:S7-Rrho1-def} and \eqref{eq:S7-Rrho2-def}. By Lemma~\ref{lem:local-fit-bounds},
\begin{equation*}
    \widehat b_\rho(\beta,x)=O_p(R_n/h),
    \qquad
    \rho_{ni}-\Phi_{n,i}(\beta,x)=O_p(R_n).
\end{equation*}
Together with $\|L_{n,i}(\beta,X_j)\|=O_p(h)$,
\begin{equation}
    \|\mathcal B^\rho_{k,n,i}(\beta,X_j)\|=O_p(R_n).
\label{eq:S7-Brho-rate}
\end{equation}
Using \eqref{eq:S7-exact-residual-rate} and
\eqref{eq:S7-Brho-rate},
\begin{equation}
\begin{aligned}
    R^{\rho,1}_{k,n}(\beta)
    &=
    O_p(R_n)
    \Bigl[
        O_p\bigl(h\|\Delta_\beta\|\bigr)+O_p(h^2)
    \Bigr] \\
    &=
    O_p\bigl(hR_n\|\Delta_\beta\|\bigr)
    +
    O_p(h^2R_n).
\end{aligned}
\label{eq:S7-Rrho1-rate}
\end{equation}
Similarly,
\begin{equation}
    R^{\rho,2}_{k,n}(\beta)
    =
    O_p(R_n)\,O_p(R_n)
    =
    O_p(R_n^2).
\label{eq:S7-Rrho2-rate}
\end{equation}
Since $R_n=o(h)$, combining \eqref{eq:S7-Rrho1-rate} and \eqref{eq:S7-Rrho2-rate} gives
\begin{equation*}
\begin{aligned}
    &O_p\bigl(hR_n\|\Delta_\beta\|\bigr)
    +
    O_p(h^2R_n)
    +
    O_p(R_n^2) \\
    &\qquad=
    o_p(h^2)\|\Delta_\beta\|+O_p(hR_n).
\end{aligned}
\end{equation*}
Putting it together with the first \eqref{eq:S7-Hkn-rate} and second term \eqref{eq:S7-RY-rate} in the decomposition \eqref{eq:S7-score-four-term-decomp} gives
\begin{equation*}
    \mathcal S_{\widehat Y,k,n}(\beta)-\mathcal S_{Y,k,n}(\beta)
    =
    O_p(hR_n)+o_p(h^2)\|\Delta_\beta\|,
    \qquad k=1,\ldots,d_0.
\end{equation*}
\end{proof}
\begin{lemma}[Generated-response score expansion]
\label{lem:generated-WY-expansion}
For $k=1,\ldots,d_0$,
\begin{equation}
\begin{aligned}
    \mathcal S_{\widehat Y,k,n}(\beta)
    &=
    -h^2\sum_{\ell=1}^{d_0}
    M_{k\ell,n}\delta_{\beta,\ell}
    +
    h^4\Delta_\beta a_{k,n} \\
    &\quad
    +
    O_p\{(h^3+h\delta_n)\|\Delta_\beta\|\}
    +
    O_p(h^5+h^3\delta_n+hR_n) \\
    &\quad
    +
    o_p(h^2)\|\Delta_\beta\|.
\end{aligned}
\label{eq:generated-score-expansion}
\end{equation}
\end{lemma}

\begin{proof}
By definition,
\begin{equation*}
    \mathcal S_{\widehat Y,k,n}(\beta)
    =
    \mathcal S_{Y,k,n}(\beta)
    +
    \{\mathcal S_{\widehat Y,k,n}(\beta)-\mathcal S_{Y,k,n}(\beta)\}.
\end{equation*}
Substituting results from Lemma~\ref{lem:exact-response-WY-expansion} and  Lemma~\ref{lem:score-perturbation} yields \eqref{eq:generated-score-expansion}.
\end{proof}

\section{Proofs of Main Results}\label{sec:proofs}

\subsection{Proof of Corollary~\ref{cor:identification_Z}}\label{sec:proof_cor_identification_Z}
\begin{proof}
Identification of $\mu_\beta(z):=\E(Y^z)$ by formulas \eqref{eq:identification_gcomp} and \eqref{eq:identification_ipw} immediately follows if $Z$ satisfies consistency, weak positivity, and conditional ignorability. Together, Assumptions~\ref{ass:identificationX}(a) and \ref{ass:well_defined_Z} imply $Z$ is consistent for $Y$, and $Y^z$ is a well-defined potential outcome. Assumption~\ref{ass:strong_positivity_Z}(b) is strong positivity of $Z$, which implies weak positivity. It remains to show that Assumption~\ref{ass:identificationX}(c), conditional ignorability of $X$, implies $Y^z \indep Z \mid C$ for all $z\in \mathcal{Z}$.

It is known that for any measurable function $h$, 
\begin{equation*}
    Y^x \indep X \mid C \Rightarrow Y^x \indep h(X) \mid C.
\end{equation*}
By setting $Z:=h(X)=\beta^\top X$, it follows that:
\begin{equation*}
    Y^x \indep Z \mid C.
\end{equation*}
By Assumption~\ref{ass:well_defined_Z}, for any $z\in\mathcal{Z}$, we can choose some $x_z$ with $\beta^\top x_z = z$ such that $Y^z=Y^x$, where $x=x_z$. For convenience, denote the potential outcome $Y(do(X=x_z))$ as $Y^{x_z}$. For any measurable set $A$: 
\begin{equation*}
\Prob(Y^z\in A\mid Z, C) = \Prob(Y^{x_z}\in A \mid Z, C) = \Prob(Y^{x_z} \in A\mid C)=\Prob(Y^z\in A \mid C),
\end{equation*}
which is the desired conditional ignorability for $Z$. 
\end{proof}

% \clearpage

\subsection{Proof of Proposition \ref{prop:affine}} 
\label{sec:proofs_prop:affine}

\begin{proof}
Let $U=AX+b$. We will interpret interventions on $U$ as implemented through the corresponding intervention on $X$, so that $U$ and $X$ are the same intervention expressed in different coordinates. Using do-notation for the potential outcomes to clarify the two distinct interventions, we define the potential outcome under $U$ as $Y(do(U=u)):=Y(do(X=A^{-1}(u-b))$. Let $\mu_U(u):=\E(Y(do(U=u))$. It follows that for all $u$ in the support of $U$,
\begin{equation*}
\mu_U(u)=\mu\big(A^{-1}(u-b)\big).
\end{equation*}

We will use a similar argument to Theorem 8.3 of \citeSM{liSufficientDimensionReduction2018SM}. We define the following sets of causal SDR subspaces. Let any basis matrix $\beta\in\mathbb{R}^{p\times d}$ be full column rank. 
\[
\begin{aligned}
\Apo{x} &= \{\text{span}(\beta): \mu(x)=g(\beta^\top x) \text{for some } g\}, \\
\Apo{u} &= \{\text{span}(\tilde{\beta}): \mu_U(u)=\tilde{g}(\tilde{\beta}^\top u) \text{for some } \tilde{g}\}, \\
A^\top\Apo{u} &= \{A^\top\mathcal{S}:\mathcal{S} \in \Apo{u}\}.
\end{aligned}
\]
We now show that $\Apo{x}=A^\top\Apo{u}$. Suppose $\mathcal{S}\in \Apo{x}$. By definition there exists a basis matrix $\beta$ of $\mathcal{S}$ such that
\begin{align*}
    \mu(x)&=g(\beta^\top x) \\
          &= g\big(\beta^\top A^{-1}(Ax+b)-\beta^\top A^{-1}b\big) \\
          &= g\big(\beta^\top A^{-1}u-\beta^\top A^{-1}b\big) \\
          &= \tilde{g}\big(\beta^\top A^{-1}u\big) \\
          &= \tilde{g}\big((A^{-\top}\beta)^\top u\big), 
\end{align*}
where $\tilde{g}$ is a function accounting for the constant shift from $\beta^\top A^{-1}b$. Since $x=A^{-1}(u-b)$, it follows that $\mu_U(u)=\mu(x)$ and $A^{-\top}\beta$ is a basis matrix for a subspace in $\Apo{u}$. Then $\mathcal{S}\in A^\top\Apo{u}$. 

A similar derivation in the reverse direction shows that if $\mathcal{S}\in A^\top \Apo{u}$, then $\mathcal{S}\in \Apo{x}$, and we conclude that $\Apo{x}=A^\top\Apo{u}$. Rewriting each set, 
\begin{equation}
    \{\mathcal{S}:\mathcal{S}\in \Apo{x}\}=A^\top \{\mathcal{S}':\mathcal{S}'\in \Apo{u}\}=\{A^\top\mathcal{S}':\mathcal{S}'\in \Apo{u}\}.
    \label{eq:prop1_sets}
\end{equation}
Taking the intersection on both sides of \eqref{eq:prop1_sets} and using Lemma \ref{lemma:subspace_tranpose},
\[
\begin{aligned}
    \mathcal{S}_{\E(Y^x)}&=\cap\{\mathcal{S}:\mathcal{S}\in \Apo{x}\} \\
                         &=\cap\{A^\top\mathcal{S}':\mathcal{S}'\in \Apo{u}\}\\
                         &=A^\top\cap\{\mathcal{S}':\mathcal{S}'\in \Apo{u}\} \\
                         &=A^\top \mathcal{S}_{\E(Y^{Ax+b})}.
\end{aligned}
\]

\end{proof}

% \clearpage

\subsection{Proof of Proposition~\ref{prop:FWL}}
\label{sec:proofs_prop:FWL}

\begin{proof}
The residualized forms of both $Y$ and $X$ are:
\[
\begin{aligned}
\tilde{X}&=X-\E(X\mid C) = [f(C) + \varepsilon_X] - f(C)=\varepsilon_X, \\
\tilde{Y}&=Y - \E(Y\mid C)= [g+h(C)+\varepsilon_Y] - \E(g \mid C) - h(C) = g-\E(g\mid C)+\varepsilon_Y.
\end{aligned}
\]
By definition, $g(\beta^\top X)=g\Big(\beta^\top \big( f(C)+\varepsilon_X \big)\Big)$. Denote
\begin{equation*}
    T=\beta^\top f(C), \hspace{2em} U=\beta^\top \varepsilon_X=\beta^\top \tilde{X}.
\end{equation*}
therefore $g(\beta^\top X) = g(T+U)$. It follows that $T \indep U$. By independence, 
\[
\E( g \mid C)=\E(g(T+U) \mid T=t)=\E_U(g(t+U))=\psi(t).
\]
Therefore $\tilde{Y}=g(T+U)-\psi(T)+\varepsilon_Y$ and $\tilde{X}=\varepsilon_X$. Taking the conditional expectation,
\[
\begin{aligned}
    \E(\tilde{Y} \mid \tilde{X}) &= \E( g(T+U) - \psi(T) + \varepsilon_Y \mid \varepsilon_X) \\
    &=  \E(g(T+U) - \psi(T) + \varepsilon_Y \mid U=u) \\
    &= \E(g(T+U) \mid U=u) - \E(\psi(T)) - \E(\varepsilon_Y \mid U) \\
    &= \phi(u)-c+0 = \tilde{g}(u) \\
    &= \tilde{g}(\beta_0^\top \tilde{x}).
\end{aligned}
\]
\end{proof}

\subsection{Proof of Theorem \ref{thm:csMAVE_convergence}}
\label{sec:proofs_thm:csMAVE_convergence}
\begin{proof}
By definition of $\hat{\beta}$,
\begin{equation}
    \mathcal S_{\widehat Y,k,n}(\widehat\beta)=0,
    \qquad k=1,\ldots,d_0.
\label{eq:score-zero-generated}
\end{equation}
Let
\begin{equation*}
    \widehat\Delta
    =
    \Delta_{\widehat\beta}
    =
    (I_p-\widehat\beta\widehat\beta^\top)\beta_0
    =
    (\widehat\delta_1,\ldots,\widehat\delta_{d_0}).
\end{equation*}
Evaluating Lemma~\ref{lem:generated-WY-expansion} at $\beta=\widehat\beta$ and using
\eqref{eq:score-zero-generated}, we obtain, for each $k=1,\ldots,d_0$,
\begin{equation}
\begin{aligned}
    0
    &=
    -h^2\sum_{\ell=1}^{d_0}
    M_{k\ell,n}\widehat\delta_\ell
    +
    h^4\widehat\Delta a_{k,n} \\
    &\quad+
    O_p(h^5+h^3\delta_n+hR_n)
    +
    o_p(h^2)\|\widehat\Delta\|.
\end{aligned}
\label{eq:score-at-betahat}
\end{equation}
Define
\begin{equation*}
    M_n=(M_{k\ell,n})_{k,\ell=1}^{d_0},
    \qquad
    A_n=(a_{1,n},\ldots,a_{d_0,n}).
\end{equation*}
Since $M_n$ is symmetric, stacking \eqref{eq:score-at-betahat} over
$k=1,\ldots,d_0$ gives
\begin{equation}
    0
    =
    -h^2\widehat\Delta M_n
    +
    h^4\widehat\Delta A_n
    +
    O_p(h^5+h^3\delta_n+hR_n)
    +
    o_p(h^2)\|\widehat\Delta\|.
\label{eq:stacked-score}
\end{equation}
By Assumption~\ref{ass:csMAVE}(d) and Law of Large Numbers,
\begin{equation*}
    M_n
    =
    \frac1n\sum_{j=1}^n
    V(\beta_0^\top X_j)V(\beta_0^\top X_j)^\top
    =
    M+o_p(1),
\end{equation*}
where
\begin{equation*}
    M
    =
    E\{V(\beta_0^\top X)V(\beta_0^\top X)^\top\}
\end{equation*}
is nonsingular. It follows that $\|M_n^{-1}\|=O_p(1)$. Multiplying
\eqref{eq:stacked-score} on the right by $M_n^{-1}$ yields
\begin{equation}
    h^2\widehat\Delta
    =
    h^4\widehat\Delta A_nM_n^{-1}
    +
    O_p(h^5+h^3\delta_n+hR_n)
    +
    o_p(h^2)\|\widehat\Delta\|.
\label{eq:delta-linear-equation}
\end{equation}
By Lemma~\ref{lem:exact-response-WY-expansion}, $\|A_n\|=O_p(1)$. Combined with $\|M_n^{-1}\|=O_p(1)$, we have
\begin{equation*}
    h^4\widehat\Delta A_nM_n^{-1}
    =
    O_p(h^4\|\widehat\Delta\|)
    =
    o_p(h^2\|\widehat\Delta\|),
\end{equation*}
because $h\to0$. Therefore \eqref{eq:delta-linear-equation} implies
\begin{equation*}
    h^2\|\widehat\Delta\|
    =
    O_p(h^5+h^3\delta_n+hR_n)
    +
    o_p(h^2\|\widehat\Delta\|).
\end{equation*}
Absorbing the $o_p(h^2\|\widehat\Delta\|)$ term into the left-hand side and dividing by $h^2$ yields
\begin{equation*}
    \|\widehat\Delta\|
    =
    O_p\left(h^3+h\delta_n+h^{-1}R_n\right).
\end{equation*}
\end{proof}

% \clearpage
\subsection{Proof of Theorem \ref{thm:structural_dimension_consistency}}
\label{sec:proofs_thm:structural_dimension_consistency}
% Proof of structural dimension consistency
\begin{proof}
Recall that our proposed model for causal SDR using the ERS variant $\{\mu(X_i), X_i\}_{i=1}^n$ is defined deterministically as $\mu(X)=g(\beta_0^\top X)$ for some smooth function $g(\cdot)$. From here on out we will use $Y$ in place of $\mu(X)$ to maintain consistency with the notation for the response variable from \citeSM{xiaAdaptiveEstimationDimension2002SM}. 

We first consider the ideal case where $Y$ is known. For a working dimension $d$, we explicitly write the bandwidth $h$ as $h_d$ to make obvious the dependence and define the Nadaraya-Watson leave-one-out (NW-LOO) estimator $\hat{a}_{d0,j}$ of $Y_j$ using the estimated directions $\hat{\beta}$:
\begin{equation*}
    \hat{a}_{d0,j} := \frac{\sum_{i\neq j}^n K^{(i,j)}_{h_d} Y_i}{\sum_{i\neq j}^n K^{(i,j)}_{h_d}}, 
    \qquad K_{h_d}^{(i,j)} := K_{h_d}\{\hat{\beta}_1^\top(X_i-X_j),\dots, \hat{\beta}_d^\top(X_i-X_j)\}.
\end{equation*}
The structural dimension estimator is defined as $\hat{d} := \argmin_{0\leq d\leq p} \text{CV}_n(d; Y, X, \hat{\beta})$, where the cross-validated error is:
\begin{equation}
    \text{CV}_n(d;Y,X,\hat{\beta}) := n^{-1} \sum_{j=1}^n (Y_j-\hat{a}_{d0,j})^2, \qquad d=1,\dots,p,
    \label{eq:CV_error}
\end{equation}
and $\text{CV}_n(0) := n^{-1}\sum_{i=1}^n(Y_i-\bar{Y})^2$.

When $d < d_0$, it is clear that as $n \to \infty$, $\Prob\{\text{CV}_n(d) > \text{CV}_n(d_0)\} \to 1$ due to considering too small a subspace. When $d \geq d_0$, \citeSM{xiaAdaptiveEstimationDimension2002SM} established that for models with $\Var(\varepsilon)=\sigma^2 > 0$, the CV error expands as:
\[
    \text{CV}_n(d;Y,X,\hat{\beta}) = \sigma^2 + h^4_d J_d + \frac{\alpha_d}{nh^d_d}\{1+o_p(1)\} + O_p(n^{-1/2}+h^5_d),
\]
and
\begin{equation*}
    \begin{aligned}
    J_d &= \int\Bigg[\frac{1}{2}\,\mathrm{tr}\!\left\{\nabla^{2}g(v)\right\}+f_d(v)^{-1}\,\nabla g(v)^{\top}\nabla f_d(v)\Bigg]^2 \, f_d(v)\, dv, \\
    \alpha_d &= \E\{\E(\varepsilon^2 \mid \beta^\top X) / f_d(\beta^\top X)\} \int K^2(v_1,...,v_d) dv_1...dv_d,
    \end{aligned}
\end{equation*}
where $f_d(\cdot)$ is the marginal density of $\beta^\top X$. Under our deterministic setting where $\varepsilon=0$, the irreducible error $\sigma^2$, the variance penalty $\alpha_d$, and the $O_p(n^{-1/2})$ empirical noise fluctuation are all zero. The expression then simplifies to:
\begin{equation}
    \text{CV}_n(d; Y,X,\hat{\beta}) = h^4_d J_d + O_p(h_d^5). 
    \label{eq:CV_known_order}
\end{equation}
Notice that $J_d = O(1)$ as it does not depend on $n$, only $d$. It will also be important for $J_d>0$ for all $d\in\{d_0, d_0+1,...,p\}$, otherwise without the leading term the asymptotic monotonicity of $\text{CV}_n$ in $d$ would be ruined. This is due to the previous overfitting term containing $\alpha_d$ now being zero. Under the standard optimal bandwidth $h_d \asymp n^{-1/(d+4)}$, the leading penalty term $h^4_d$ strictly increases with $d$, and the remainder $O_p(h^5_d)$ is asymptotically negligible. Consequently, $\Prob\{\text{CV}_n(d) > \text{CV}_n(d_0)\} \to 1$ for all $d > d_0$, establishing consistency when $Y$ is known.

We now move on to our scenario for causal SDR where the true response is replaced by an estimate $\hat{Y}_i = Y_i + \Delta(X_i)$, with an estimation error bounded by $\max_i \Delta(X_i) = O_p(R_n)$ for some $R_n = o(1)$ as $n \to \infty$. By the linearity of the NW estimator, the LOO prediction using $\hat{Y}$ separates additively: $\hat{a}_{d0,j}^* = \hat{a}_{d0,j} + \bar{\Delta}_{d,j}$, where $\bar{\Delta}_{d,j}$ is the NW-LOO smoother applied solely to the error $\Delta(X_i)$:
\begin{equation*}
    \hat{a}^*_{d0,j} := \frac{\sum_{i\neq j}^n K^{(i,j)}_{h_d} \hat{Y_i}}{\sum_{i\neq j}^n K^{(i,j)}_{h_d}}=\frac{\sum_{i\neq j}^n K^{(i,j)}_{h_d} {Y_i}}{\sum_{i\neq j}^n K^{(i,j)}_{h_d}}+\frac{\sum_{i\neq j}^n K^{(i,j)}_{h_d} \Delta(X_i)}{\sum_{i\neq j}^n K^{(i,j)}_{h_d}}=\hat{a}_{d0,j} + \bar{\Delta}_{d,j}.
\end{equation*}

Substituting this into the CV criterion yields:
\begin{equation}
\begin{aligned}
   \text{CV}_n(d;\hat{Y},X,\hat{\beta}) &= n^{-1} \sum_{j=1}^n (\hat{Y}_j-\hat{a}^*_{d0,j})^2 \\ &=n^{-1}\sum_{j=1}^n \big[ (Y_j - \hat{a}_{d0,j}) + (\Delta(X_j) - \bar{\Delta}_{d,j})\big]^2 \\
   &= \text{CV}_n(d;Y,X,\hat{\beta}) + 2B_{1,d} + B_{2,d}.
\end{aligned}
\label{eq:CV_error_perturbed}
\end{equation}
Here, $B_{2,d} = n^{-1}\sum_{j=1}^n (\Delta(X_j) - \bar{\Delta}_{d,j})^2$ and $B_{1,d}$ is the cross-term. Since we have already derived the asymptotic order of the first term by \eqref{eq:CV_known_order}, it remains to bound the last two terms. 

Start with $B_{2,d}$. Because $\Delta(X)$ is uniformly bounded by $O_p(R_n)$ and the NW weights are non-negative and sum to one, 
\begin{equation*}
    |\bar{\Delta}_{d,j} | \leq \max_{i\neq j} |\Delta (X_i) | = O_p(R_n).
\end{equation*}
By the triangle inequality, the LOO residual for $\Delta(X)$ is also bounded uniformly:
\begin{equation*}
    | \Delta(X_j) - \bar{\Delta}_{d,j} | \leq |\Delta (X_j) | + |\bar{\Delta}_{d,j}|=O_p(R_n).
\end{equation*}
Then it is clear that 
\begin{equation}
    B_{2,d} = n^{-1}\sum_{j=1}^n (\Delta(X_j) - \bar{\Delta}_{d,j})^2 \leq n^{-1}\sum_{j=1}^n [O_p(R_n)]^2 =O_p(R_n^2).
    \label{eq:B2d_order}
\end{equation}

To bound $B_{1,d}$, we apply Cauchy-Schwarz and by \eqref{eq:CV_known_order} and \eqref{eq:B2d_order}:
\begin{align*}
    |B_{1,d}| &\leq \left\{ n^{-1}\sum_{j=1}^n (Y_j-\hat{a}_{d0,j})^2 \right\}^{1/2} \left\{ n^{-1}\sum_{j=1}^n (\Delta(X_j) - \bar{\Delta}_{d,j})^2 \right\}^{1/2} \\
    &= \{ \text{CV}_n(d; Y,X,\hat{\beta}) \}^{1/2} \{ B_{2,d} \}^{1/2} \\
    &= \{ O_p(h_d^4) \}^{1/2} \{ O_p(R_n^2) \}^{1/2} = O_p(R_n h_d^2).
\end{align*}
Substituting these bounds back into the expansion \eqref{eq:CV_error_perturbed}, we obtain:
\begin{align*}
    \text{CV}_n(d;\hat{Y},X,\hat{\beta}) &= h_d^4 J_d + O_p(h_d^5) + O_p(R_n h_d^2) + O_p(R_n^2).
\end{align*}

It follows that so long as $R_n=o(h_d^2)$ as given by Assumption~\ref{ass:csMAVE}(e),  $\text{CV}_n(d;\hat{Y},X,\hat{\beta})$ behaves asymptotically identically to $\text{CV}_n(d;Y,X,\hat{\beta})$, preserving the monotonicity of the error when $d \geq d_0$ and $d$ increases. This is a stronger condition than needed for mere consistency of $\hat{d}$ but we enforce it to preserve the asymptotic properties of the structural dimension estimator for the original MAVE procedure. Then $\hat{d} := \argmin_{0\leq d\leq p} \text{CV}_n(d; \hat{Y}, X, \hat{\beta})$ is consistent for $d_0$.
\end{proof}

\subsection{Proof of Proposition~\ref{prop:MAVE_asymptotic_efficient}}
\label{sec:proof_MAVE_eff}
% Proof of asymptotic efficiency of MAVE
\begin{proof}
Define the following terms:
$$
Z_i=\beta^\top X_i, \qquad \ell_i=\ell(Z_i) = \E(Y\mid Z=Z_i), \qquad e_i=Y_i-\ell_i.
$$
In addition let $\ell'(\cdot)$ be the $d$-length row vector with entries equal to the gradient of $\ell(\cdot)$ with respect to $Z$. 
Recall that the efficient score for estimating the central mean subspace is:
\begin{equation*}
S_{\mathrm{eff}}(\beta)= \frac{1}{\sigma^{2}(X)} \left[ X - \frac{\mathbb{E}(X / \sigma^{2}(X)\mid \beta ^{\top}X)}{\mathbb{E}(1 / \sigma^{2}(X)\mid \beta ^{\top}X)} \right]\ell'(\beta ^{\top}X)[Y - \ell(\beta ^{\top}X)].
\end{equation*}
Assuming the conditional variance $\sigma^2(X)$ is a fixed constant, our goal is to show that the MAVE sample objective function converges to the expected efficient score:
\begin{equation}\label{eq:expected_eff_score}
\begin{aligned}
\mathbb{E}\big\{\left[ X - \mathbb{E}(X\mid \beta ^{\top}X) \right]\ell'(\beta ^{\top}X)[Y - \ell(\beta ^{\top}X)]\big\} \\
= \mathbb{E}\big\{\left[ X - \mathbb{E}(X\mid Z) \right]\ell'(Z)e\big\}.
\end{aligned}
\end{equation}

We will focus on the refined MAVE (RMAVE) objective function which is given by:
\begin{equation}\label{eq:RMAVE_objective}
\argmin_{\beta,a_j,b_j} \sum_{i=1}^n \sum_{j=1}^n [Y_{i}-a_{j}-b_{j}^{\top}\beta ^{\top}(X_{i}-X_{j})]^{2} w_{ij}, \qquad w_{ij}=\frac{K_{h}(\beta^\top( X_{i}- X_{j}))}{\sum_{k=1}^n K_{h}(\beta^\top (X_{j}-X_{k}))}.
\end{equation}
The RMAVE estimate for $\hat{\beta}$ is more efficient than that of standard MAVE as it smooths over the lower dimensional $Z$ rather than $X$. The estimator of the CMS is then the solution to the following profiled score:
\begin{equation}
\label{eq:score_MAVE}
S_{\mathrm{MAVE}}(\beta)=n^{-1}\sum_{j=1}^n\sum_{i=1}^n w_{ij}[Y_{i}-\hat{a}_{j}-\hat{b}_{j}^{\top}\beta ^{\top}(X_{i}-X_{j})][X_{i}-X_{j}]\hat{b}_{j}^\top=0,
\end{equation}
where $\hat{a}_j=\hat{a}(Z_j)$ and $\hat{b}_j=\hat{b}(Z_j)$ are the solutions of \eqref{eq:score_MAVE} keeping $\beta$ fixed to an initial estimate fit with a separate procedure like the outer product of gradients (OPG) method \citepSM{xiaAdaptiveEstimationDimension2002SM}. Notice that \eqref{eq:score_MAVE} is the gradient of \eqref{eq:RMAVE_objective} with respect to $\beta$, but treating $w_{ij}$ as fixed constants rather than as $\beta$-dependent weights. The $w_{ij}$ should therefore be treated as fixed quantities within an update, rather than as terms to be differentiated as part of the score. More precisely, at iteration $t$, the weights have the form
$$
w_{ij}^{(t)} = w_{ij}(\widehat\beta^{(t-1)}),
$$
and the update $\widehat\beta^{(t)}$ is obtained by solving the corresponding
weighted estimating equation while treating $w_{ij}^{(t)}$ as fixed. After the
update, the weights are recomputed using $\widehat\beta^{(t)}$. Therefore, we analyze the asymptotic score at convergence of this iterative procedure. At such a point at iteration $t$, the weights and the estimate $\widehat\beta$ are evaluated at asymptotically equivalent values, so the dependence of $w_{ij}$ on
$\beta$ may be ignored in the score.

We first examine the asymptotic behavior of $\hat{a}_j$ and $\hat{b}_j$. Let $\delta_n = \left(\frac{\log n}{nh^d}\right)^{1/2}$. Under standard regularity conditions, uniform convergence of local linear regression \citepSM{fanVariableBandwidthLocal1992SM} and our assumption that $h\to 0$ and $nh^{d+2} / \log (n) \to \infty$ give us:
\begin{equation}\label{eq:A_B_bound}
\begin{aligned}
    \max_{1\leq j\leq n} \lVert A_j \rVert &=  \max_{1\leq j\leq n} \lVert \hat{a}_j - \ell(Z_j) \rVert \leq \sup_{z \in \mathcal{Z}} \lVert \hat{a}(z) - \ell(z) \rVert = O_p(h^2+ \delta_n)=o_p(1), \\
    \max_{1\leq j\leq n} \lVert B_j \rVert &= \max_{1\leq j\leq n} \lVert \hat{b}_j - \ell'(Z_j) \rVert \leq \sup_{z \in \mathcal{Z}} \lVert \hat{b}(z) - \ell(z) \rVert = O_p\left(h^2+ \frac{\delta_n}{h}\right)=o_p(1).
\end{aligned}
\end{equation}
A first-order Taylor expansion of $\ell(\cdot)$ also gives us the Lagrange remainder form, $\ell(Z_i)=\ell(Z_j)+\ell'(Z_j)(Z_i-Z_j) + O(h^2)$, which is given by bounded, continuous derivatives of $\ell(\cdot)$ up to order three. We now decompose the residual term as
\begin{align*}
\hat{r}_{ij} &= Y_{i}-\hat{a}_{j}-\hat{b}_{j}^{\top}(Z_i-Z_j)\\
    &= (\ell_i+e_i) - (\ell_j + A_j) - (\ell'_j+B_j^\top) (Z_i-Z_j) \\
    &= e_i + \big[\ell_i - \ell_j - \ell'_j(Z_i-Z_j)\big] - A_j - B_j^\top(Z_i-Z_j) \\ 
    &= e_i + \tau_{ij} - A_j - B_j^\top(Z_i-Z_j), \qquad \tau_{ij}=\ell_i-\ell_j-\ell'_j(Z_i-Z_j).
\end{align*}
Let $\rho_{ij}=\tau_{ij} - A_j - B_j^\top (Z_i-Z_j)$ so that $\hat{r}_{ij}=e_i+\rho_{ij}$. On the bounded support of $K_h(\cdot)$, it follows that $\lVert Z_i-Z_j \rVert = O(h)$. Then
\begin{equation}\label{eq:rho_bound}
\begin{aligned}
    \max_{i,j: w_{ij}>0} \|\rho_{ij}\| 
    &\leq  \max_{i,j: w_{ij}>0} \| \tau_{ij}\| + \max_{j}\|A_j\| + \max_j \|B_j\| \cdot \|Z_i-Z_j \| \\
    &\leq O(h^2) +O_p\left( h^2 + \delta_n\right) + O_p\left( h^2 + \frac{\delta_n}{h}\right)O_p(h)=O_p(h^2+\delta_n)=o_p(1).   
\end{aligned}
\end{equation}
Plugging terms back into \eqref{eq:score_MAVE}, we get the following decomposition of $S_{\mathrm{MAVE}}(\beta)$:
\begin{equation*}
\begin{aligned}
S_{\mathrm{MAVE}}(\beta)
&=
\frac{1}{n}\sum_{j=1}^n\sum_{i=1}^n
w_{ij}(X_i-X_j)\hat{b}_j\hat{r}_{ij} 
=
\frac{1}{n}\sum_{j=1}^n\sum_{i=1}^n
w_{ij}(X_i-X_j)
\{\ell'_j+B_j^\top\}\{e_i+\rho_{ij}\} \\
&=
\frac{1}{n}\sum_{j=1}^n\sum_{i=1}^n
w_{ij}(X_i-X_j)\ell'_je_i \\
&\quad+
\frac{1}{n}\sum_{j=1}^n\sum_{i=1}^n
w_{ij}(X_i-X_j)B_j^\top e_i \\
&\quad+
\frac{1}{n}\sum_{j=1}^n\sum_{i=1}^n
w_{ij}(X_i-X_j)\ell'_j\rho_{ij} \\
&\quad+
\frac{1}{n}\sum_{j=1}^n\sum_{i=1}^n
w_{ij}(X_i-X_j)B_j^\top\rho_{ij} \\
&=: T_{0n}+T_{1n}+T_{2n}+T_{3n}.
\end{aligned}
\end{equation*}

Since $X$ has bounded support and $\E(|e_i|^3)<\infty$, it follows that
\begin{equation}\label{eq:weighted_sum_constant}
    \begin{aligned}
        \frac{1}{n} \sum_{j=1}^n \sum_{i=1}^n w_{ij} ( X_i-X_j) &= O_p(1), \\
        \frac{1}{n} \sum_{j=1}^n \sum_{i=1}^n w_{ij} e_i( X_i-X_j) &= O_p(1).
    \end{aligned}
\end{equation}
Combining \eqref{eq:weighted_sum_constant} with Assumption C8 in Section~\ref{sec:supp_MAVE_assumption} (which implies that $\ell'_j$ is bounded), together with the vanishing bounds on $B_j$ \eqref{eq:A_B_bound} and $\rho_{ij}$ \eqref{eq:rho_bound}, it follows that $T_{1n}$, $T_{2n}$, and $T_{3n}$ are all $o_p(1)$. Thus, we have reduced the MAVE score to its leading term $T_{0n}$ and a set of asymptotically negligible terms. Writing $\ell_j$ now as $\ell(Z_j)$ to make the dependence on $Z$ clear, it follows that:
$$
S_{\mathrm{MAVE}}(\beta)
=
\frac{1}{n}\sum_{j=1}^n\sum_{i=1}^n
w_{ij}(X_i-X_j)\ell'(Z_j)e_i
+o_p(1).
$$
Collecting the terms inside the inner sum over $i$:
\begin{align}\label{eq:Smave_innersum}
S_{\mathrm{MAVE}}(\beta)
&=
\frac{1}{n}\sum_{j=1}^n
\left[
\sum_{i=1}^n w_{ij}X_i e_i
-
X_j\sum_{i=1}^n w_{ij}e_i
\right]\ell'(Z_j)
+o_p(1).
\end{align}
The two weighted sums are Nadaraya--Watson estimators evaluated at $Z_j$. Under uniform convergence for local linear smoothers:
\begin{align*}
\sum_{i=1}^n w_{ij}X_i e_i
&=
\E(Xe\mid Z=Z_j)+r_{0n,j}, \qquad \max_{1\leq j\leq n} |r_{0n,j}| = o_p(1), \\
\sum_{i=1}^n w_{ij}e_i
&=
\E(e\mid Z=Z_j)+r_{1n,j}, \qquad \max_{1\leq j\leq n} |r_{1n,j}| = o_p(1).
\end{align*}
Since $\E(e\mid Z)=0$, the second term of the inner sum in \eqref{eq:Smave_innersum} is asymptotically negligible. By the law of large numbers, 
\begin{align*}
S_{\mathrm{MAVE}}(\beta)
&=
\frac{1}{n}\sum_{j=1}^n
\E(Xe\mid Z=Z_j)\ell'(Z_j)
+o_p(1) \\
&=
\E\left\{\E(Xe\mid Z)\ell'(Z)\right\}
+o_p(1),
\end{align*}

To resolve the $\E(Xe\mid Z)$ term, we again use the fact that $\E(e\mid Z)=0$: 
\begin{align*}
\E(Xe\mid Z)
&=
\E\left[\{X-\E(X\mid Z)\}e\mid Z\right]
+
\E(X\mid Z)\E(e\mid Z) \\
&=
\E\left[\{X-\E(X\mid Z)\}e\mid Z\right].
\end{align*}
Finally, plugging back into the score and using the tower rule:
\begin{align*}
S_{\mathrm{MAVE}}(\beta)
&=
\E\left[
\E\left\{\{X-\E(X\mid Z)\}e\mid Z\right\}\ell'(Z)
\right]
+o_p(1) \\
&=
\E\left[
\{X-\E(X\mid Z)\}\ell'(Z)e
\right]
+o_p(1).
\end{align*}
which converges to the expected efficient score in \eqref{eq:expected_eff_score}  as $n\to\infty$.
\end{proof}

\clearpage

\section{Reproducibility of SDR Efficient Score}\label{sec:reproducing_sdr}

Because our causal SDR method relies on both the association-based MAVE \citepSM{xiaAdaptiveEstimationDimension2002SM} and efficient score \citepSM{maEstimationEfficiencyCentral2014SM,luoNewEstimatorEfficient2016SM} for estimating $\CMS$, we conducted a replication study to verify the correctness of our implementation. For MAVE, we used the \texttt{R} package \texttt{MAVE} \citepSM{XiaHang2025MAVESM}, written by the first author of the original MAVE paper. For the efficient score–based one-step Newton–Raphson estimator, we implemented the algorithm from scratch. We first explain why we do not adopt the approach of \citeSM{maEstimationEfficiencyCentral2014SM}, and then justify our use of the method proposed by \citeSM{luoNewEstimatorEfficient2016SM}, including a replication of their simulation results. The software implementation of the Newton–Raphson estimator is available at \url{https://github.com/tXiao95/causalpca/blob/main/R/Seff.R}.

\subsection{Implementation details}

We use one-step estimation \citepSM{vandervaartAsymptoticStatistics1998SM} where the initial estimate is from MAVE and the updates are from the efficient score $S_\text{eff}$ with $\hat{S}_{\text{eff}}$ indicating the version with plug-in nuisances. Using the Newton--Raphson algorithm and iterating until convergence, for any $k$, with $\E_n$ denoting the empirical mean and $\dagger$ the Moore-Penrose generalized inverse, we obtain
\begin{equation*}
    \text{vec}(\hat{\beta}^{(k+1)}) = \text{vec}(\hat{\beta}^{(k)}) + \E_n^{\dagger}\left[\text{vec}\{\hat{S}_{\text{eff}}(\hat{\beta}^{(k)})\}^{\otimes2}\right] \E_n\left[\text{vec}\{\hat{S}_{\text{eff}}(\hat{\beta}^{(k)})\}\right].
\end{equation*}
We define convergence as when $\left\lVert P_{\hat{\beta}^{(k)}}-P_{\hat{\beta}^{(k+1)}}\right\rVert_2 < p / n$ or the number of iterations reaches $100$, closely following the stopping rule established by \citeSM{luoNewEstimatorEfficient2016SM}.

\subsection{Description of Ma and Zhu (2014)}

To develop asymptotic theory for the basis matrix $\beta$, \citeSM{maEstimationEfficiencyCentral2014SM} impose an identifiability constraint in which the top $d \times d$ block of $\beta \in \mathbb{R}^{p \times d}$ is fixed to the identity matrix $I_d$. While this constraint ensures uniqueness in the remainder of $\beta$, it requires prior knowledge of which $d$ exposures are linearly independent in the central mean subspace (CMS). If the selected independent exposures are incorrect, estimation of the remaining entries of $\beta$ may be invalid. Therefore this procedure is not robust to the ordering of the original $p$ variables in $X$.

 Due to the reliance on this identifiability condition, we do not adopt their efficient score approach. However, to understand the numerical performance of MAVE relative to the efficient score, we still examine their simulation results for ``Scenario 1," a nonlinear, heteroskedastic variance data generating process which compares MAVE to three implementations of their efficient score, targeting the $(p-d)\times d$ free parameters: 1) a constant variance model (EE1), 2) parametric model (EE2) and 3) semiparametric model (EE3) for $\sigma^2(X)$. (Table~\ref{tab:ma_zhu_sim_results}). The true values (Truth) and efficient score using known $\sigma^2(X)$ (Oracle) are also included for reference. The EE methods showed only modest improvements over MAVE. Although they tend to yield slightly smaller standard errors, MAVE attains lower mean error across all entries in this setting. The remaining simulation scenarios reported in their study (not shown here) report larger gains for EE, but the overall pattern is consistent. In their empirical studies, EE provides only modest improvements relative to MAVE.

\begin{table}[!htbp]
\centering
\caption{Abridged Table 1 from Ma and Zhu (2014) for Scenario 1.}
\label{tab:ma_zhu_sim_results}
\begin{tabular}{lccccc}
\toprule
Method & $\beta_1$ & $\beta_2$ & $\beta_3$ & $\beta_4$ & $\beta_5$ \\
\midrule

Truth 
& 0.553 & 1.106 & 1.032 & 0.885 & 0.590 \\

\midrule

Oracle 
& 0.563 (0.041) & 1.129 (0.052) & 1.054 (0.048) & 0.902 (0.046) & 0.601 (0.040) \\

EE1 
& 0.555 (0.048) & 1.113 (0.063) & 1.040 (0.057) & 0.888 (0.055) & 0.592 (0.048) \\

EE2 
& 0.563 (0.041) & 1.131 (0.053) & 1.056 (0.048) & 0.903 (0.047) & 0.602 (0.041) \\

EE3 
& 0.564 (0.040) & 1.131 (0.051) & 1.056 (0.047) & 0.903 (0.045) & 0.602 (0.040) \\

MAVE 
& 0.551 (0.043) & 1.104 (0.055) & 1.033 (0.055) & 0.883 (0.048) & 0.590 (0.043) \\

\bottomrule
\end{tabular}
\end{table}

\subsection{Reproducing EE from Luo and Cai (2016)}

Unlike \citeSM{maEstimationEfficiencyCentral2014SM}, \citeSM{luoNewEstimatorEfficient2016SM} report substantial improvements in CMS estimation when relying on the efficient score. However, they do not impose an identifiability constraint and instead estimate $\sigma^2(X)$ in the score via kernel smoothing. Because this approach is invariant to the ordering of $X$, we adopt it for the Newton–Raphson procedure. 

To verify correctness of our implementation, we sought to replicate Table 1 from \citeSM{luoNewEstimatorEfficient2016SM}. The comparison focuses on MAVE, EE4: the efficient estimator using bandwidth $h_1 = n^{-1/4p}$ for kernel estimation of $\sigma^2(X)$, and the Oracle estimator, which substitutes the true $\sigma^2(X)$ into the efficient score. All simulations were conducted under the specifications in Section 4 of the original paper, including implementation details such as stopping criteria and bandwidth selection. Because the paper does not clearly specify whether the reported 2-norm corresponds to the Frobenius or spectral norm, we include both in Table~\ref{tab:luo_cai_replication}.

\begin{table}[!t]
\centering
\caption{Original Table 1 from Luo and Cai (2016) and our replication attempt under both Frobenius and spectral norm. Entries are mean (SD) across simulation runs.}
\label{tab:luo_cai_replication}
\begin{tabular}{lccc}
\toprule
Model & MAVE & EE4 & Oracle \\
\midrule

\multicolumn{4}{l}{\textbf{Panel A: Original paper (Table 1)}} \\
\midrule
I   & 0.282 (0.060) & 0.298 (0.066) & 0.281 (0.052) \\
II  & 0.471 (0.109) & 0.291 (0.048) & 0.218 (0.082) \\
III & 0.117 (0.057) & 0.077 (0.038) & 0.056 (0.029) \\
IV  & 0.398 (0.358) & 0.164 (0.132) & 0.168 (0.140) \\
V   & 0.653 (0.420) & 0.193 (0.141) & 0.150 (0.132) \\

\addlinespace
\midrule
\multicolumn{4}{l}{\textbf{Panel B: Replication Study (Frobenius norm)}} \\
\midrule
I   & 0.152 (0.035) & 0.165 (0.038) & 0.156 (0.037) \\
II  & 0.063 (0.017) & 0.063 (0.017) & 0.057 (0.017) \\
III & 0.507 (0.116) & 0.494 (0.114) & 0.458 (0.369) \\
IV  & 0.849 (0.282) & 0.837 (0.285) & 0.773 (0.292) \\
V   & 1.252 (0.243) & 1.239 (0.255) & 1.016 (0.345) \\

\addlinespace
\midrule
\multicolumn{4}{l}{\textbf{Panel C: Replication Study (Spectral norm)}} \\
\midrule
I   & 0.108 (0.025) & 0.117 (0.027) & 0.110 (0.026) \\
II  & 0.044 (0.012) & 0.045 (0.012) & 0.041 (0.012) \\
III & 0.316 (0.081) & 0.308 (0.080) & 0.286 (0.213) \\
IV  & 0.541 (0.200) & 0.535 (0.201) & 0.496 (0.205) \\
V   & 0.795 (0.166) & 0.789 (0.173) & 0.666 (0.239) \\

\bottomrule
\end{tabular}
\end{table}
In the original study, Model I is the only setting with constant variance, under which MAVE is expected to be efficient (Proposition~\ref{prop:MAVE_asymptotic_efficient}), whereas EE4 and Oracle should outperform MAVE in the rest of the models as they exhibit heteroscedastic variance. This pattern is clearly reflected in the original simulation results in Panel A. However, we were unable to reproduce their findings. In our implementation, EE4 and Oracle made only modest reductions in the error reported in the original paper. We have confirmed with the original authors that their code is no longer available.

Despite this discrepancy, two observations provide reassurance regarding the correctness of our implementation. First, Model I is the only setting in which MAVE outperforms EE4 and Oracle. In all other models, EE4 and Oracle uniformly outperform MAVE, although the magnitude of improvement is modest compared to the original study. Second, the Oracle estimator consistently achieves the lowest error out of the three methods, with more noticeable gains in Models IV and V. These patterns align with theoretical expectations. The primary difference from the original study lies in the magnitude of improvement attributed to EE4 and Oracle and the scale of the errors.

A plausible explanation is that differences arise from the implementation of the initial MAVE estimator rather than the Newton–Raphson update itself. The \texttt{MAVE} R package used in our study was released in 2025, well before the publication of  \citeSM{luoNewEstimatorEfficient2016SM}. The specific implementation details for both MAVE and the efficient score used in the original paper are unknown and may differ. We do not know whether such differences are sufficient to explain the observed discrepancy in CMS estimation error. Our observed results are closer to the simulations from \citeSM{maEstimationEfficiencyCentral2014SM}, where EE only led to modest improvements in estimation error over MAVE, even under heteroskedastic variance. We include this replication study to promote transparency and provide support for our decision to use MAVE.

\clearpage

\section{Additional Simulation Results}\label{sec:additional_sim}

\subsection{MAVE vs. EE estimation}\label{sec:sim_mave_vs_ee}

\begin{table}[!htbp]
\centering
\caption{MAVE vs. EE estimation error for $\CCMS$.}
\label{tab:mave_vs_ee_sim}
% --- Start Resizebox ---
\resizebox{\textwidth}{!}{
\begin{tabular}{ll ccccc}
\toprule
 & & \multicolumn{5}{c}{Method} \\
\cmidrule(lr){3-7}
Sample Size & Estimator & Oracle & RA & DR & PO & RP \\
\midrule
\multirow{2}{*}{$n = 100$} 
& MAVE & 0.216 (0.136) & 1.343 (0.341) & 1.537 (0.267) & 1.487 (0.318) & 1.768 (0.218) \\
& EE   & 0.635 (0.261) & 1.352 (0.332) & 1.543 (0.262) & 1.491 (0.315) & 1.776 (0.217) \\
\addlinespace % Elegant vertical gap instead of a hard line

\multirow{2}{*}{$n = 500$} 
& MAVE & 0.030 (0.010) & 0.447 (0.392) & 0.568 (0.369) & 0.663 (0.450) & 1.413 (0.259) \\
& EE   & 0.379 (0.131) & 0.463 (0.384) & 0.589 (0.361) & 0.657 (0.445) & 1.443 (0.254) \\
\addlinespace

\multirow{2}{*}{$n = 1000$} 
& MAVE & 0.015 (0.005) & 0.223 (0.288) & 0.303 (0.257) & 0.432 (0.401) & 1.192 (0.240) \\
& EE   & 0.316 (0.107) & 0.232 (0.282) & 0.323 (0.253) & 0.412 (0.393) & 1.226 (0.237) \\
\addlinespace

\multirow{2}{*}{$n = 2500$} 
& MAVE & 0.006 (0.002) & 0.073 (0.157) & 0.140 (0.151) & 0.287 (0.345) & 0.913 (0.255) \\
& EE   & 0.254 (0.086) & 0.080 (0.153) & 0.162 (0.148) & 0.246 (0.333) & 0.952 (0.249) \\
\addlinespace

\multirow{2}{*}{$n = 5000$} 
& MAVE & 0.003 (0.001) & 0.026 (0.035) & 0.078 (0.058) & 0.221 (0.293) & 0.676 (0.285) \\
& EE   & 0.219 (0.070) & 0.031 (0.035) & 0.100 (0.061) & 0.165 (0.280) & 0.697 (0.271) \\
\bottomrule
\end{tabular}
} % --- End Resizebox ---
\end{table}

Table~\ref{tab:mave_vs_ee_sim} compares the error in estimating $\CCMS$ between MAVE and the efficient one-step estimator (EE) in the main simulation study described in Section~\ref{sec:sims}. Since the error was not meaningfully reduced for any of the causal SDR methods, and often even made the Oracle method worse, we opted to use MAVE in our implementation. The consistently worse error for EE applied to the Oracle first-stage estimator may stem from the deterministic nature of $\mu(x)$, under which $\mathrm{Var}(\mu(X)\mid X=x)$ is degenerate and therefore not well-defined for estimation. The one causal SDR variant where improvements were made is PO, where both mean error and SD decreased, possibly due to the nonzero variance present in the pseudo-outcomes. All other variants appear to have slightly higher error after the update with no meaningful change. This echoes what we observed in our effort to reproduce the association-based SDR results of \citeSM{maEstimationEfficiencyCentral2014SM} and \citeSM{luoNewEstimatorEfficient2016SM} in Appendix~\ref{sec:reproducing_sdr}. Therefore, throughout the main text and Appendix we only report the results of MAVE estimation for causal SDR. 

\clearpage

\subsection{DGP with additive confounding}\label{subsec:add_sim_add_confound}

We simulated the same DGP described in the main experiment \eqref{eq:main_sim_dgp} under an additive confounding assumption (Assumption~\ref{ass:additive_confounding}) where $f(Z,C)=0$ to give an example where the assumptions underlying the RP variant of causal SDR were met. Table~\ref{tab:frob_error_additive} reports the error in estimating $\CCMS$ under the additive DGP in a format similar to Table~\ref{tab:frob_error_dhat}. Under this setting, RP now estimates $\CCMS$ more accurately than the baseline methods, unlike in the main experiment with interactions, where its error was much larger. However, RP still does not outperform the RA, DR, and PO variants (only does better than PO at $n=2500$ and beyond). Although RP avoids the need to estimate $\mu(X)$, it still requires correctly estimating all $p+1$ regressions, which remains nontrivial even when each regression is lower-dimensional in the $q$ covariates.

% Frobenius error table
\begin{table}[!htbp]
\centering
\caption{Frobenius norm error between estimated and true subspace using $\hat{d}$ under additive confounding. The RP variant is highlighted in bold.}
\label{tab:frob_error_additive}
\resizebox{\textwidth}{!}{
\begin{tabular}{llccccc}
\toprule
Category & Method & $n = 100$ & $n = 500$ & $n = 1000$ & $n = 2500$ & $n = 5000$ \\ 
\midrule
\multirow{3}{*}{Baseline} 
& PCA         & 1.997 (0.002) & 1.997 (0.001) & 1.997 (0.001) & 1.997 (0.000) & 1.997 (0.000) \\
& pCCA        & 1.300 (0.126) & 1.142 (0.179) & 0.988 (0.205) & 0.781 (0.202) & 0.604 (0.174) \\
& MAVE        & 1.415 (0.216) & 1.159 (0.089) & 1.100 (0.043) & 1.067 (0.021) & 1.053 (0.011) \\
\midrule
\multirow{5}{*}{CSDR} 
& Oracle & 0.216 (0.136) & 0.030 (0.010) & 0.015 (0.005) & 0.006 (0.002) & 0.003 (0.001) \\
& RA     & 1.048 (0.340) & 0.238 (0.235) & 0.105 (0.137) & 0.037 (0.070) & 0.018 (0.036) \\
& DR     & 1.274 (0.295) & 0.303 (0.211) & 0.160 (0.128) & 0.083 (0.074) & 0.054 (0.037) \\
& PO     & 1.228 (0.356) & 0.399 (0.349) & 0.272 (0.339) & 0.191 (0.283) & 0.160 (0.253) \\
& \textbf{RP}     & \textbf{1.534 (0.258)} & \textbf{0.654 (0.263)} & \textbf{0.330 (0.136)} & \textbf{0.179 (0.046)} & \textbf{0.120 (0.029)} \\
\bottomrule
\end{tabular}
}
\end{table}

Table~\ref{tab:ers_rmse_additive} summarizes the RMSE for estimating the ERS at the 100 evaluation points under additive confounding, using the same format as Table~\ref{tab:rmse_ers}. RP again shows marked improvement. However, it still underperforms relative to the RA and DR variants, and even PO, which it outperformed at large $n$ in estimating $\CCMS$ back in Table~\ref{tab:frob_error_additive}.

\begin{table}[!htbp]
\centering
\caption{Mean (SD) RMSE for ERS estimation across 100 evaluation points under additive confounding. The RP variant is highlighted in bold.}
\label{tab:ers_rmse_additive}
\resizebox{\textwidth}{!}{
\begin{tabular}{llccccc}
\toprule
Category & Method & $n = 100$ & $n = 500$ & $n = 1000$ & $n = 2500$ & $n = 5000$ \\
\midrule
\multirow{3}{*}{Baseline}
& pCCA   & 9.20 (5.44) & 9.57 (4.97) & 8.57 (4.82) & 5.85 (3.68) & 4.18 (2.58) \\
& MAVE   & 6.98 (3.85) & 2.94 (0.90) & 2.43 (0.45) & 2.14 (0.29) & 2.04 (0.21) \\
& Full $X$ & 6.97 (3.87) & 3.50 (0.92) & 2.71 (0.50) & 2.09 (0.31) & 1.87 (0.22) \\
\midrule
\multirow{5}{*}{CSDR}
& Oracle & 3.87 (3.02) & 1.12 (0.70) & 0.79 (0.38) & 0.56 (0.28) & 0.44 (0.28) \\
& RA     & 4.78 (3.14) & 1.32 (0.73) & 0.86 (0.47) & 0.58 (0.35) & 0.45 (0.26) \\
& DR     & 5.36 (3.41) & 1.35 (0.67) & 0.88 (0.42) & 0.59 (0.30) & 0.45 (0.24) \\
& PO     & 5.28 (3.48) & 1.45 (0.79) & 0.97 (0.51) & 0.67 (0.39) & 0.53 (0.42) \\
& \textbf{RP}     & \textbf{8.47 (5.13)} & \textbf{2.92 (1.23)} & \textbf{1.82 (0.67)} & \textbf{1.10 (0.38)} & \textbf{0.78 (0.31)} \\
\bottomrule
\end{tabular}
}
\end{table}

\subsection{Super Learner for ERS estimation}\label{subsec:add_sim_superlearner}

Table~\ref{tab:rmse_ers} and Figure~\ref{fig:sim_pointwise_coverage} showed that the Full $X$ method has higher error and lower coverage in estimating the ERS compared to the causal SDR methods. To assess the extent to which this performance gap is driven by the choice of nuisance function estimator, we repeated the simulation experiment from Section~\ref{sec:sims}, replacing the neural network used in the second-stage ERS estimation with a SuperLearner ensemble comprised of linear regression (\texttt{glm}), elastic net (\texttt{glmnet}), generalized additive models (\texttt{gam}), multivariate adaptive regression splines (\texttt{earth}), gradient boosted trees (\texttt{xgboost}), and random forests (\texttt{ranger}).

\begin{table}[!htbp]
\centering
\caption{Mean (SD) RMSE for ERS estimation across simulation runs using SuperLearner ensemble in place of neural network for second-stage. The focus, Full $X$, is bolded. Full $X$ (new) refers to ERS estimation using the SuperLearner ensemble, whereas Full $X$ (old) is the implementation from the main text using the neural network.}
\label{tab:ers_rmse_superlearner}
\resizebox{\textwidth}{!}{
\begin{tabular}{llccccc}
\toprule
Category & Method & $n = 100$ & $n = 500$ & $n = 1000$ & $n = 2500$ & $n = 5000$ \\
\midrule
\multirow{4}{*}{Baseline}
& pCCA     & 14.28 (9.87) & 9.11 (6.61) & 7.71 (4.30) & 6.30 (3.35) & 5.27 (3.02) \\
& MAVE     & 10.88 (7.16) & 5.66 (2.50) & 4.48 (1.45) & 3.75 (1.08) & 3.48 (0.85) \\
& \textbf{Full $X$ (new)} & \textbf{7.24 (6.08)} & \textbf{3.26 (2.02)} & \textbf{2.34 (1.19)} & \textbf{1.67 (0.61)} & \textbf{1.35 (0.32)} \\
& \textbf{Full $X$ (old)} & \textbf{9.36 (5.86)} & \textbf{3.98 (1.15)} & \textbf{3.01 (0.70)} & \textbf{2.35 (0.38)} & \textbf{2.12 (0.26)} \\
\midrule
\multirow{5}{*}{CSDR}
& Oracle & 5.80 (5.16) & 2.40 (1.55) & 1.61 (0.74) & 1.06 (0.61) & 0.79 (0.57) \\
& RA     & 8.31 (6.77) & 2.89 (1.73) & 1.85 (0.94) & 1.13 (0.69) & 0.79 (0.56) \\
& DR     & 9.43 (6.81) & 2.93 (1.73) & 1.88 (1.00) & 1.14 (0.66) & 0.81 (0.55) \\
& PO     & 9.02 (6.86) & 3.07 (1.75) & 2.02 (1.08) & 1.25 (0.79) & 0.90 (0.69) \\
& RP     & 14.63 (10.60) & 7.13 (4.76) & 4.87 (2.51) & 3.20 (1.28) & 2.54 (1.05) \\
\bottomrule
\end{tabular}
}
\end{table}
In general, using the SuperLearner ensemble in place of the neural network resulted in higher ERS estimation error and undercoverage for the causal SDR methods. In contrast, it led to improved error (Table~\ref{tab:ers_rmse_superlearner}) and substantially improved coverage (Figure~\ref{fig:sim_pointwise_coverage_SLxgboost}) for the baseline methods, including Full $X$. Their performance was still inferior to that of the causal SDR approaches. This indicates that although the choice of nuisance estimator affects both accuracy and coverage, the advantages of causal SDR persist across choice of flexible learners, reflecting gains in estimation efficiency and stability beyond what can be achieved through choice of nuisance model alone. 
\begin{figure}[!htbp]
    \centering
    \includegraphics[width=0.99\textwidth]{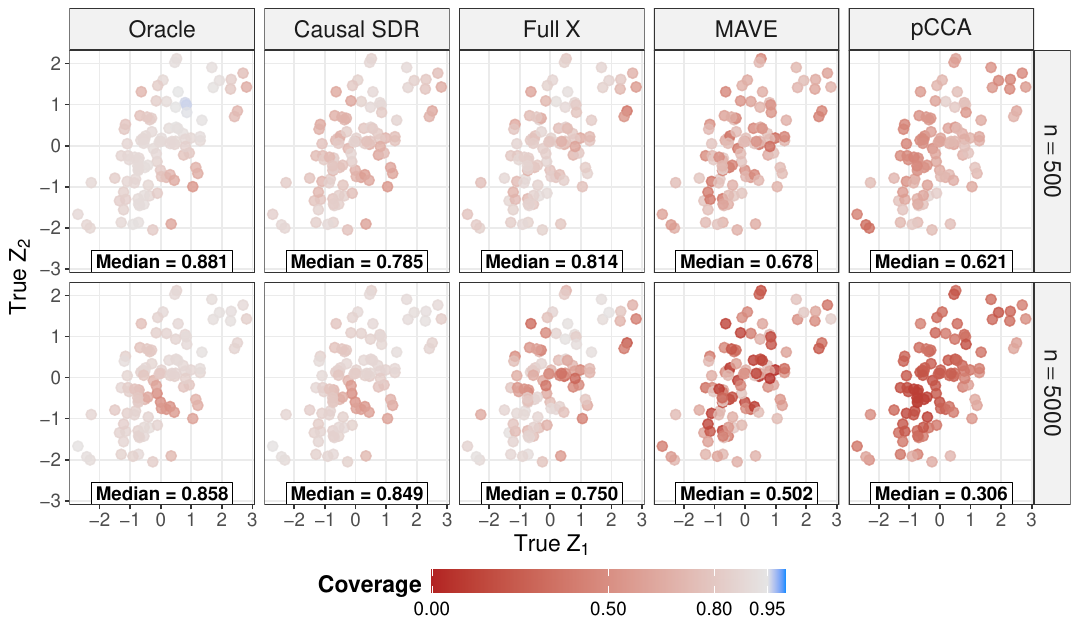}
    \caption{Pointwise 95\% coverage of ERS estimates across 100 evaluation points under SuperLearner second-stage estimation instead of neural network.}
    \label{fig:sim_pointwise_coverage_SLxgboost}
\end{figure}

\subsection{DGP with Nonsparse Dimension Reduction}\label{subsec:add_sim_nonsparse}
In the simulation from the main text in Section~\ref{sec:sims}, our $\beta$ was defined as:
$$
\beta_{\text{}}^\top = 
\begin{bmatrix}
1 & 0 & 0 & 0 & 0 & 0 & 0 & 0 & 0 & 0 \\
0 & 1 & 0 & 0 & 0 & 0 & 0 & 0 & 0 & 0
\end{bmatrix},
$$
where only $X_1$ and $X_2$ contribute to the CCMS. To evaluate causal SDR under a nonsparse dimension reduction where more of the original variables lie in $\CCMS$, we kept the original simulation design in \eqref{eq:main_sim_dgp} but changed $\beta$ to: 
$$
\beta_{new}^\top = 
\begin{bmatrix}
0.5 & 0 & 0.5 & 0 & 0 & 0 & 0.5 & 0 & 0.5 & 0 \\
0 & 0.5 & 0 & 0.5 & 0 & 0 & 0 & 0.5 & 0 & 0.5
\end{bmatrix}.
$$

Here, all of the original variables with the exception of $X_5$ and $X_6$ contribute equally to the CCMS. We then used the exact same learners in the first and second stage to estimate $\mu(X)$ and $\mu_\beta(z)$. However, while causal SDR still outperformed the alternatives, we found the benefit of causal SDR to be much less pronounced compared to the main simulation, and the error in estimating both $\CCMS$ and the ERS to be much higher (results not shown). Our hypothesis was that once more $X$ variables are required for estimation of $\mu(X)$, estimation becomes a greater challenge especially for a highly nonlinear DGP with several interactions. 

To investigate whether the problem was due to slow convergence of the first stage estimator, we moved from the simple SuperLearner ensemble with \texttt{SL.glm}, \texttt{SL.glmnet}, and \texttt{SL.earth} to a \texttt{torch} neural network architecture featuring \textit{three} hidden layers ($128 \times 64 \times 32$ nodes) trained for 1000 epochs with an initial learning rate of 0.01 and a step-decay scheduler that halved the learning rate every 200 epochs. This architecture was also used to estimate $\mu_\beta(z)$ in the second stage. Since the main simulation already showed the RP variant does not perform well under this DGP due to interactions between $Z$ and $C$, we did not include it in this simulation experiment. 

% Errors were higher 
The errors in estimating $\CCMS$ under nonsparse $\beta$ are shown in Table~\ref{tab:frob_error_nonsparse}. Overall, errors are higher for all methods compared to Table~\ref{tab:frob_error_dhat}. The pCCA method no longer shows improvement with increasing $n$. The Oracle variant also exhibits substantially higher error, starting at $1.019$ at $n=100$, compared to $0.216$ in the main simulation. Although the error for the causal SDR methods decreases as $n$ approaches $5000$, the rate of improvement is slower than in the sparse $\beta$ setting.

% Frobenius error table
\begin{table}[!tbp]
\centering
\caption{Frobenius norm error between estimated and true subspace evaluated at $\hat{d}$ under the nonsparse dimension reduction.}
\label{tab:frob_error_nonsparse}
\resizebox{\textwidth}{!}{
\begin{tabular}{llccccc}
\toprule
Category & Method & $n = 100$ & $n = 500$ & $n = 1000$ & $n = 2500$ & $n = 5000$ \\ 
\midrule
\multirow{3}{*}{Baseline} 
& PCA   & 1.647 (0.027) & 1.650 (0.012) & 1.651 (0.008) & 1.651 (0.005) & 1.651 (0.004) \\
& pCCA  & 1.740 (0.105) & 1.741 (0.104) & 1.737 (0.103) & 1.734 (0.107) & 1.722 (0.109) \\
& MAVE  & 1.578 (0.221) & 1.317 (0.203) & 1.265 (0.184) & 1.205 (0.152) & 1.175 (0.118) \\
\midrule
\multirow{4}{*}{CSDR} 
& Oracle & 1.019 (0.077) & 0.230 (0.287) & 0.062 (0.022) & 0.025 (0.009) & 0.013 (0.004) \\
& RA     & 1.324 (0.190) & 0.982 (0.127) & 0.718 (0.373) & 0.143 (0.034) & 0.114 (0.024) \\
& DR     & 1.456 (0.228) & 1.000 (0.103) & 0.810 (0.329) & 0.156 (0.046) & 0.122 (0.025) \\
& PO     & 1.448 (0.232) & 1.027 (0.146) & 0.813 (0.349) & 0.256 (0.250) & 0.201 (0.216) \\
\bottomrule
\end{tabular}
}
\end{table}

Structural dimension estimation was also more challenging under nonsparse $\beta$ (Table~\ref{tab:dim_selection_nonsparse}).  The RA, DR, and PO variants each incorrectly select $d=1$  over $80\%$ of the time at $n=500$, and the Oracle correctly identifies $d_0=2$ in only $88\%$ of simulations (compared to $100\%$ in the sparse $\beta$ case). However, all causal SDR methods eventually recover the correct dimension with increasing $n$, consistent with the theoretical results. Under the nonsparse $\beta$, MAVE also eventually selects the correct dimension.

% Dimension selection
\begin{table}[!htbp]
\centering
\caption{Dimension selection proportions under nonsparse dimension reduction. The true dimension ($d_0 = 2$) is highlighted in bold.}
\label{tab:dim_selection_nonsparse}
\begin{tabular}{ll c c >{\bfseries}c c c c c}
\toprule
\cmidrule(lr){3-9}
Sample Size & Method & $d=0$ & $d=1$ & $d=2$ & $d=3$ & $d=4$ & $d=5$ & $d=6$ \\ 
\midrule
\multirow{5}{*}{$n = 500$} 
& MAVE   & 0.00 & 0.24 & 0.66 & 0.09 & 0.01 & 0.00 & 0.00 \\
& Oracle & 0.00 & 0.12 & 0.88 & 0.00 & 0.00 & 0.00 & 0.00 \\
& RA     & 0.00 & 0.92 & 0.08 & 0.00 & 0.00 & 0.00 & 0.00 \\
& DR     & 0.00 & 0.95 & 0.05 & 0.00 & 0.00 & 0.00 & 0.00 \\
& PO     & 0.01 & 0.84 & 0.14 & 0.01 & 0.00 & 0.00 & 0.00 \\
\midrule
\multirow{5}{*}{$n = 5000$} 
& MAVE   & 0.00 & 0.00 & 1.00 & 0.00 & 0.00 & 0.00 & 0.00 \\
& Oracle & 0.00 & 0.00 & 1.00 & 0.00 & 0.00 & 0.00 & 0.00 \\
& RA     & 0.00 & 0.00 & 1.00 & 0.00 & 0.00 & 0.00 & 0.00 \\
& DR     & 0.00 & 0.00 & 1.00 & 0.00 & 0.00 & 0.00 & 0.00 \\
& PO     & 0.00 & 0.01 & 0.98 & 0.00 & 0.00 & 0.00 & 0.00 \\
\bottomrule
\end{tabular}
\end{table}

While MAVE correctly selects $d=2$, the SIS values in Figure~\ref{fig:sim_SIS_nonsparse} show that it does not correctly recover the variable importance structure. It places most weight on $X_2$ and $X_3$, rather than distributing weight equally across $X_1$–$X_4$ and $X_7$–$X_{10}$. Scaled PCA and pCCA primarily emphasize $X_1$ and $X_2$. In contrast, the causal SDR methods accurately capture the true contribution of the original variables to $\CCMS$.

% ERS estimation harder
ERS estimation errors are also higher in this setting (Table~\ref{tab:ers_rmse_nonsparse}). The pCCA method has especially high errors, consistent with its poor estimation of $\CCMS$. The overall ranking mirrors the main simulation. Oracle performs best, followed by RA, DR, and PO, with Full $X$ and then baseline methods trailing.

% SIS Plot (?)
\begin{figure}[!t]
    \centering
    \includegraphics[width=1\textwidth]{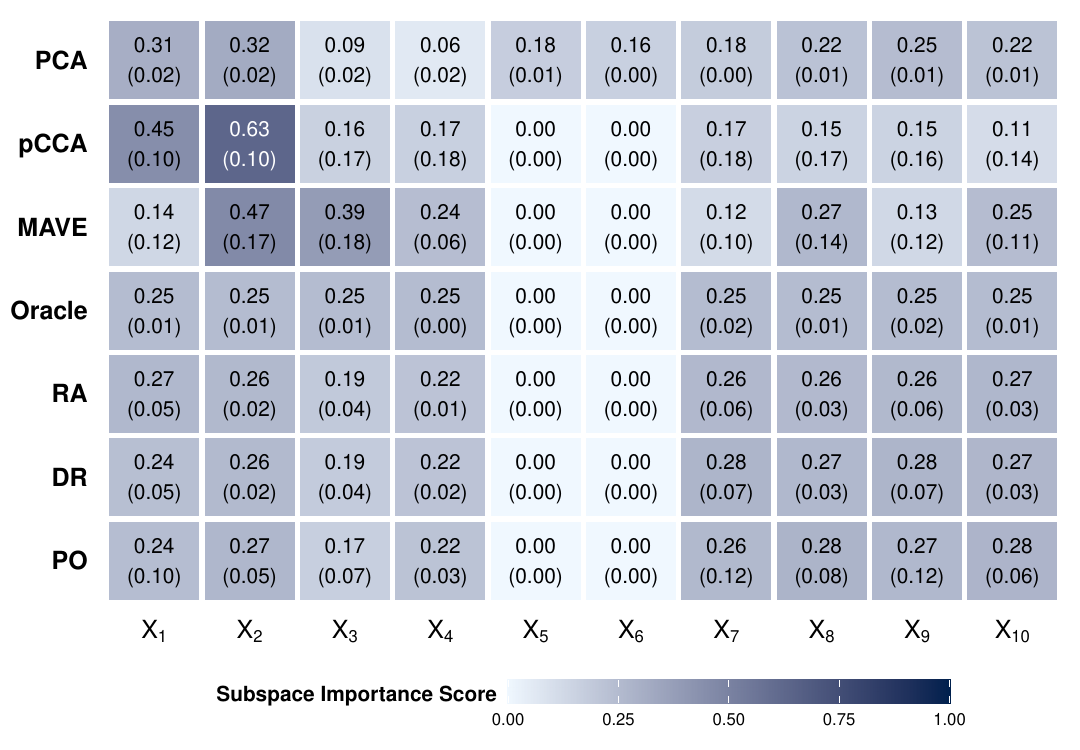}
    \caption{Subspace importance scores (SIS) for each method under a working dimension of $d=2$ at $n=1000$ under nonsparse dimension reduction.}
    \label{fig:sim_SIS_nonsparse}
\end{figure}

Although the improvements made in nuisance function estimation (moving from SuperLearner to the neural network) reduced error in estimating the ERS, they were still insufficient for valid inference. In particular, pointwise 95\% coverage was only attained for a subset of the 100 evaluation points. Undercoverage was most common for points near the edge of the observed density, possibly due to near positivity issues (Figure~\ref{fig:sim_pointwise_coverage_nonsparse}). More accurate nuisance estimation would likely improve inference but we do not pursue further model tuning and refinement here.

% RMSE Table (?)
\begin{table}[!bp]
\centering
\caption{Mean (SD) RMSE for ERS estimation across simulation runs under nonsparse dimension reduction.}
\label{tab:ers_rmse_nonsparse}
\resizebox{\textwidth}{!}{
\begin{tabular}{llccccc}
\toprule
Category & Method & $n = 100$ & $n = 500$ & $n = 1000$ & $n = 2500$ & $n = 5000$ \\
\midrule
\multirow{3}{*}{Baseline}
& pCCA     & 33.71 (20.48) & 27.55 (12.80) & 27.75 (12.51) & 28.10 (14.05) & 30.09 (16.57) \\
& MAVE     & 16.77 (7.88)  & 7.00 (2.49)   & 5.34 (1.43)   & 4.00 (0.80)   & 3.51 (0.54) \\
& Full $X$ & 15.18 (6.37)  & 6.77 (1.28)   & 4.56 (0.76)   & 3.44 (0.49)   & 3.06 (0.36) \\
\midrule
\multirow{4}{*}{CSDR}
& Oracle & 8.11 (5.61) & 2.75 (1.06) & 1.97 (0.79) & 1.41 (0.48) & 1.17 (0.34) \\
& RA     & 12.90 (5.24) & 5.49 (1.14) & 3.57 (0.97) & 1.71 (0.48) & 1.37 (0.35) \\
& DR     & 14.35 (6.57) & 5.45 (1.17) & 3.70 (0.96) & 1.72 (0.47) & 1.38 (0.40) \\
& PO     & 14.54 (8.81) & 5.61 (1.29) & 3.74 (1.03) & 1.92 (0.71) & 1.52 (0.55) \\
\bottomrule
\end{tabular}
}
\end{table}

% Coverage Plot (?)
\begin{figure}[!tbp]
    \centering
    \includegraphics[width=.99\textwidth]{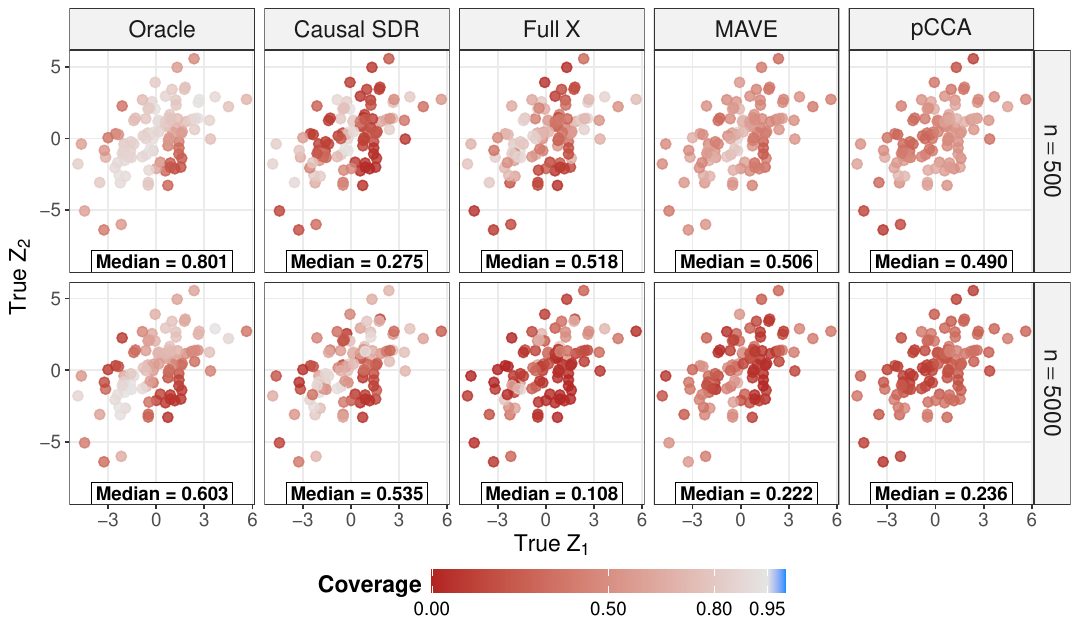}
    \caption{Pointwise 95\% coverage of ERS estimates across 100 evaluation points under nonsparse dimension reduction.}
    \label{fig:sim_pointwise_coverage_nonsparse}
\end{figure}

This extra simulation highlights the increased difficulty of causal SDR in the nonsparse $\beta$ setting. Although the structural dimension remains low at $d_0=2$, a larger set of the original variables in $X$ contributes to the ERS, making the first-stage estimation  of $\mu(X)$ more difficult. As a result, more flexible learners were required for both first- and second-stage estimation, yet errors were higher than in the main experiment when estimating $\CCMS$ and the ERS. This underscores the importance of sufficiently fast convergence rates for nuisance estimation that can depend on the underlying DGP, not only for consistency but also for achieving valid inference. At the same time, this experiment demonstrates the value of causal SDR in diverse scenarios. Despite the increased error and undercoverage for ERS estimation, the causal SDR variants still improves over the baseline methods. These results reinforce that causal SDR is not only useful for interpretability of effects, but can also deliver practical benefits for estimation and inference.

\clearpage

\section{Additional Real Data Analysis Results on ATL-AA}\label{sec:additional_RDA}
We compare the results of our ATL-AA analysis with two widely used approaches for environmental mixtures: the single-index \textit{qgcomp} 
\citepSM{keilQuantileBasedGComputationApproach2020SM} and the full exposure–response model \textit{BKMR} \citepSM{bobbBayesianKernelMachine2015SM}. These methods operate under fundamentally different modeling frameworks and produce distinct estimands, making them difficult to directly compare to CSDR. For this reason, we did not include them in the simulation study. We do present them here as part of the real data analysis to assess whether the general conclusions are consistent across methods. Code to run the study is available at \url{https://github.com/tXiao95/causalpca/blob/main/code/ATL-AA/analysis.R}.

\subsection{Quantile $g$-computation (qgcomp)}
The estimated qgcomp mixture effect was $\psi_1 = -26.3$ (95\% CI: $-78.1, 25.4$), corresponding to a 26.3-gram decrease in birthweight for a one-quartile simultaneous increase in all four PFAS exposures, albeit with substantial uncertainty. Our estimated ERS in Figure~\ref{fig:PFAS_ERS} also showed a general downward trend with increasing PFAS exposure. The qgcomp weights indicate that PFOS contributes 57.6\% of the negative association, followed by PFNA (35.4\%) and PFOA (7.0\%), while PFHxS is the only exposure with a positive contribution. This aligns with the CSDR results, where $\beta = (0.77, 0.26, 0.14, -0.56)$ and the corresponding subspace importance scores are $\lambda=(0.60, 0.07, 0.02, 0.31)$, again highlighting PFOS as the primary driver of the negative trend and PFHxS as the only exposure associated with increases in birthweight.

\begin{figure}[!tbp]
    \centering
    \includegraphics[width=1\textwidth]{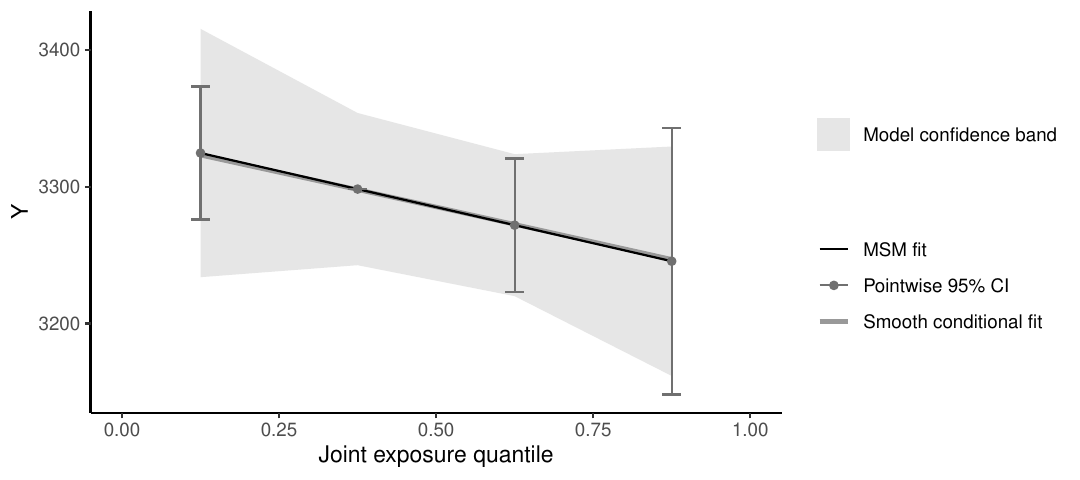}
    \caption{Dose-response curve estimated by qgcomp with $B=500$ bootstrap iterations.}
    \label{fig:ATL_AA_qgcomp}
\end{figure}

We also ran qgcomp with $B = 500$ bootstrap iterations and cubic basis functions to estimate a nonlinear dose–response curve (Figure~\ref{fig:ATL_AA_qgcomp}). However, the resulting curve was approximately linear. Overall, qgcomp yields conclusions consistent with those from CSDR. Because qgcomp does not extrapolate beyond the range of observed quantiles, its estimated dose–response is most comparable to the central region of the ERS shown in Figure~\ref{fig:PFAS_ERS}.

\subsection{Bayesian kernel machine regression (BKMR)}

We ran BKMR with 10,000 MCMC iterations and variable selection. The univariate response functions in Figure~\ref{fig:ATL_AA_bkmr_univariate} are broadly consistent with CSDR and qgcomp: increasing PFOS, PFOA, and PFNA is associated with decreases in birthweight, whereas PFHxS shows a positive association. However, the variable selection procedure yielded very low posterior inclusion probabilities (PIPs): 0.001 (PFOS), 0.029 (PFOA), 0.023 (PFNA), and 0.005 (PFHxS), suggesting weak evidence for the importance of any individual exposure. In this case, relative variable importance is difficult to assess, since the BKMR model did not identify a strong effect for any single component of the mixture. The weak effect was also present in the qgcomp results, where the estimated $\psi_1$ exhibited substantial uncertainty. While we did not conduct a formal hypothesis test for CSDR, the central region of the ERS (where the observed data was most concentrated) was relatively flat with wide confidence intervals.

\begin{figure}[!bp]
    \centering
    \includegraphics[width=.95\textwidth]{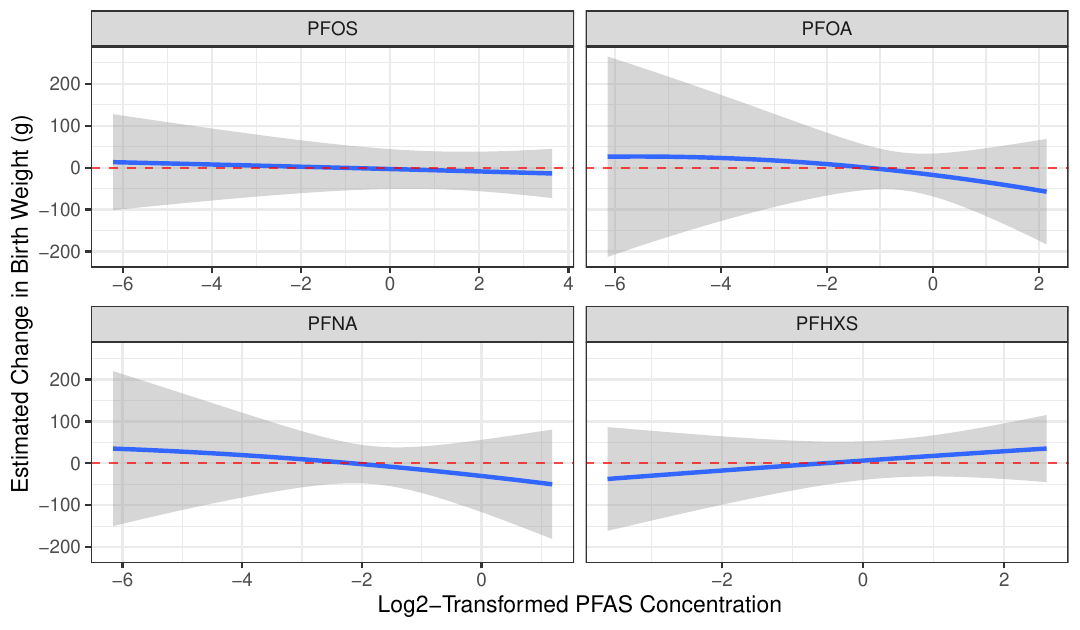}
    \caption{Univariate change in response after setting all other PFAS exposures to their median value.}
    \label{fig:ATL_AA_bkmr_univariate}
\end{figure}

% PIP
% Overall effect
Figure~\ref{fig:ATL_AA_bkmr_diff} visualizes the overall effect estimated by BKMR by comparing birthweight when all exposures are set to a given quantile versus the median. The results show an approximately linear decrease in birthweight with increasing PFAS levels, consistent with the dose–response patterns observed from both CSDR and qgcomp.

\begin{figure}[!bp]
    \centering
    \includegraphics[width=.95\textwidth]{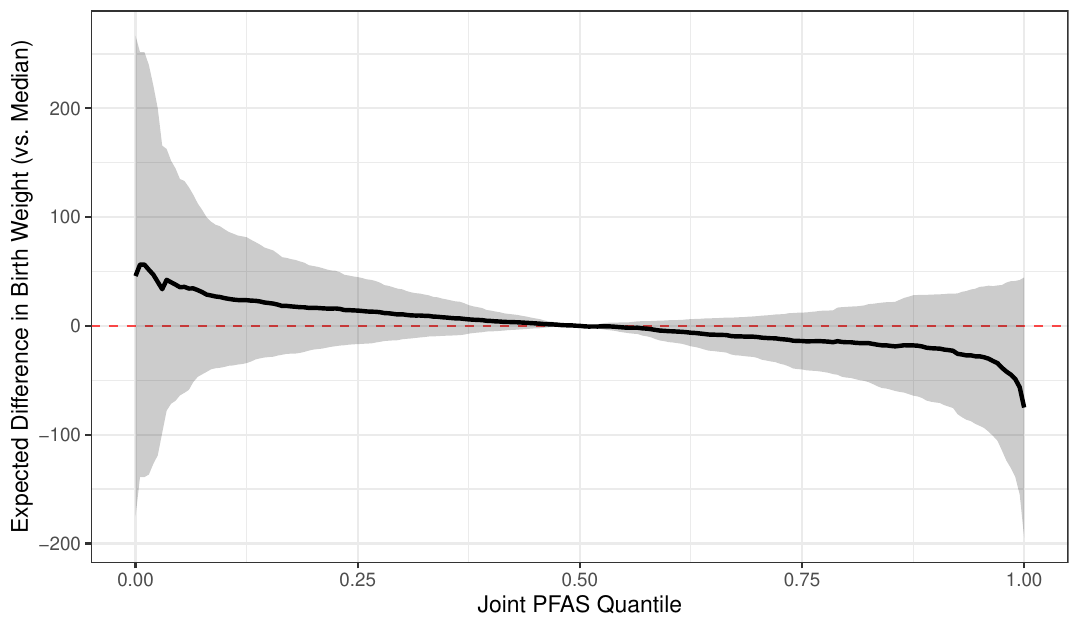}
    \caption{Difference in birthweight when all PFAS exposures are set to a common quantile vs. the median, estimated from posterior draws of BKMR.}
    \label{fig:ATL_AA_bkmr_diff}
\end{figure}

\clearpage

\bibliographystyleSM{apalike}
\bibliographySM{references_SM}

\end{document}